\documentclass[11pt,a4paper]{article}
\usepackage{graphicx}
\usepackage{natbib}
\usepackage{amsmath,amsfonts,amssymb,amsbsy,amsthm,mathtools}
\usepackage{bbm,bm, centernot}
\usepackage{color}
\usepackage{subfig,float}
\usepackage{booktabs}
\usepackage{url}
\usepackage[a4paper]{geometry}
\usepackage{a4wide}
\allowdisplaybreaks
\numberwithin{equation}{section}
\usepackage{lmodern}
\usepackage{paralist}

\newtheorem{thm}{Theorem}[section]
\newtheorem{cor}[thm]{Corollary}

\theoremstyle{definition}
\newtheorem{defn}{Definition}[section]
\newtheorem*{defn*}{Definition}

\newtheorem{exa}[defn]{Example}

\DeclareMathOperator*{\argmin}{arg\,min}
\newcommand{\RR}{\mathbb{R}}
\newcommand{\vc}[1]{\bm{#1}}

\newcommand{\PP}{\mathbb{P}}

\newcommand{\dto}{\overset{\textnormal{d}}{\longrightarrow}}
\newcommand{\pto}{\overset{\textnormal{p}}{\longrightarrow}}

\newcommand{\htheta}{\widehat{\vc{\theta}}_{n,k}}

\newcommand{\Normal}{\mathcal{N}}

\newcommand{\vtheta}{\vc{\theta}}
\newcommand{\cl}[1]{\textnormal{cl} (#1)}

\title{Hypothesis testing for tail dependence parameters on the boundary of the parameter space}
\author{Anna Kiriliouk \\ 
   \small Universit\'{e} de Namur\\
    \small Facult\'e des sciences \'economiques, sociales et de gestion \\
    \small Rempart de la vierge~8, B-5000 Namur, Belgium.\\
    \small E-mail: anna.kiriliouk@unamur.be}

\begin{document}
\maketitle

%\vspace{-0.95cm}

\begin{abstract}
Modelling multivariate tail dependence is one of the key challenges in extreme-value theory.
Multivariate extremes are usually characterized using parametric models, some of which have simpler submodels at the boundary of their parameter space.
Hypothesis tests are proposed for tail dependence parameters that, under the null hypothesis, are on the boundary of the alternative hypothesis. The asymptotic distribution of the weighted least squares estimator (Einmahl, Kiriliouk and Segers, Extremes 21, pages 205--233, 2018) is given when the true parameter vector is on the boundary of the parameter space, and two test statistics are proposed. 
The performance of these test statistics is evaluated for the Brown--Resnick model and the max-linear model. In particular, simulations show that it is possible to recover the optimal number of factors for a max-linear model. Finally, the methods are applied to characterize the dependence structure of two major stock market indices, the DAX and the CAC40.
\end{abstract}

\noindent%

\section{Introduction}
% \label{sec:introduction}
Extreme-value theory is the branch of statistics concerned with the characterization of extreme events. These occur in a large variety of fields, such as hydrology, meteorology, finance and insurance, but also in settings like human life span or athletic records. For some examples, see \cite{einmahl2008nr2}, \cite{towler2010}, \cite{chavez2016}, \cite{thomas2016} and \cite{rootzen2017}. 
In the univariate case, the limiting distribution of suitably normalized block maxima or threshold exceedances can be characterized entirely using the generalized extreme-value \citep{fisher1928, gnedenko1943} and the generalized Pareto distribution \citep{balkema1974,pickands1975} respectively.
However, many extreme events are inherently multivariate, and an important challenge is to model the tail dependence between two or more random variables of interest. If dependence disappears as the variables take on more and more extreme values, we say that they are asymptotically independent. Testing for asymptotic independence and modelling asymptotically independent data has been done in \cite{ledford1996}, \cite{draisma2004}, \cite{husler2009}, \cite{wadsworth2016nr2} and \cite{guillou2018}, among others. In the following, we consider the framework of asymptotic dependence.

The family of limiting distributions for multivariate maxima or threshold exceedances is infinite-dimensional, so that realistic and computationally feasible parametric models need to be proposed. Examples include \cite{gumbel1960}, \cite{smith1990}, \cite{tawn1990} or \cite{kabluchko2009}. Many different approaches of estimating tail dependence parameters exist. Fitting a multivariate extreme-value distribution 
to componentwise block maxima is done in \cite{padoan2010}, \cite{davison2012}, \cite{castruccio2016} or \cite{dombry2017nr2}, along others.
Multivariate threshold exceedances, which are the focus of this work, can be tackled using the stable tail dependence function \citep{dreeshuang1998} or the multivariate generalized Pareto distribution \citep{rootzen2006}. Common estimation methods include maximum likelihood \citep{wadsworth2014,kiriliouk2018,defondeville2018} 
and minimum distance estimation 
\citep{einmahl2012,einmahl2016,einmahl2016nr2}.

Some popular parametric models have simpler submodels at the boundary of the parameter space. The main goal of this paper is to propose hypothesis tests for tail dependence parameters that, under the null hypothesis, are on the boundary of the alternative hypothesis.  A first example is the Brown--Resnick model \citep{kabluchko2009}, popular in spatial extremes, which reduces to the simpler Smith model \citep{smith1990} when the shape parameter of the former attains its upper bound. A second example is the max-linear model, where each component of a vector can be interpreted as the maximum shock among a set of independent factors. Recent work on max-linear models is numerous.
In \citet{gissibl2015} and \citet{gissibl2017}, a recursive max-linear model on a directed acyclic graph is considered. In \cite{cui2016} and \cite{zhao2018}, max-linear models are used to model a latent factor structure for financial returns. Finally, the Marshall--Olkin model \citep{embrechts2003,segers2012} is a submodel that is frequently used in practice, see for instance \cite{burtschell2009}, \citet{su2017} or \citet{brigo}.

Because of their non-differentiability, max-linear models cannot be estimated by standard likelihood methods and hence we use the weighted least squares estimator of \cite{einmahl2016nr2}, which is centred around the stable tail dependence function. The hypothesis testing framework described in that paper is valid for parameters in the interior of the parameter space, and hence it cannot be used to test whether one or more factor loadings are equal to zero. The hypothesis tests proposed in this paper can be applied to decide how many factors are necessary to model a certain dataset or to test if a more specific model (eg, Marshall--Olkin) is sufficient, bypassing the limitation that the number of factors has to be chosen upfront.

The paper is organized as follows. Section~\ref{sec:estimator} presents background on multivariate extremes and introduces the weighted least squares estimator. Section~\ref{sec:main} establishes asymptotic normality for the weighted least squares estimator when the true parameter vector lies on the boundary of the parameter space. Using results from \cite{andrews1999}, \cite{andrews2001} and \cite{andrews2002}, a deviance- and a Wald-type test statistic are proposed, whose asymptotic distributions are easily computable when the dimension of the parameter vector on the boundary is moderate. 
In Section~\ref{sec:gml}, we perform simulations to evaluate the performance of the test statistics for the Brown--Resnick and the max-linear models. Finally, Section~\ref{sec:application}
illustrates our methods on two major stock market indices, the DAX and the CAC40. Technical definitions and proofs are deferred to the appendix.

\section{Background}
\label{sec:estimator}
\subsection{Multivariate extreme-value theory}
\label{sec:EVT}

Let $\vc{X}_i = (X_{i1},\ldots,X_{id})$, $i \in \{1,\ldots,n\}$, be random vectors in $\RR^d$ with cumulative distribution function $F$ and marginal cumulative distribution functions $F_1,\ldots,F_d$. Let $\vc{M}_n := (M_{n,1},\ldots,M_{nd})$ with $M_{nj} := \max (X_{1j},\ldots,X_{nj})$ for $j=1,\ldots,d$. We say that $F$ is in the \emph{max-domain of attraction} of an \emph{extreme-value distribution} $G$ if there exist sequences of normalizing constants $\vc{a}_n = (a_{n1},\ldots,a_{nd}) >0$ and $\vc{b}_n = (b_{n1},\ldots,b_{nd}) \in \RR^d$ such that
\begin{equation}\label{eq:ell1}
\PP \bigg[ \frac{\vc{M}_{n} - \vc{b}_{n}}{\vc{a}_{n}} \leq \vc{x} \bigg] = F^n (\vc{a}_n \vc{x} + \vc{b}_n) \dto G(\vc{x}), \qquad \text{as } n \rightarrow \infty.
\end{equation}
The margins, $G_1, \ldots, G_d$, of $G$ are \emph{univariate extreme-value distributions},
\begin{equation*}
G_j (x_j) = \exp \left\{ - \left(  1 + \gamma \frac{x_j - \mu_j}{\sigma_j} \right)_+^{-1/\gamma_j}\right\}, \qquad \sigma_j > 0, \gamma_j, \mu_j \in \mathbb{R},
\end{equation*} 
where $x_+ := \max(x,0)$. The function $G$ is determined by
\begin{equation*}
G(\vc{x}) = \exp{\{- \ell(- \log G_1 (x_1), \ldots , - \log G_d (x_d))\}},
\end{equation*}
where $\ell: [0,\infty)^d \rightarrow [0,\infty)$ is called the \emph{(stable) tail dependence function},
\begin{equation}\label{eq:ell2}
  \ell(\vc{x}):= 
  \lim_{t \downarrow 0} t^{-1} \, 
  \PP[  1 - F_1(X_{11}) \leq t x_1 \text{ or } \ldots \text{ or } 1 - F_d (X_{1d}) \leq t x_d].
\end{equation}
The cumulative distribution function $F$ is in the max-domain of attraction of a $d$-variate extreme value distribution $G$ if and only if the limit in \eqref{eq:ell2} exists \emph{and} the marginal distributions in \eqref{eq:ell1} converge to univariate extreme-value distributions. In what follows, we only assume existence of \eqref{eq:ell2}, which concerns the dependence structure of $F$, but not the marginal distributions $F_1,\ldots,F_d$.

Because the class of stable tail dependence functions is infinite-dimensional, one usually considers parametric models for $\ell$. Henceforth we assume that $\ell$ belongs to a parametric family $\{\ell(\cdot \, ; \vtheta) : \vtheta \in \Theta \}$ with $\Theta \subset \RR^p$. Some examples can be found below, see also \citet{dehaan2006}, \citet{falk2010}, \citet{segers2012} and references therein. 

\begin{exa}
The $d$-dimensional \emph{logistic model} \citep{gumbel1960} has stable tail dependence function
\begin{equation*}\label{eq:logistic}
  \ell(\vc{x} ; \theta)
  = \bigl( x_1^{1/\theta} + \cdots + x_d^{1/\theta} \bigr)^\theta,
  \qquad \theta \in (0,1].
\end{equation*} 
If $\theta = 1$, the variables are (asymptotically) independent, while $\theta \downarrow 0$ corresponds to the situation of complete dependence.
\end{exa}

\begin{exa}
The $d$-dimensional \emph{Brown--Resnick model} defined on spatial locations $\vc{s}_1,\ldots,\vc{s}_d \in \RR^2$  has stable tail dependence function
\begin{equation*}
\ell  (\vc{x} ; \vtheta) =  \sum_{j=1}^d  x_j \Phi_{d-1} ( \eta^{(j)} (1/\vc{x}) ; \Upsilon^{(j)} ) ,
\end{equation*}
where $\eta^{(j)} (\vc{x})  = (\eta_1^{(j)} (x_1,x_j), \ldots, \eta_{j-1}^{(j)} (x_{j-1},x_j), \eta_{j+1}^{(j)} (x_{j+1},x_j) , \ldots , \eta_d^{(j)} (x_d,x_j) ) \,\in \mathbb{R}^{d-1}$,
\begin{align*}
\eta_l^{(j)} (x_l,x_j)  = \sqrt{\frac{\gamma (\vc{s}_{j} - \vc{s}_{l})}{2}} + \frac{\log{(x_l/x_j)}}{\sqrt{2 \gamma (\vc{s}_{j} - \vc{s}_{l})}} \,\, \in \mathbb{R},
\end{align*}
$\gamma (\vc{s}) = (||\vc{s}|| / \rho)^{\alpha}$ and $\Upsilon^{(l)} \in \mathbb{R}^{(d-1) \times (d-1)}$ is the correlation matrix with entries
\begin{equation*}
\Upsilon^{(j)}_{lk} = \frac{\gamma (\vc{s}_{j} - \vc{s}_{l}) + \gamma(\vc{s}_{j} - \vc{s}_{k}) - \gamma(\vc{s}_{l} - \vc{s}_{k})}{2 \sqrt{\gamma(\vc{s}_{j}-\vc{s}_{l}) \gamma(\vc{s}_{j} - \vc{s}_{k})}}, \qquad l,k=1,\ldots,d; \, l,k \neq j,
\end{equation*}
\citep{kabluchko2009,huser2013}. The parameter vector is $\vc{\theta} = (\rho,\alpha) \in (0,\infty) \times (0,2]$. The \emph{Smith model} \citep{smith1989} is obtained when $\alpha = 2$.
\end{exa}

\begin{exa}
The \emph{max-linear model} with $r$ factors has stable tail dependence function
\begin{equation}\label{eq:mlstdf}
  \ell(\vc{x} ; \vc{\theta}) =   \sum_{t=1}^r \max_{j=1,\ldots,d}{b_{jt} x_j},  \qquad \vc{x} \in [0,\infty)^d,
\end{equation}
where the factor loadings $b_{jt}$ are non-negative constants such that $\sum_{t=1}^r b_{jt} = 1$ for every $j \in \{1,\ldots,d\}$ and all column sums of the $d \times r$ matrix $B := (b_{jt})_{j,t}$ are positive \citep{einmahl2012}.
Since the rows of $B$ sum up to one, the parameter matrix has only $d \times (r-1)$ free elements. Rearranging the columns of $B$ will not change the value of the stable tail dependence function. For identification purposes, we define the parameter vector $\vc{\theta}$ by stacking the columns of $B$ in decreasing order of their sums, leaving out the column with the lowest sum. 
An example of a random vector $\vc{Z}= (Z_1, \ldots, Z_d)$ that has stable tail dependence function \eqref{eq:mlstdf} is 
\begin{equation*} \label{eq:maxlinear}
Z_j = \max_{t=1,\ldots,r}{b_{jt} S_t}, \qquad \text{for } j \in \{1,\ldots,d\},
\end{equation*}
where $S_1,\ldots,S_r$ are independent unit Fr\'{e}chet variables, $\PP [S_j \leq x] = \exp(-1/x)$ for $x > 0$ and $j \in \{1,\ldots,d\}$. 
\end{exa}

\begin{exa}
Let $I_1,\ldots,I_d$ denote random (possibly dependent) Bernoulli random variables with $p_j = \PP[I_j = 1] \in (0,1]$ for $j \in \{1,\ldots,d\}$. Let, for $J \subset \{1,\ldots,d\}$, $p(J) = \PP[ \{j = 1,\ldots,d : I_j = 1\} = J]$, so that $(p(J))_{J \subset \{1,\ldots,d\}}$ is a probability distribution.
The \emph{multivariate Marshall--Olkin model} \citep{embrechts2003, segers2012} has stable tail dependence function 
\begin{equation*}
\ell(\vc{x} ; \vtheta) = \sum_{\varnothing \neq J \subset \{1,\ldots,d\}} p(J) \max_{j \in J} \left( \frac{x_j}{p_j} \right).
\end{equation*} 
In this model, any subset of components of the vector is assigned a shock that influences all components of that subset. An interpretation in terms of credit risk modelling can be found in \citet{embrechts2003}.

The Marshall--Olkin model is obtained as a special case of the max-linear model \eqref{eq:mlstdf} by setting $b_{jt} = \left\{ p(J_t)/p_j \right\} \mathbbm{1} ( j \in J_t)$, where $J_t \in \mathcal{P}(\{1,\ldots,d\})$, $J_t \neq \varnothing$, and  $\mathcal{P}(A)$ denotes the power set of $A$.
%It is easy to check that $\sum_{t=1}^r b_{jt} = 1$ for every $j \in \{1,\ldots,d\}$. 
%\ell(\vc{x} ; \vtheta) = \sum_{t=1}^r \max_{j=1,\ldots,d} \left( \frac{p(J_t)}{p_j} \mathbbm{1} ( j \in J_t) x_j  \right)
The Marshall--Olkin model iswell-known in $d = 2$, where
\begin{equation*}
B = \begin{pmatrix} b_{11} & 0 & 1 - b_{11}  \\ 
0 & b_{22} & 1 - b_{22}  \end{pmatrix}  
= \begin{pmatrix} \frac{P[I_1 = 1, I_2 = 0]}{P[I_1 = 1]} & 0 & \frac{P[I_1 = 1, I_2 = 1]}{P[I_1 = 1]}  \\ 
0 & \frac{P[I_1 = 0, I_2 = 1]}{P[I_2 = 1]} & \frac{P[I_1 = 1, I_2 = 1]}{P[I_2 = 1]}  \end{pmatrix}.
\end{equation*} 
\end{exa}

\subsection{Estimation of the stable tail dependence function}\label{sec:wls}
\subsubsection{Nonparametric estimation of the stable tail dependence function}\label{sec:nonpar}
Let $k = k_n \in (0,n]$ be such that $k \rightarrow \infty$ and $k/n \rightarrow 0$ as $n \rightarrow \infty$.
A straightforward nonparametric estimator of $\ell$ is obtained by replacing $\mathbb{P}$ and $F_1, \ldots, F_d$ in~\eqref{eq:ell2} by the (modified) empirical distribution functions and replacing $t$ by $k/n$, yielding
\begin{equation*}\label{eq:ellclassic}
  \widetilde{\ell}_{n,k} (x) 
  := 
  \frac{1}{k} \sum_{i=1}^n 
  \mathbbm{1} 
  \left\{
    R_{i1} > n + 1/2 - kx_1  \text{ or } \ldots \text{ or }  R_{id} > n + 1/2 - kx_d
  \right\}.
\end{equation*}
Here, $R_{ij} = \sum_{t=1}^n \mathbbm{1} \left\{ X_{tj} \leq X_{ij}\right\}$ denotes the rank of $X_{ij}$ among $X_{1j},\ldots,X_{nj}$.
This estimator, the \emph{empirical tail dependence function}, was introduced in slightly different form for $d=2$  in \citet{huang1992} and studied further in \citet{dreeshuang1998}. Using $n + 1/2$ rather than $n$ allows for better finite-sample properties, that is, for a lower mean squared error \citep{einmahl2012}.

Another nonparametric estimator is the \emph{beta tail dependence function}, which was very recently proposed in \citet{kiriliouk2017}. Contrary to the empirical tail dependence function, the beta tail dependence function is a smooth estimator which leads to some improvement in its finite-sample behavior. Its name stems from the fact that it is based on the
empirical beta copula \citep{segers2017}. We define 
\begin{align*}
\widetilde{\ell}^{\beta}_{n,k} (\vc{x}) = \frac{n}{k} \left\{ 1 - \mathbb{C}_n^{\beta} \left(1 - \frac{k \vc{x}}{n} \right) \right\}, \quad \text{where} \quad
\mathbb{C}_n^{\beta} (\vc{u}) = \frac{1}{n} \sum_{i=1}^n \prod_{j=1}^d F_{n,R_{ij}} (u_j),
\end{align*}
and for $r \in \{1,\ldots,n\}$, $F_{n,r} (u) = \sum_{s=1}^r \binom{n}{s} u^s (1-u)^{n-s}$ is the cumulative distribution function of a Beta$(r, n+1-r)$ random variable.

A drawback of $\widetilde{\ell}_{n,k}$ or $\widetilde{\ell}^{\beta}_{n,k}$ might be their possibly growing bias as $k$ increases: \cite{fougeres2015nr2} show that, under suitable conditions, the estimator $\widetilde{\ell}_{n,k}$ satisfies the asymptotic expansion
\begin{equation*}
\widetilde{\ell}_{n,k} (\vc{x}) - \ell( \vc{x}) \approx k^{-1/2} Z_{\ell} (\vc{x}) + \alpha (n/k) M(\vc{x}),
\end{equation*}
where $Z_{\ell}$ is a continuous centered Gaussian process, $\alpha$ is a function such that $\lim_{x \rightarrow \infty} \alpha(x) = 0$ and $M$ is a continuous function. When $\sqrt{k} \alpha (n/k)$ tends to a non-zero constant, an asymptotic bias appears. In \cite{fougeres2015nr2} and \cite{beirlant2016}, this bias is estimated and subtracted from the estimator $\widetilde{\ell}_{n,k} (\vc{x}) $ in order to propose new bias-corrected estimators.

Finally, many other nonparametric estimators exist of the Pickands dependence function, which is the restriction of $\ell$ to the unit simplex. See, for instance, \citet{caperaa1997}, \citet{zhang2008}, \citet{Gudendorf2012},  \citet{berghaus2013}, \cite{vettori2016} and 
\citet{marcon2017}. While these can be transformed into estimators of the stable tail dependence function, they rely on stronger assumptions, i.e., they are based on data from an extreme-value distribution, and we will not consider them here.

\subsubsection{Semi-parametric estimation of the stable tail dependence function}
Nonparametric estimation forms a stepping stone for semi-parametric estimation \citep{einmahl2012,einmahl2016,einmahl2016nr2}. Here, we focus on the method proposed in \citet{einmahl2016nr2}. 
Let $\widehat{\ell}_{n,k}$ denote any initial estimator of $\ell$ based on $\vc{X}_1, \ldots, \vc{X}_n$. 
Let $\vc{c}_1,\ldots,\vc{c}_q \in [0,\infty)^d$, with $\vc{c}_m = (c_{m1},\ldots,c_{md})$ for $m = 1,\ldots,q$,  be $q$ points in which we will evaluate $\ell$ and $\widehat{\ell}_{n,k}$. Consider for $\vtheta \in \Theta$ the $q \times 1$ column vectors
\begin{equation}\label{eq:Lmap}
  \widehat{\vc{L}}_{n,k} 
  := 
  \bigl( \widehat{\ell}_{n,k} (\vc{c}_m) \bigr)_{m = 1}^q, \,\,
  L( \vtheta ) 
  := 
  \bigl( \ell (\vc{c}_m ; \vc{\theta}) \bigr)_{m=1}^q, \,\,
  D_{n,k} (\vtheta) 
  := 
  \widehat{\vc{L}}_{n,k} - L(\vc{\theta}).
\end{equation}
The points $\vc{c}_1,\ldots,\vc{c}_q$ need to be chosen in such a way that the map $L : \Theta \to \RR^q$ is one-to-one, i.e., $\vtheta$ is identifiable from $\ell(\vc{c}_1;\vtheta), \ldots, \ell(\vc{c}_q;\vtheta)$. In particular, we will need to assume that $q \ge p$.

For $\vtheta \in \Theta$, let $\Omega(\vtheta)$ be a symmetric, positive definite $q\times q$ matrix and define
%:= \norm{ D_{n,k} (\theta) }^2_{\Omega(\theta)}
\begin{equation}\label{eq:fnk}
  f_{n,k} (\vtheta) := 
  D_{n,k}^T (\vtheta) \, \Omega (\vtheta) \, D_{n,k} (\vtheta).
\end{equation}
The \emph{continuous updating weighted least squares estimator} for $\vtheta_0$ is defined as
\begin{equation}\label{eq:estimator}
  \htheta 
  := 
  \argmin_{\vtheta \in \Theta} f_{n,k} (\vtheta) 
  = 
  \argmin_{\vtheta \in \Theta} \left\{ D_{n,k} (\vtheta)^T \, \Omega (\vtheta) \, 
  D_{n,k} (\vtheta) \right\}.
\end{equation} 
In Section~\ref{sec:simulation}, we will study the performance of this estimator for $\Omega (\vtheta) = I_q$, the $q \times q$ identity matrix. Expression~\eqref{eq:estimator} then simplifies to
\begin{equation*}
\htheta = \argmin_{\vtheta \in \Theta} \sum_{m=1}^q
  \big(   \widehat{\ell}_{n,k} (\vc{c}_m) - \ell (\vc{c}_m; \vtheta) \big)^2.
\end{equation*}

\citet{einmahl2016nr2} show the consistency and asymptotic normality of $\htheta$ under the assumption that the true parameter vector $\vtheta_0$ is in the interior of $\Theta$. In the next section we give the asymptotic distribution of $\htheta$ without this restriction.

\section{Inference on tail dependence parameters on the boundary of the parameter space}\label{sec:main}
The asymptotic results presented in Section \ref{sec:asymptotic} build on \citet{andrews1999}, who established the asymptotic distribution of a general extremum estimator when one or more parameters lie on the boundary of the parameter space. 
The methodology to obtain the asymptotic distribution of $\sqrt{k} (\htheta - \vtheta_0)$ can be summarized as follows: first, one shows that as $n \rightarrow \infty$, the criterion function $f_{n,k} (\vtheta)$ in \eqref{eq:fnk} is equal to a quadratic function $q_{n,k} (\sqrt{k} (\vtheta - \vtheta_0))$ plus a term that does not depend on $\vtheta$. If $\htheta$ is consistent, then its asymptotic distribution is only affected by the part of $\Theta$ close to $\vtheta_0$; equivalently, we are only concerned with the shifted parameter space $\Theta - \vtheta_0$ near the origin. If $\Theta - \vtheta_0$ can be approximated near the origin by a convex cone $\Lambda$, one can show that minimizing $f_{n,k} (\vtheta)$ over $\vtheta \in \Theta$ is asymptotically equivalent to minimizing $q_{n,k} (\vc{\lambda})$ over $\vc{\lambda} \in \Lambda$. Finally, $\sqrt{k} (\htheta - \vtheta_0)$ converges in distribution to the argument minimizing the limit of $q_{n,k} (\vc{\lambda})$ as $n \rightarrow \infty$.

In Section \ref{sec:simplify}, we show how a closed-form expression can be obtained for the asymptotic distribution of $\sqrt{k} (\htheta - \vtheta_0)$. For $\vc{\beta} \subset \vtheta$, we show in Section \ref{sec:test} how to test $H_0: \vc{\beta} = \vc{\beta}_*$ against $H_1: \vc{\beta} \neq \vc{\beta}_*$ when, under the null hypothesis, $\vc{\beta}^*$ is on the boundary of the alternative hypothesis.

We set up some notation first. Let $B_{\varepsilon} (\vtheta)$ denote an open ball centered at $\vtheta$ with radius $\varepsilon$ and let $C_{\varepsilon} (\vtheta)$ denote an open cube centered at $\vtheta$ with sides of length $2 \varepsilon$. 
Let $\cl{\Theta}$ denote the closure of $\Theta$. 
A set $\Gamma \subset \RR^p$ is said to be \emph{locally equal} to a set $\Lambda \subset \RR^p$ if $\Gamma \cap B_{\varepsilon} (\vc{0}) = \Lambda \cap B_{\varepsilon} (\vc{0})$ for some $\varepsilon > 0$. Finally, a set $\Lambda \subset \RR^p$ is a \emph{cone} if $\vc{\lambda} \in \Lambda$ implies $a \vc{\lambda} \in \Lambda$ for all $a > 0$.

\subsection{Estimation, consistency and asymptotic normality}\label{sec:asymptotic}
When $\vtheta_0$ is on the boundary of $\Theta$, the map $L$ in \eqref{eq:Lmap} is not defined and thus not differentiable on a neighbourhood of $\vtheta_0$. We will need the following assumption.
\begin{description}
\item[(A1)] $\Theta$ includes a set $\Theta^+$ such that $\Theta^+ - \vtheta_0$ equals the intersection of a union of orthants and an open cube $C_{\varepsilon} (\vc{0})$ for some $\varepsilon > 0$. Moreover, $\Theta \cap B_{\varepsilon_1} (\vtheta_0) \subset \Theta^+$ for some $\varepsilon_1 > 0$. 
\end{description}
If $\Theta - \vtheta_0$ happens to be locally equal to a union of orthants, we can simply set $\Theta^+ = \Theta \cap C_{\varepsilon} (\vtheta_0)$. This is the case for the models considered in Section~\ref{sec:simulation}.
We will assume existence of the so-called \emph{left/right (l/r) partial derivatives} on $\Theta^+$; a formal definition is given in the appendix. The shape of $\Theta^+$ is such that these can always be defined.

Write $\dot{L} := (\partial / \partial \vc{\theta}) L(\vc{\theta}) \in \RR^{q \times p}$ for $\vtheta \in \Theta^+$, where $\dot{L}$ denotes the matrix of l/r partial derivatives. Let $\lambda_1(\vtheta) > 0$ denote the smallest eigenvalue of $\Omega(\vtheta)$. 

\begin{thm}[Existence, uniqueness and consistency]
\label{theorem1}
Let $\vc{c}_1,\ldots,\vc{c}_q \in [0,\infty)^d$ be $q \ge p$ points
such that the map $L: \theta \mapsto (\ell (\vc{c}_m ; \vtheta))_{m=1}^q$ is a homeomorphism. Let $\vc{\vtheta}_0 \in \cl{\Theta}$ and assume that (A1) holds, that each element of $L(\vtheta)$ has continuous l/r partial derivatives of order two on $\Theta^+$, that $\dot{L} (\vtheta_0)$ is of full rank, that $\Omega: \Theta \to \RR^{q \times q}$ has continuous l/r partial derivatives on $\Theta^+$ and that $\inf_{\vtheta \in \Theta} \lambda_1(\vtheta)>0$ . Finally assume, for $m=1, \ldots, q$, 
\begin{equation}\label{eq:ltol}
  \widehat{\ell}_{n,k}(\vc{c}_m) \pto \ell (\vc{c}_m ; \vtheta_0), 
  \qquad \text{as $n \to \infty$}.
\end{equation}
Then with probability tending to one, the minimizer $\htheta$ in \eqref{eq:estimator} exists and is unique. Moreover,
\begin{equation*}
  \htheta \pto \vc{\theta}_0, 
  \qquad \text{as $n \rightarrow \infty$}.
\end{equation*}
\end{thm}
%Note that the only difference between Theorem~\ref{theorem1} and \citet[Theorem 2.1]{einmahl2016nr2} is assumption (A1); differentiability on a neighbourhood of $\vtheta_0$ is replaced by existence of l/r partial derivatives on $\Theta^+$.
We omit the proof of this theorem since it is directly obtained by replacing $B_{\varepsilon} (\vc{\theta}_0)$ by $\Theta^+$ in the proof of \citet[Theorem 1]{einmahl2016nr2}.

When $f_{n,k}$ is not defined on a neighbourhood of $\vtheta_0$, but there exists a set $\Theta^+$ which satisfies (A1), a Taylor expansion of $f_{n,k} (\vtheta)$ around $f_{n,k} (\vtheta_0)$ holds \citep[Theorem 6]{andrews1999}. For each $\vtheta \in \Theta^+$, we have
\begin{align*}
f_{n,k} (\vtheta) & = f_{n,k} (\vtheta_0) + Df_{n,k} (\vtheta_0)^T (\vtheta - \vtheta_0) \\
& \quad + (\vtheta - \vtheta_0)^T D^2 f_{n,k} (\vtheta_0) (\vtheta - \vtheta_0)/2 + R_{n,k} (\vtheta),
\end{align*}
where $D f_{n,k}$ and $D^2 f_{n,k}$ are based on l/r partial derivatives and $R_{n,k} (\vtheta)$ is the remainder term. 
%When $\vtheta_0$ is in the interior of $\Theta$, $f_{n,k}$ is differentiable on a neighbourhood of $\vtheta_0$ and $D f_{n,k}$ and $D^2 f_{n,k}$ represent the gradient and hessian of $f_{n,k}$ respectively. 
We suppress dependence of $\Omega$, $\dot{L}$ and $D_{n,k}$ on $\vc{\theta}_0$ for ease of notation. From \citet[Proof of Theorems 3.2.1 and 3.2.2]{einmahl2016nr2} we know that
\begin{align*}
D f_{n,k} (\vtheta_0) = -2 D_{n,k}^T \Omega  \dot{L}  + o_p (1), \qquad
D^2 f_{n,k} (\vtheta_0) = 2 \dot{L}^T \Omega \dot{L}  + o_p (1),
\end{align*}
so the quadratic expansion above is equal to
\begin{align*}\label{eq:exp1}
f_{n,k} (\vtheta) = f_{n,k} (\vtheta_0) - 2 D_{n,k}^T \Omega \dot{L} (\vtheta - \vtheta_0) + (\vtheta - \vtheta_0)^T \dot{L}^T  \Omega   \dot{L} (\vtheta - \vtheta_0) + R_{n,k} (\vtheta).
\end{align*}
Define
\begin{equation*}
J := \dot{L}^T  \Omega  \dot{L} \in \RR^{p \times p} , \qquad \vc{Y}_{n,k} := \sqrt{k} J^{-1} \dot{L}^T  \Omega  D_{n,k}  \in \RR^p,
\end{equation*}
and
\begin{equation*}
q_{n,k} (\vc{\lambda}) = (\vc{\lambda} - \vc{Y}_{n,k})^T J (\vc{\lambda} - \vc{Y}_{n,k}).
\end{equation*}
Then $f_{n,k} (\vtheta)$ can be written as 
\begin{align*}\label{eq:exp2}
f_{n,k} (\vtheta) = f_{n,k} (\vtheta_0) - k^{-1} \vc{Y}_{n,k}^T \, J \, \vc{Y}_{n,k} + k^{-1}  q_{n,k} \left( \sqrt{k} (\vtheta - \vtheta_0) \right) + R_{n,k} (\vtheta). 
\end{align*}
We see that $f_{n,k} (\vtheta)$ is equal to a quadratic function plus a term that does not depend on $\vtheta$ plus $R_{n,k} (\vtheta)$.
In \citet[Lemma 3]{andrews2002}, it is shown that under the assumptions of Theorem~\ref{theorem1} and (A3) below, the remainder term $R_{n,k} (\vtheta)$ is sufficiently small so that minimizing the quadratic approximation to $f_{n,k} (\vtheta)$ is equivalent to minimizing $q_{n,k} \left( \sqrt{k} (\vtheta - \vtheta_0) \right)$. Now, assume that
\begin{description}
\item[(A2)] $\Theta - \vtheta_0$ is locally equal to a convex cone $\Lambda \subset \RR^p$.
\end{description}
%Assumption (A2) hold for the model presented in Section~\ref{sec:gml}. 
%The convexity of $\Lambda$ is necessary for the uniqueness of a minimizer. 
Then
\begin{equation*}\label{eq:min}
\inf_{\vtheta \in \Theta} q_{n,k} \left( \sqrt{k} (\vtheta - \vtheta_0) \right)   = \inf_{\vc{\lambda} \in \Lambda} q_{n,k}(\vc{\lambda}) + o_p(1),
\end{equation*}
for some cone $\Lambda$. 
%This allows for linear, kinked and curved boundaries as well as parameter spaces $\Theta - \vtheta_0$ that are defined by multivariate linear (in)equality 
Finally, assume that
\begin{description}
\item[(A3)] $\sqrt{k} \, D_{n,k}(\vtheta_0)
  \dto \vc{D} \sim
  \Normal_q ( 0, \Sigma (\vtheta_0) )$ as $n \rightarrow \infty$,
for some covariance matrix $\Sigma (\vtheta_0)$.
\end{description}
Assumptions (A3) and \eqref{eq:ltol} hold for the nonparametric estimators presented in Section~\ref{sec:nonpar} under some general regularity conditions. The matrix $\Sigma$ can be expressed in terms of $\ell$ and its first-order partial derivatives.
We have
\begin{equation*}
\vc{Y}_{n,k} \dto \vc{Y} := J^{-1} \dot{L}^T \Omega \vc{D}, \quad q_{n,k} (\vc{\lambda}) \dto q(\vc{\lambda}) := (\vc{\lambda} - \vc{Y})^T J (\vc{\lambda} - \vc{Y}) ,
\end{equation*}
for all $\vc{\lambda} \in \RR^p$.
%Let $\widehat{\vc{\lambda}} \in \textnormal{cl}(\Lambda)$ be such that $q(\widehat{\vc{\lambda}}) = \inf_{\vc{\lambda} \in \Lambda} q(\vc{\lambda})$ \comment{necessary?}. 

\begin{thm}[Asymptotic normality]\label{theorem2}
If all of the above assumptions hold, 
\begin{equation*}
\sqrt{k} (\htheta - \vtheta_0) \dto \widehat{\vc{\lambda}} = \argmin_{\vc{\lambda} \in \Lambda} q(\vc{\lambda}), \qquad \text{ as } n \rightarrow \infty.
\end{equation*}
\end{thm}
Here $\widehat{\vc{\lambda}}$ can be interpreted as the projection of $\vc{Y}$ on $\Lambda$ with respect to the norm $\lVert y \rVert_{J} := y^T J y$. This theorem is a special case of \citet[Theorem 3]{andrews1999}. In the appendix, we verify that our assumptions are sufficient. 
Note that if $\Theta$ includes a neighbourhood of $\vtheta_0$, then $\Lambda = \RR^p$ and we find the same distribution as in \citet[Theorem 2]{einmahl2016nr2} since $\widehat{\vc{\lambda}} = \vc{Y}$.

\subsection{Simplifying the asymptotic distribution}\label{sec:simplify}
The goal of this section is to simplify the asymptotic distribution of $\sqrt{k} (\htheta - \vtheta_0) $ and to give conditions under which (part of) $\widehat{\vc{\lambda}}$ has a closed-form expression. Let $\vc{G} := \dot{L}^T \Omega \vc{D}$, i.e., $\vc{Y} = J^{-1} \vc{G}$. We start by partitioning the vector $\vtheta_0 \in \RR^p$ into two subvectors, $\vc{\beta}_0 \in \RR^c$ and $\vc{\delta}_0 \in \RR^{p-c}$ for $c \in \{1,\ldots,p\}$, where $\vc{\delta}_0$ consists of all parameters that are in the interior of the parameter space. 
We partition $\vc{\theta}_{n,k}$, $\vc{Y}$, $J$, $\vc{G}$ and $\vc{\lambda}$ accordingly:
\begin{align*}
\vc{\theta}_{n,k} = \begin{pmatrix} \vc{\beta}_{n,k} \\ \vc{\delta}_{n,k} \end{pmatrix}, 
\,\,
\vc{Y} = \begin{pmatrix} \vc{Y}_{\vc{\beta}} \\ \vc{Y}_{\vc{\delta}} \end{pmatrix}, \,\,
J = \begin{pmatrix} J_{\beta} & J_{\beta \delta} \\ J_{\delta \beta} & J_{\delta} 
\end{pmatrix}, \,\,
\vc{G} = \begin{pmatrix} \vc{G}_{\vc{\beta}} \\ \vc{G}_{\vc{\delta}} \end{pmatrix}, \,\,
\vc{\lambda} = \begin{pmatrix} \vc{\lambda}_{\beta} \\ \vc{\lambda}_{\delta} \end{pmatrix}.
\end{align*}
%Next, note that $\vc{Y}_{\beta} = H \vc{Y}$ where 
Let $I_c$ denote the $c \times c$ identity matrix. For $H := \big( I_c : \vc{0} \big) \in \RR^{c \times p}$, define
\begin{align*}
q_{\beta} (\vc{\lambda}_{\beta}) := (\vc{\lambda}_{\beta} - \vc{Y}_{\beta})^T \left( H J^{-1} H^T \right)^{-1} (\vc{\lambda}_{\beta} - \vc{Y}_{\beta}).
\end{align*}
Suppose that
\begin{description}
\item[(A4)] The cone $\Lambda$ of assumption (A2) is equal to the product set $\Lambda_{\vc{\beta}} \times \RR^{p-c}$, where $\Lambda_{\vc{\beta}} \subset \RR^{c}$ is a cone.
\end{description}
%Assumption (A4) holds for all the models discussed in Section~\ref{sec:gml}.

\begin{cor}\label{cor1}
If all of the above assumptions hold, then
\begin{align*}
\sqrt{k} (\widehat{\vc{\beta}}_{n,k} - \vc{\beta}_0) & \dto \widehat{\vc{\lambda}}_{\beta} := \argmin_{\vc{\lambda}_{\beta} \in \Lambda_{\beta}} q_{\beta} (\vc{\lambda}_{\beta}), \\
 \sqrt{k} (\widehat{\vc{\delta}}_{n,k} - \vc{\delta}_0) & \dto J^{-1}_{\delta} \vc{G}_{\delta} - J^{-1}_{\delta} J_{\delta \beta} \widehat{\vc{\lambda}}_{\beta}.
\end{align*}
\end{cor}
When $\Lambda_{\beta}$ is defined by equality and/or inequality constraints, a closed-form expression for $\widehat{\vc{\lambda}}_{\beta}$ can be computed; we give some examples that are relevant for Section~\ref{sec:simulation}. For a more formal solution, see \citet[Theorem 5]{andrews1999}. 

\begin{exa}[$c=1$]\label{exa1}
Suppose that $\Theta = [0,1]^p$. If $\beta_0 = 0$, then $\Lambda_{\beta} = [0,\infty)$ and $\widehat{\lambda}_{\beta} = \max (Y_{\beta}, 0)$. If $\beta_0 = 1$, then $\Lambda_{\beta} = (-\infty,0]$ and $\widehat{\lambda}_{\beta} = \min (Y_{\beta}, 0)$.
\end{exa}

\begin{exa}[$c = 2$]\label{exa2}
Suppose that $\Theta = [0,1]^p$. If $\vc{\beta}_0 = (1,1)$, then $\Lambda_{\beta} = (-\infty,0]^2$. Let $\rho: = H J^{-1} H^T$. Then
\begin{align*}
\widehat{\vc{\lambda}}_{\beta} & := \mathbbm{1} \left\{ Y_{\beta 1} < 0, Y_{\beta 2} < 0 \right\} \vc{Y}_{\beta} \\
& \qquad + \mathbbm{1} \left\{ Y_{\beta 1} - \rho_{21} Y_{\beta 2} < 0, Y_{\beta 2} \geq 0 \right\} (Y_{\beta 1} - \rho_{21} Y_{\beta 2}, 0)^T  \\
& \qquad + \mathbbm{1} \left\{ Y_{\beta 1} \geq 0, Y_{\beta  2} - \rho_{12} Y_{\beta 1} < 0 \right\} (0, Y_{\beta  2} - \rho_{12} Y_{\beta 1})^T.
\end{align*}
If $\vc{\beta}_0 = (0,0)$, $\vc{\beta}_0 = (0,1)$ or $\vc{\beta}_0 = (1,0)$, similar expressions can be obtained by reversing the first and/or the second inequalities in each of the indicator functions above.
\end{exa}

\subsection{Hypothesis testing for tail dependence parameters on the boundary of the parameter space}\label{sec:test}
We are interested in constructing hypothesis tests for parameter values that, under the null hypothesis, are on the boundary of the alternative hypothesis. We propose two test statistics, whose asymptotic distribution follows from the results in \citet{andrews2001}.

Assume that there are no nuisance parameters on the boundary, i.e., all components of $\vtheta_0$ that lie on the boundary are part of the null hypothesis. We are interested in testing
\begin{equation*}
H_0 : \vc{\beta} = \vc{\beta}_* \quad \text{vs} \quad H_1: \vc{\beta} \neq \vc{\beta}_*, \qquad \vc{\beta}_* \in \RR^c.
\end{equation*}
Let $\Theta_0 := \{ \vtheta \in \Theta: \vtheta = (\vc{\beta}_*,\vc{\delta}) \text{ for some } \vc{\delta} \in \RR^{p-c} \}$ denote the restricted parameter space. Assume that
\begin{description}
\item[(A5)] For all $\vtheta \in \Theta_0$, $\Theta$ is a product set with respect to $(\vc{\beta}, \vc{\delta})$ local to $\vtheta$. That is, for all $\vtheta \in \Theta_0$,
\begin{equation*}
\Theta \cap B_{\varepsilon} (\vtheta) = \left( \mathcal{B} \times \Delta \right) \cap B_{\varepsilon} (\vtheta) \quad \text{ for some } \mathcal{B} \subset \RR^c, \, \Delta \subset \RR^{p-c} \text{ and } \varepsilon > 0.
\end{equation*}
\end{description}
%If all of the above assumptions hold \comment{more specific}, then $\Theta_0 - \vtheta_0$ is locally equal to $\Lambda_0 = \{\vc{0} \} \times \RR^{p-c}$. 
Define $\htheta^{(0)} := \argmin_{\vtheta \in \Theta_0} f_{n,k} (\vtheta)$. A deviance test statistic can be defined as
\begin{equation*}
T^{(1)}_{n,k} : = k \left(f_{n,k} (\htheta^{(0)}) - f_{n,k} (\htheta) \right).
\end{equation*}
%Define $\widehat{\vc{\lambda}}_0 = \argmin_{\vc{\lambda} \in \Lambda_0} q(\vc{\lambda}) \in \cl{\Lambda_0}$. We can partition $\widehat{\vc{\lambda}}_0$ as $\widehat{\vc{\lambda}}_0 = \left(\widehat{\vc{\lambda}}_{\beta 0}, \widehat{\vc{\lambda}}_{\delta 0} \right) = \left(\vc{0}, \widehat{\vc{\lambda}}_{\delta 0} \right)$.

\begin{cor}\label{cor2}
Suppose $\vtheta_0 \in \Theta_0$ and all previous assumptions hold. Then
\begin{equation*}\label{eq:adtest}
T^{(1)}_{n,k} \dto  \widehat{\vc{\lambda}}_{\beta}^T \left( H J^{-1} H^T \right)^{-1} \widehat{\vc{\lambda}}_{\beta}.
\end{equation*}
\end{cor}
Recall that $\vc{G} = \dot{L}^T \Omega \vc{D}$ and let $\mathcal{J}$ denote its covariance matrix, i.e., $\vc{G} \sim \mathcal{N}_p (\vc{0}, \mathcal{J})$ with $\mathcal{J} = \dot{L}^T \Omega \Sigma \Omega \dot{L} \in \RR^{p \times p}$. Note that $\mathcal{J} = J$ if and only if $\Omega = \Sigma^{-1}$.
If $\Lambda_{\beta} = \RR^c$ (no parameters on the boundary) and $\mathcal{J} = J$, then 
\begin{equation*}
T^{(1)}_{n,k} \dto \vc{Y}_{\beta}^T \left( H J^{-1} H^T \right)^{-1} \vc{Y}_{\beta} \sim \chi^2_c,
\end{equation*}
where $\chi^2_c$ denotes a chi-squared random variable with $c$ degrees of freedom. %If $\mathcal{J} = J$ but $\Lambda_{\beta} \neq \RR^c$, then the asymptotic distribution of $T^{(1)}_{n,k}$ is that of a mixture of chi-square random variables. Formulas for the mixing weights for $c \leq 4$ can be found in \citet[Section 4]{shapiro1985}. We will not always want to assume that $\mathcal{J} = J$; in that case, we'll approximate the distribution of \eqref{eq:adtest} by simulation.

A Wald-type test statistic can be based on the quadratic form in $\widehat{\vc{\beta}}_{n,k} - \vc{\beta}_*$. Let $\widehat{V}_n^{-1} := (H J^{-1}_n H^T)^{-1} \in \RR^{c \times c}$ denote a weight matrix where $J_n := 
J(\widehat{\vtheta}_{n,k})$. Define 
\begin{equation*}
T_{n,k}^{(2)} := k \left(\widehat{\vc{\beta}}_{n,k} - \vc{\beta}_* \right)^T  \widehat{V}^{-1}_n \left(\widehat{\vc{\beta}}_{n,k} - \vc{\beta}_* \right).
\end{equation*}
\begin{cor}\label{cor3}
Suppose $\vtheta_0 \in \Theta_0$ and all of the above assumptions hold. Then
\begin{equation*}
T_{n,k}^{(2)} \dto \widehat{\vc{\lambda}}_{\beta}^T \left( H J^{-1} H^T \right)^{-1} \widehat{\vc{\lambda}}_{\beta}.
\end{equation*}
\end{cor}
Corollary~\ref{cor2} is a special case of \citet[Theorem 4c]{andrews2001} and Corollary~\ref{cor3} is a special case of \citet[Theorem 6d]{andrews2001}; their proofs are direct and thus omitted.

\section{Simulation studies}\label{sec:simulation}
We conduct simulation experiments where we simulate $1000$ samples of size $n=5000$ from a parametric model and we assess the quality of the weighted least squares estimator in terms of its root mean squared error (RMSE) for settings where one or more of the parameters are on the boundary of the parameter space. We take $k \in \{25, 50, \ldots, 300 \}$ and we use the empirical and the beta tail dependence function as initial estimators. Next, we study the performance of the two test statistics introduced in Section~\ref{sec:test} in terms of empirical level and power for $k \in \{25, 50, 75, 100 \}$, as we observe that higher $k$ leads to a steadily growing bias that quickly deteriorates the performance of the hypothesis test. In all experiments, we take $\Omega = I_q$, the $q \times q$ identity matrix, because using an optimal weight matrix has a very minor effect on the quality of estimation while severely slowing down the estimation procedure. Moreover,
inverting $\Sigma$ may be hindered by numerical problems in the case of a max-linear model \citep{einmahl2016nr2}. 

\subsection{Brown--Resnick model}
We simulate data from a Brown--Resnick model with $\alpha = 2$ and $\rho = 1$ on a regular $3 \times 2$ grid ($d = 6$) and on a regular $4 \times 4$ grid ($d = 16$). As in \cite{einmahl2016nr2}, we let $\vc{c}_m \in \{0,1\}^d$ such that exactly two components of $\vc{c}_m $ are equal to one (meaning that we consider a pairwise estimator) and we focus on pairs of neighbouring locations only, i.e., locations that are at most a distance of $\sqrt{2}$ apart. This leads to $q = 11$ and $q = 42$ respectively. 

\begin{figure}[ht]
\centering
\subfloat{\includegraphics[width=0.245\textwidth]{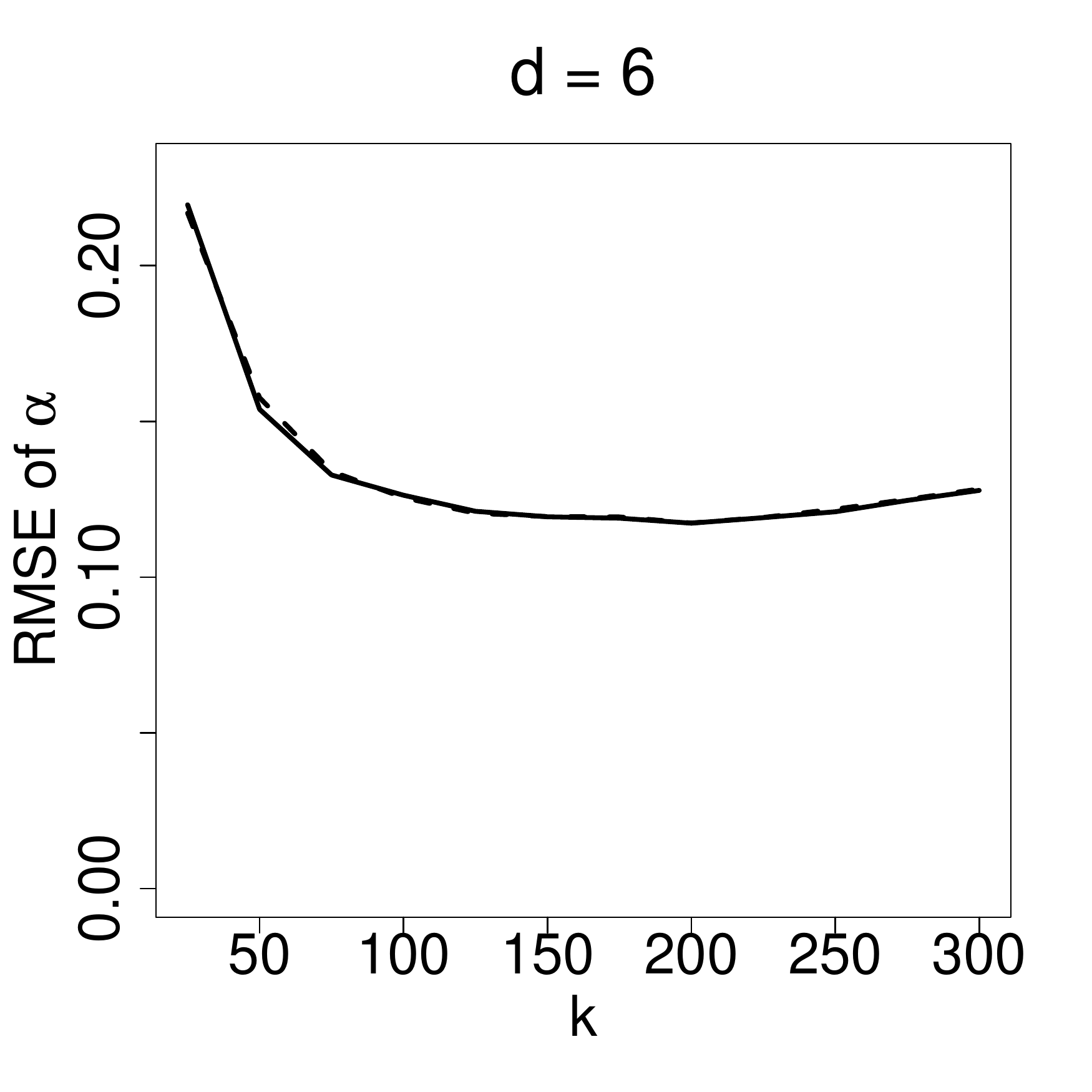}}
\subfloat{\includegraphics[width=0.245\textwidth]{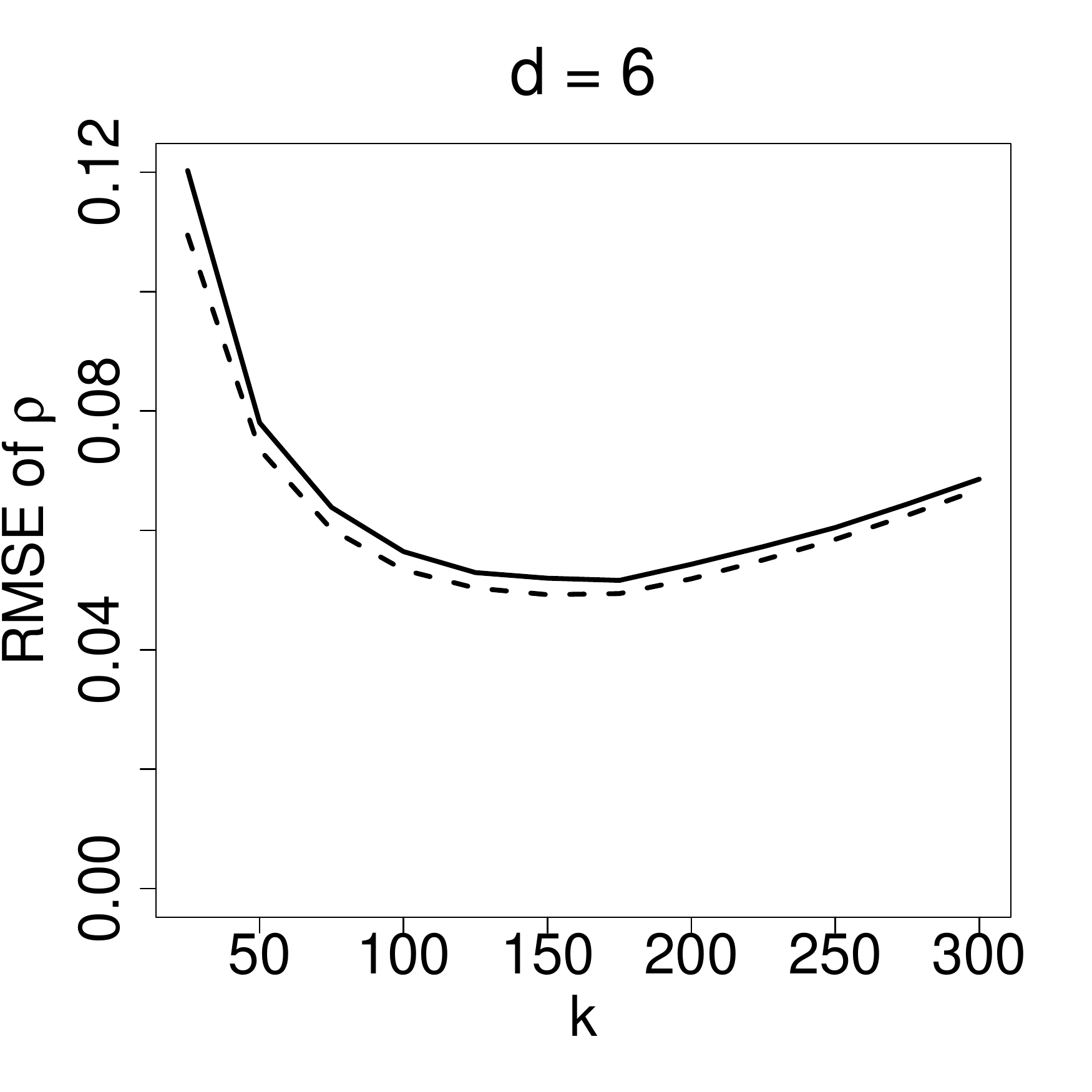}} 
\subfloat{\includegraphics[width=0.245\textwidth]{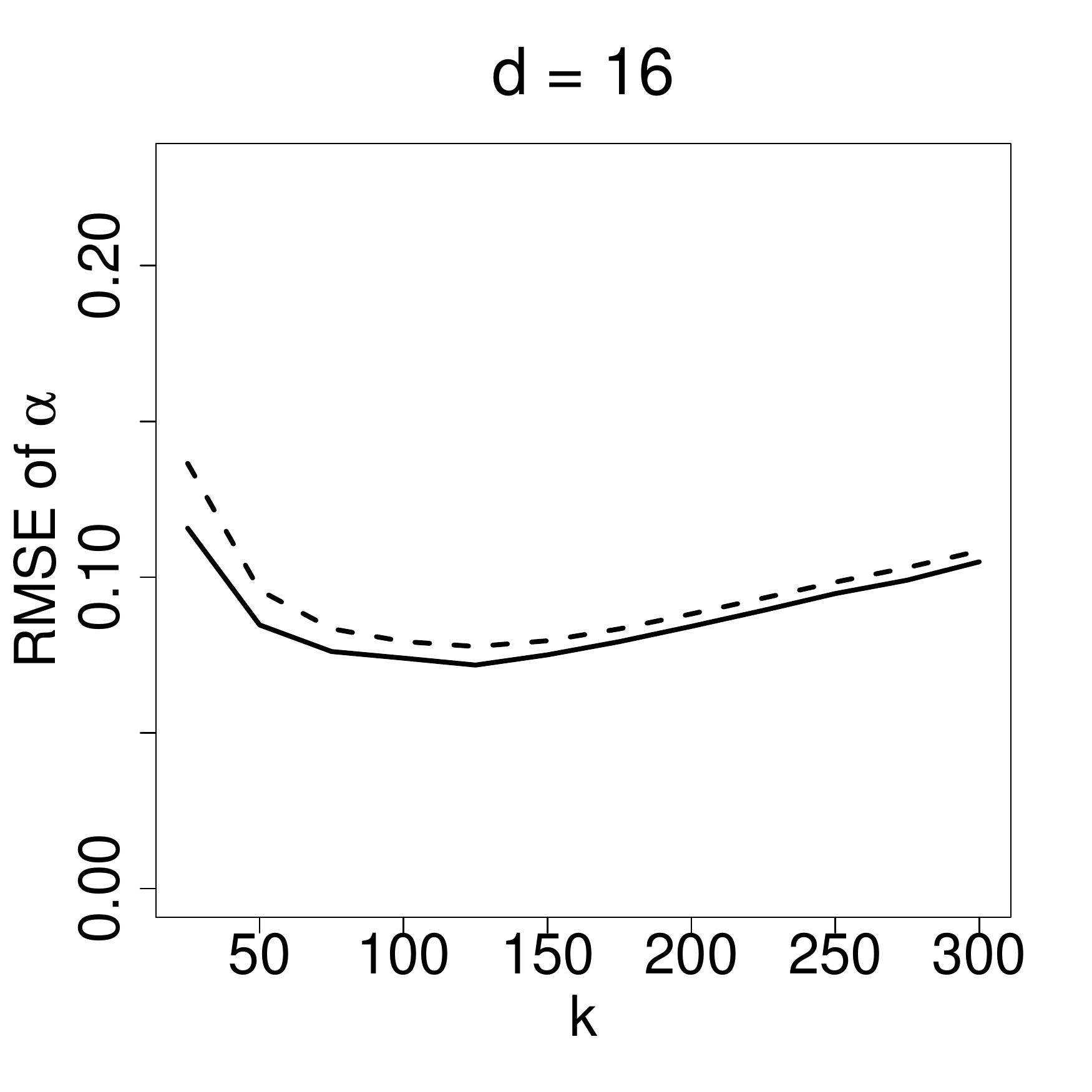}}
\subfloat{\includegraphics[width=0.245\textwidth]{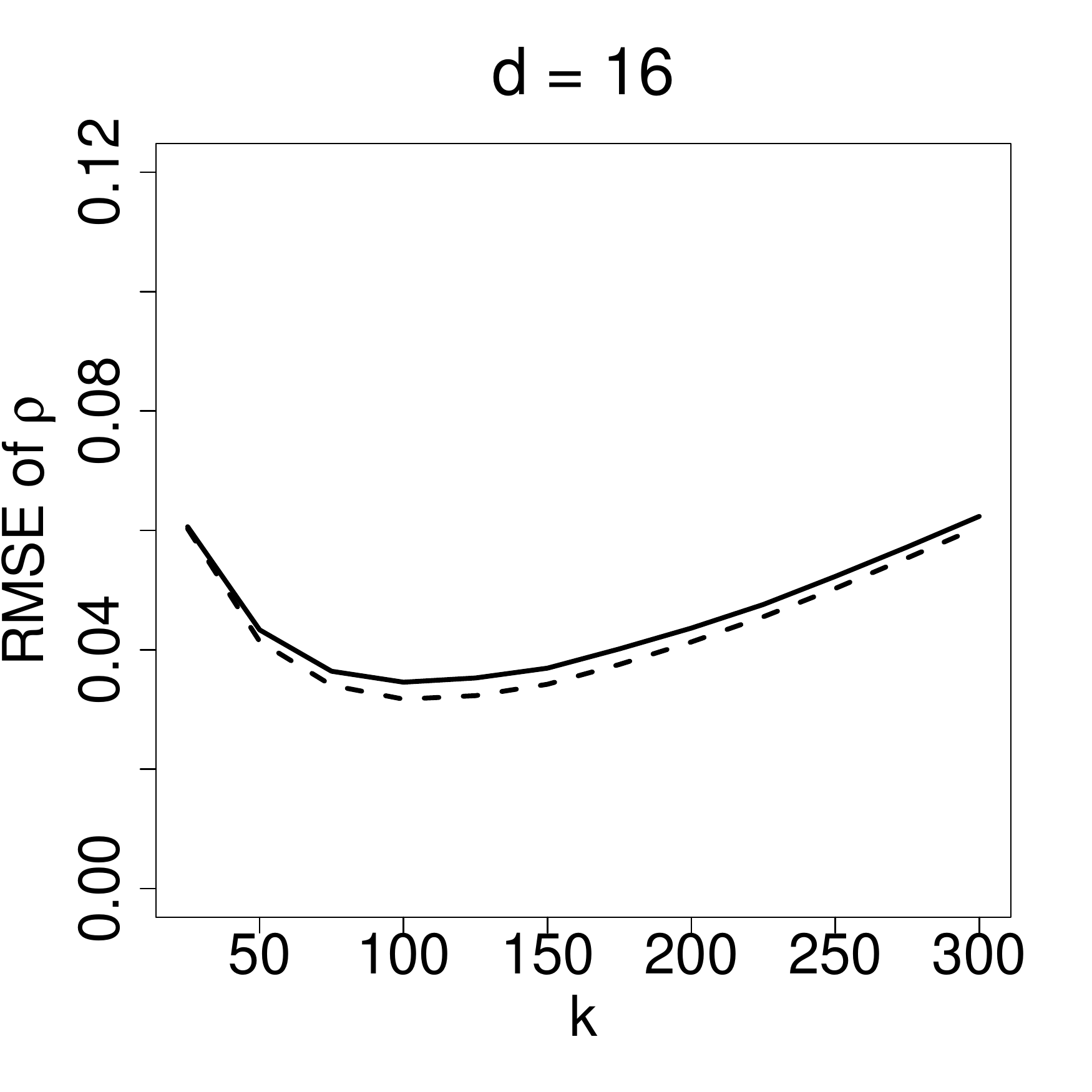}} 
\caption{RMSE for parameter estimators of the Brown--Resnick model based on the empirical tail dependence function (solid lines) and the beta tail dependence function (dashed lines) for $(\alpha,\rho) = (2,1)$; $d = 6$ (left) and $d = 16$ (right).} 
\label{fig:ellTotBR}
\end{figure}

Figure~\ref{fig:ellTotBR} shows the RMSE of the parameter estimates based on the empirical tail dependence function (solid lines) and the beta tail dependence function (dashed lines). We see that the empirical tail dependence function outperforms the beta tail dependence function for $\alpha$, while the opposite is true for $\rho$, although differences are very minor. For small values of $k$, the RMSE is lower for $d = 16$ than for $d = 6$.

Next, we study the empirical level and power of the test $H_0: \alpha = 2$. Assumtion (A4) holds with $\Lambda = (-\infty,0] \times \RR$. Table~\ref{tab:level1BR} shows the empirical level of the test statistic $T_{n,k}^{(2)}$ based on a significance level of $\alpha = 0.05$. We do not display $T_{n,k}^{(1)}$ as its behaviour is identical to that of $T_{n,k}^{(2)}$. Surprisingly, the empirical tail dependence function is preferred for $d = 16$ while the beta tail dependence function performs best for $d = 6$. In general, we can conclude that low values of $k$ are to be preferred and that the test performs better for lower dimension. 
Figure \ref{fig:powerBR} shows the empirical power of $T_{n,k}^{(2)}$ as a function of $\alpha \in [1.5,2)$. The power increases both with the dimension and with $k$.

\begin{table}[ht]
\centering
 \begin{tabular}{lrrrrrrrr}
 & \multicolumn{2}{c}{$k = 25$} & \multicolumn{2}{c}{$k = 50$}  & \multicolumn{2}{c}{$k = 75$}  & \multicolumn{2}{c}{$k = 100$}  \\
\cmidrule(r){2-3}
\cmidrule(r){4-5}
\cmidrule(r){6-7}
\cmidrule(r){8-9}
& emp & beta & emp & beta & emp & beta & emp  & beta  \\
\cmidrule(r){2-9}
$d = 6$ & $6.0$ & $4.3$ & $5.1$ & $5.0$ & $7.0$ & $5.5$ & $8.8$ & $8.0$ \\
$d = 16$ & $6.7$ & $10.3$ & $ 8.3$ & $10.2$ & $9.8$  & $12.4$ & $13.2$  & $16.0$ \\
\bottomrule
\end{tabular}
\caption{Empirical level of $T_{n,k}^{(2)}$ based on the empirical and the beta tail dependence function for a significance level of $0.05$.}
\label{tab:level1BR}
\end{table}

\begin{figure}[ht]
\centering
\subfloat{\includegraphics[width=0.245\textwidth]{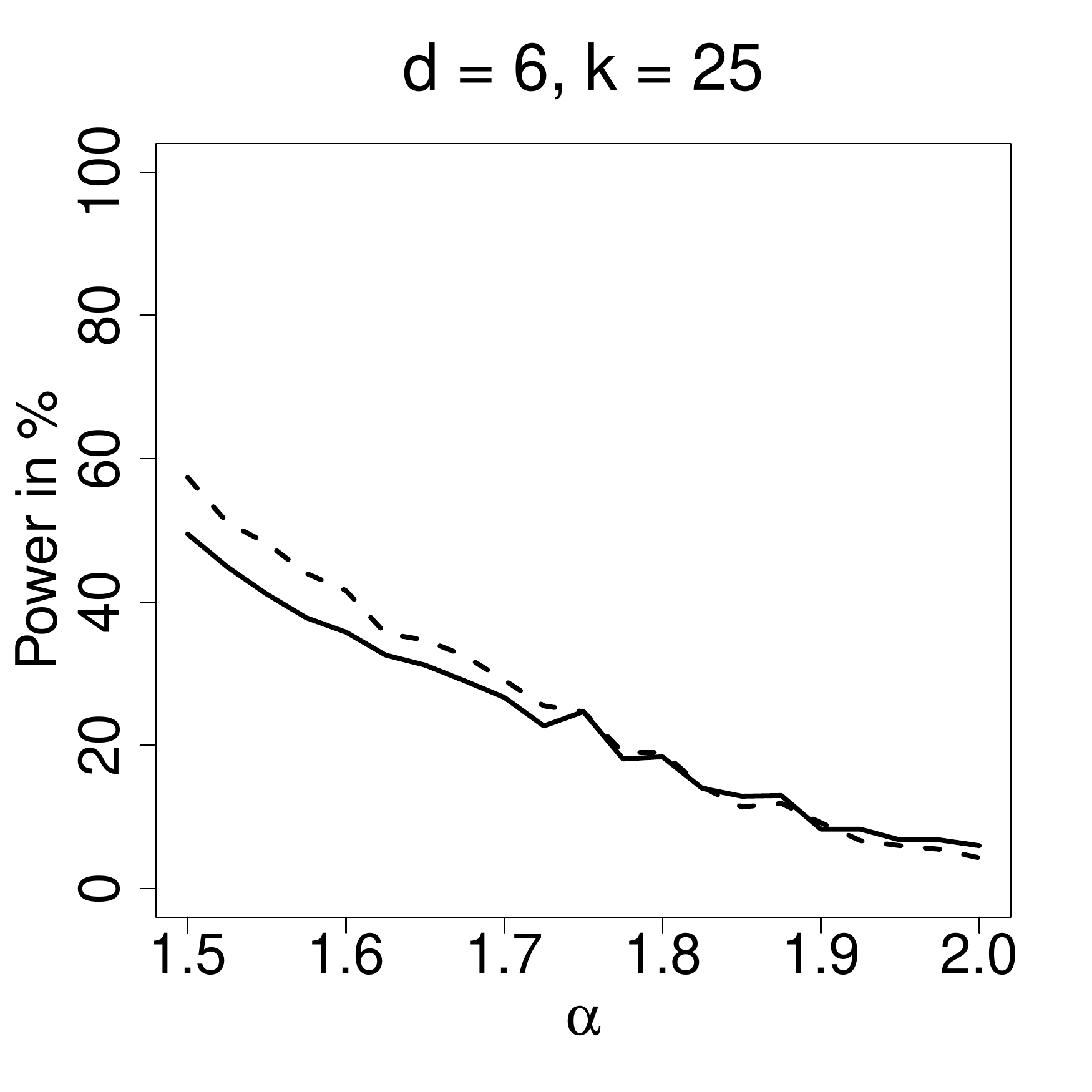}} 
\subfloat{\includegraphics[width=0.245\textwidth]{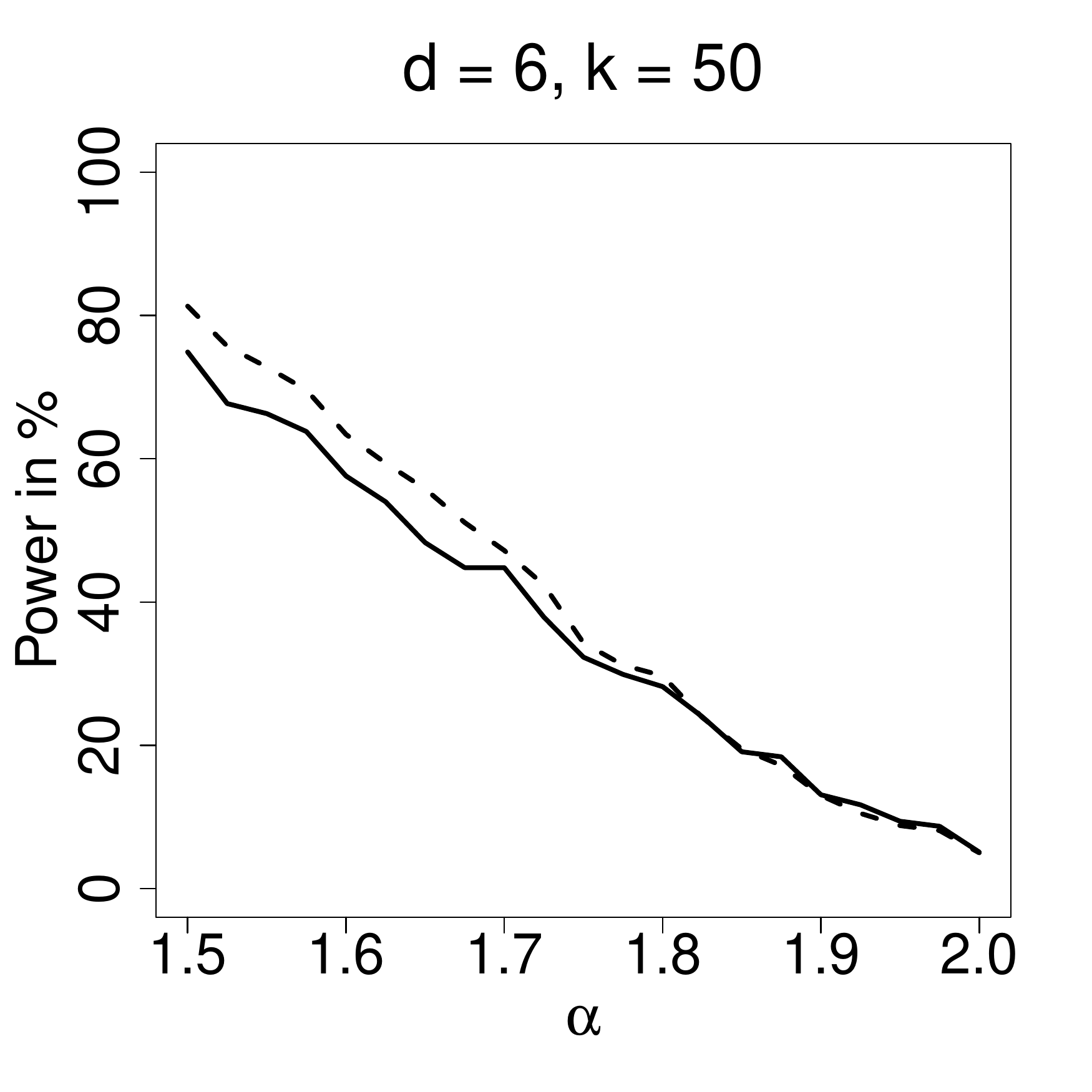}} 
\subfloat{\includegraphics[width=0.245\textwidth]{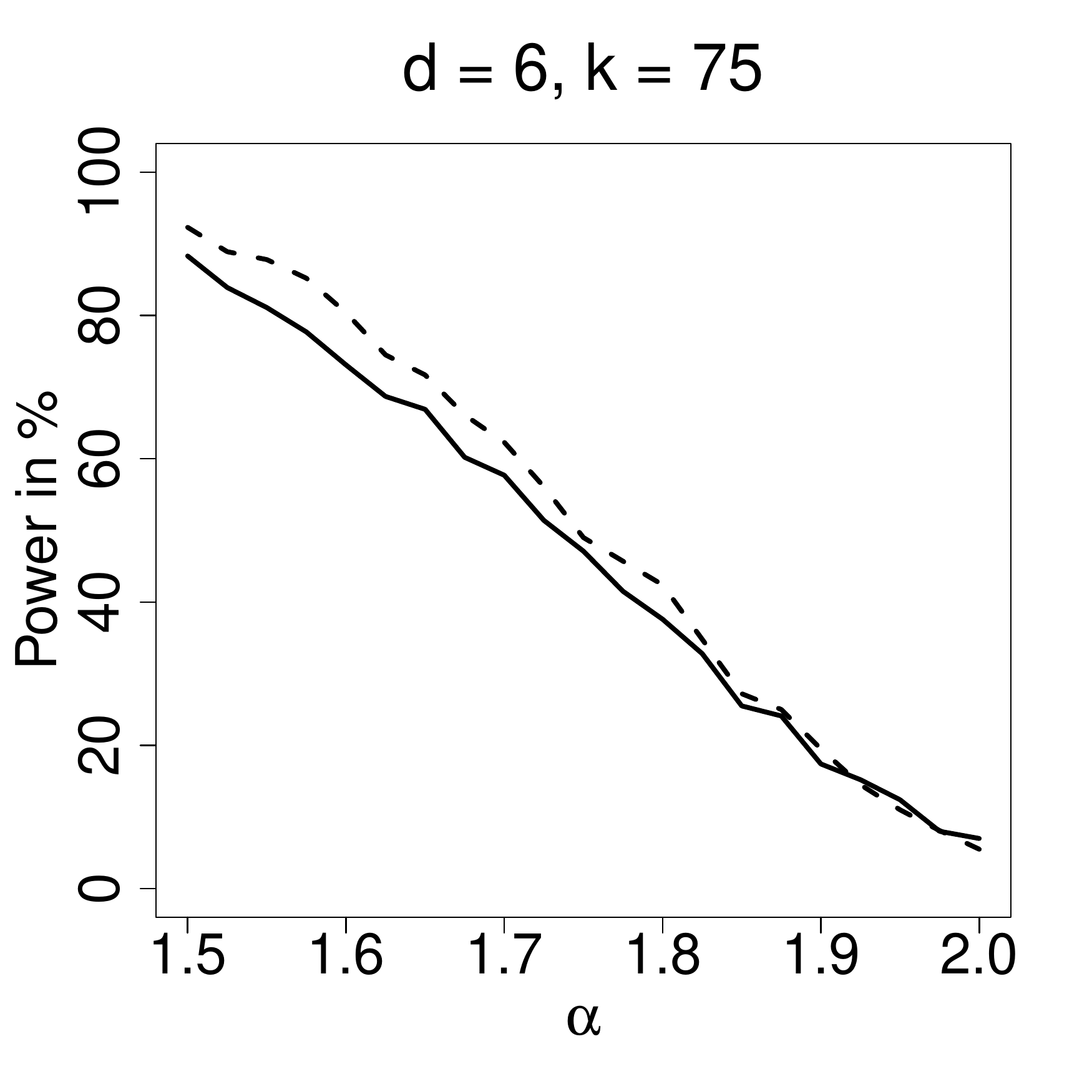}}
\subfloat{\includegraphics[width=0.245\textwidth]{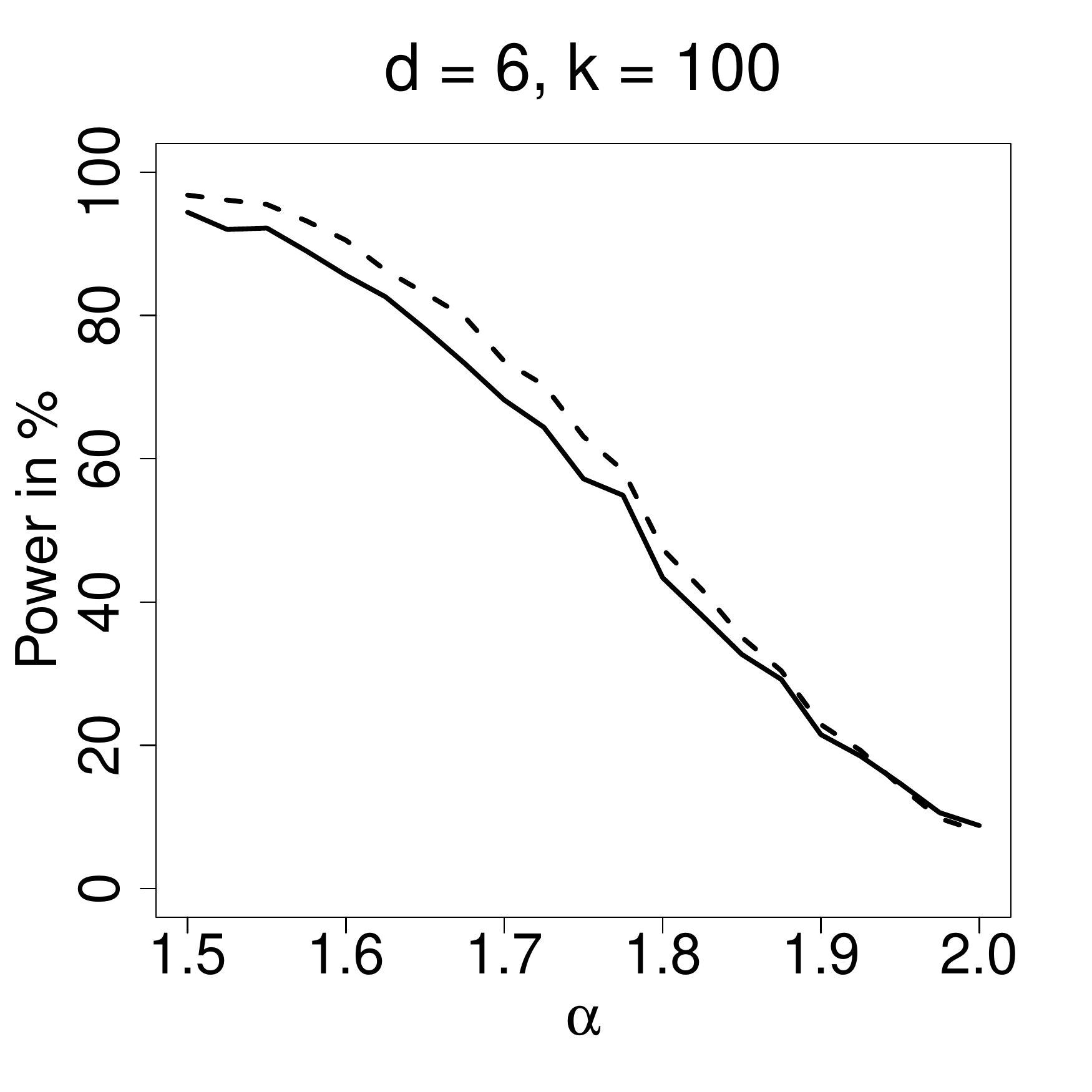}} \\
\subfloat{\includegraphics[width=0.245\textwidth]{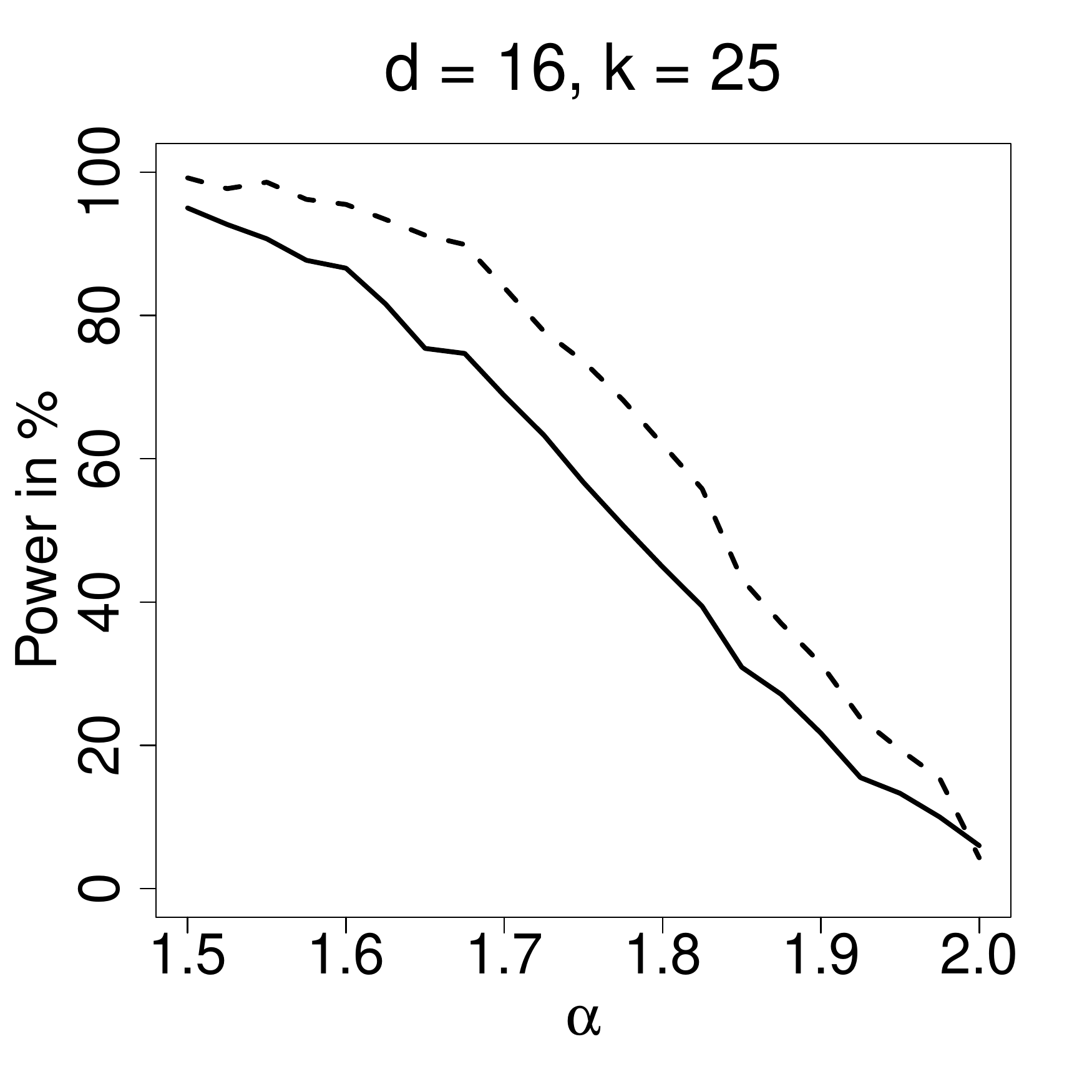}} 
\subfloat{\includegraphics[width=0.245\textwidth]{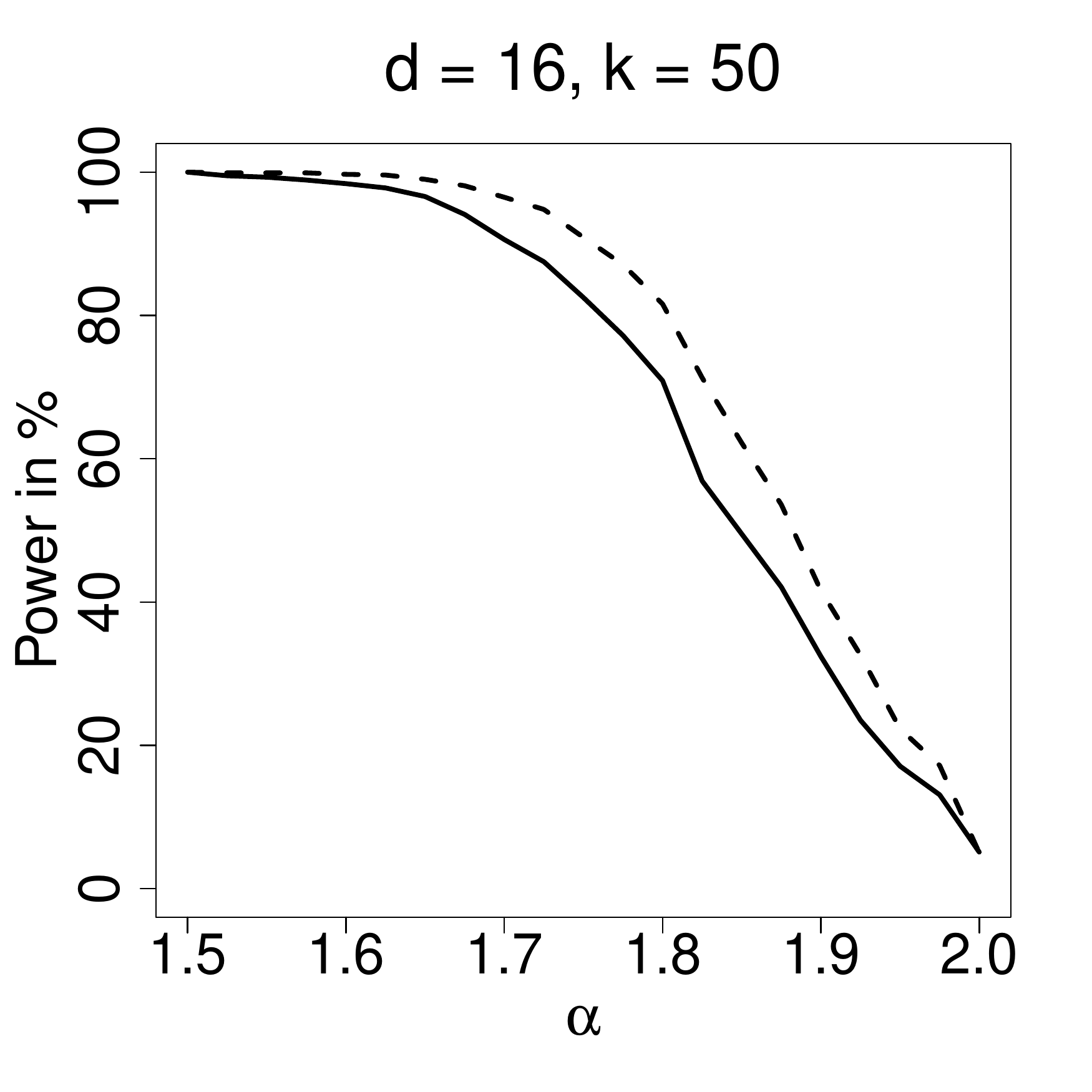}} 
\subfloat{\includegraphics[width=0.245\textwidth]{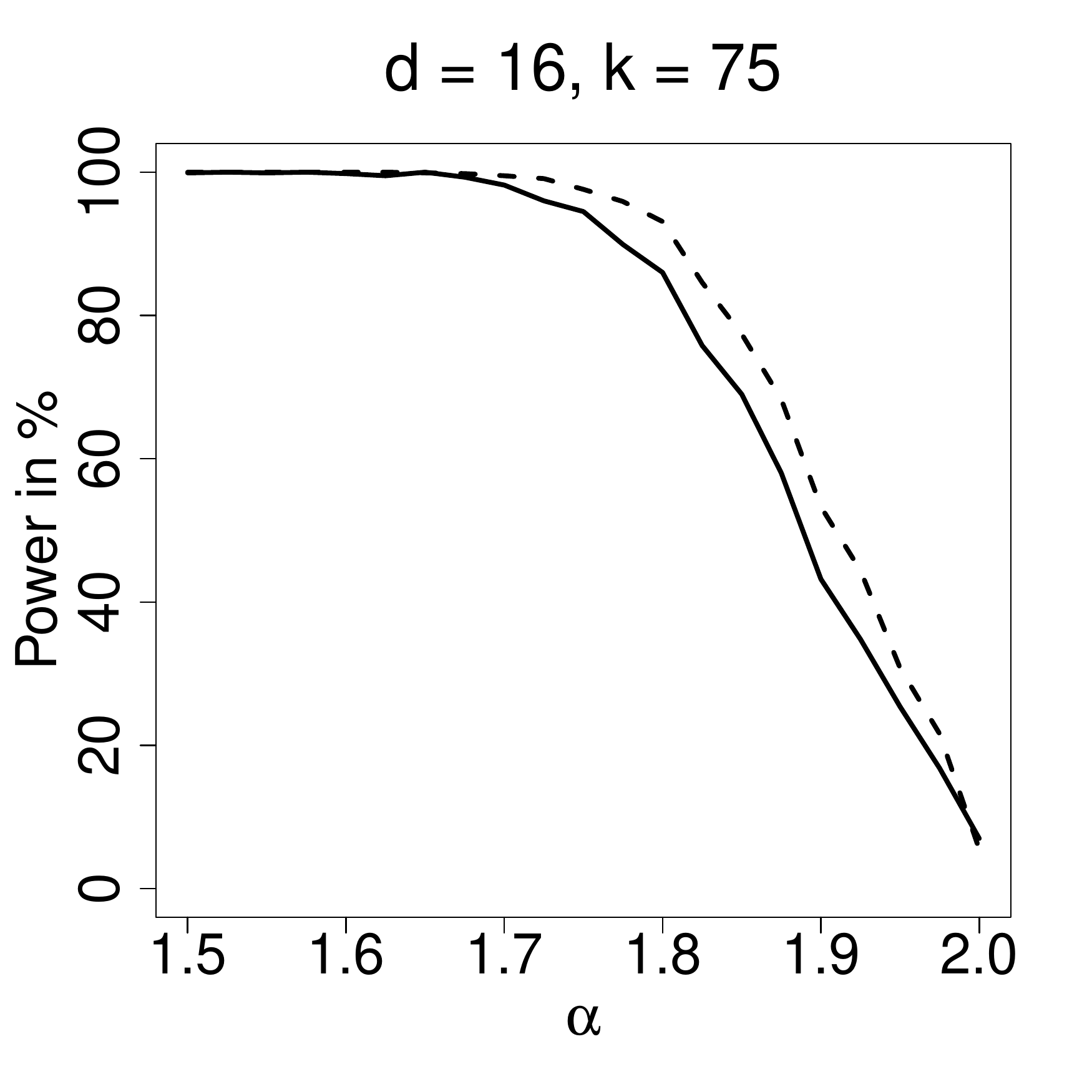}}
\subfloat{\includegraphics[width=0.245\textwidth]{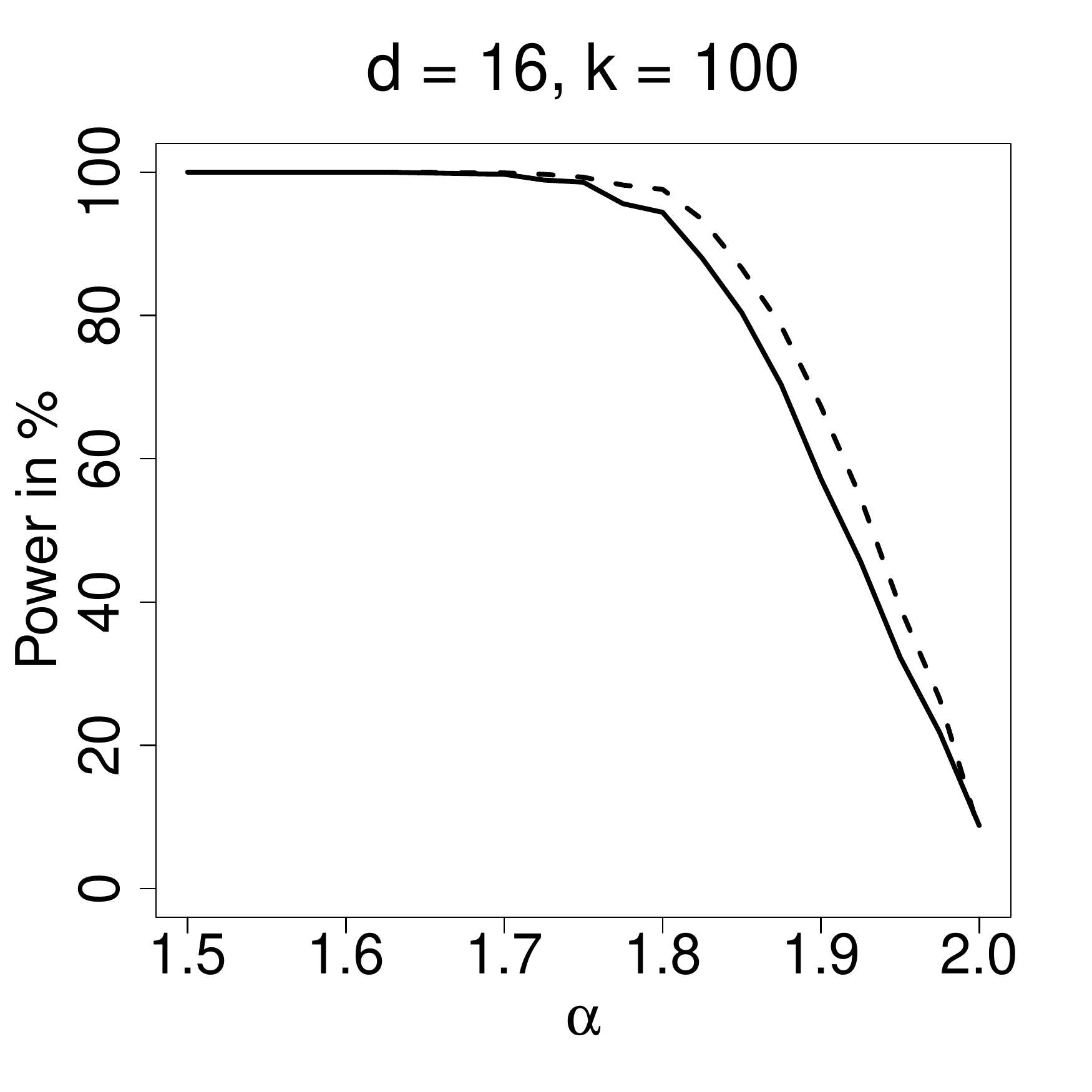}} 
\caption{Empirical power of $T_{n,k}^{(2)}$ based on the empirical tail dependence function (solid lines) and the beta tail dependence function (dashed lines) for $k \in \{25, 50, 75, 100\}$; $d = 6$ (top) and $d = 16$ (bottom).} 
\label{fig:powerBR}
\end{figure}

\subsection{Max-linear model}\label{sec:gml}
The constants $\vc{c}_1,\ldots,\vc{c}_q$ need to be chosen such that the max-linear model is identifiable. In \citet{einmahl2016nr2} it was observed that taking extremal coefficients for $\vc{c}_m$, i.e., $\vc{c}_m \in \{ 0,1 \}^d$ which at least two non-zero elements, is not enough for identifiability. Instead, one needs to choose $\vc{c}_m$ such that its non-zero elements are unequal. For some theoretical considerations, see also \citet[Appendix B]{einmahl2016nr2}. 
Simulation experiments showed that in practice, taking a large grid of values (meaning that $q \gg p$) on $[0,1]^d$ will lead to good estimators in terms of RMSE. Hence, in all following simulation experiments, we choose $\vc{c}_1,\ldots,\vc{c}_q$ such that
\begin{equation*}
\vc{c}_m \in \{0,0.01,0.1,0.2,\ldots,0.8,0.9,0.99,1\}^2, \qquad m \in \{1,\ldots,q\}.
\end{equation*}

\subsubsection{Testing the structure of a max-linear model}
We consider three different scenarios:
%\begin{equation*}
%B = \begin{pmatrix} b_{11} & 1 - b_{11} \\ 
%b_{21} & 1 - b_{21} \\ 
%b_{31} & 1 - b_{31} \end{pmatrix}.
%\end{equation*}
\begin{description}
\item[Model 1] Let $r = 2$ and $d = 3$, so that $\Theta = [0,1]^3$. Let $\vtheta_0 = (b_{11}, b_{21}, b_{31}) = (1, 0.7, 0.2)$. We test $H_0: b_{11} = 1$. Assumption (A4) holds with $\Lambda = (-\infty,0] \times \RR^2$.
\item[Model 2] Let $r = 2$ and $d = 3$, so that $\Theta = [0,1]^3$. Let $\vtheta_0 = (b_{11}, b_{21}, b_{31}) = (1, 0.7, 0)$. We test $H_0: (b_{11}, b_{31}) = (1,0)$. Assumption (A4) holds with $\Lambda = (-\infty,0] \times [0,\infty) \times \RR$.
\item[Model 3] Let $r = 3$ and $d = 2$, so that $\Theta = [0,1]^4$ . Let $\vtheta_0 = (b_{11}, b_{21}, b_{12},b_{22}) = (0, 0.8, 0.6,0)$. We test $H_0: (b_{11},b_{22}) = (0,0)$. Assumption (A4) holds with $\Lambda = [0,\infty) \times \RR^2 \times [0,\infty)$. If the null hypothesis cannot be rejected, it means that a Marshall--Olkin model suffices.
\end{description}

%\begin{equation*}
%B = \begin{pmatrix} b_{11} & b_{12} & 1 - b_{11} - b_{12} \\ 
%b_{21} & b_{22} & 1 - b_{21} - b_{22}  \end{pmatrix}.
%\end{equation*}

\begin{figure}[p]
\centering
\subfloat{\includegraphics[width=0.245\textwidth]{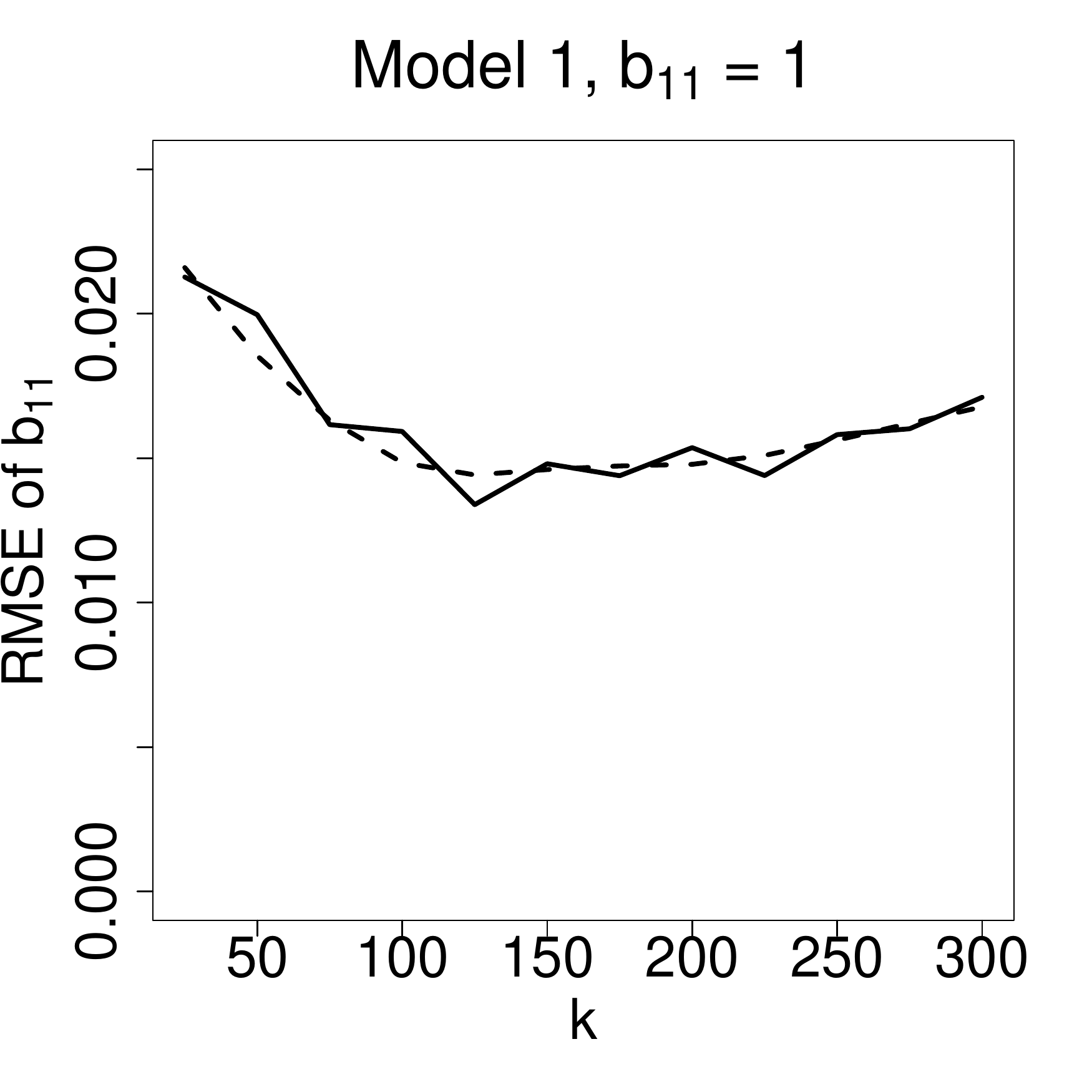}} 
\subfloat{\includegraphics[width=0.245\textwidth]{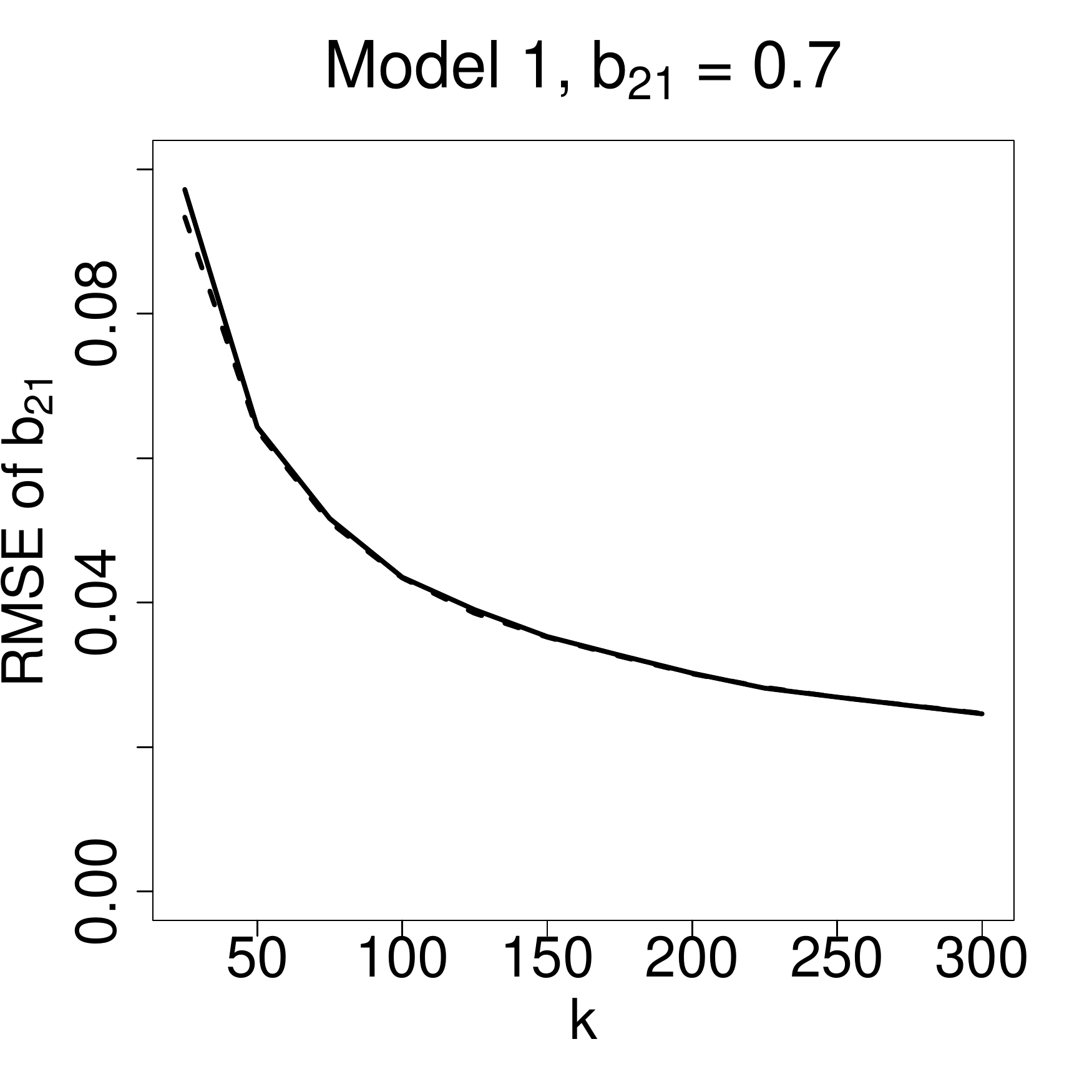}} 
\subfloat{\includegraphics[width=0.245\textwidth]{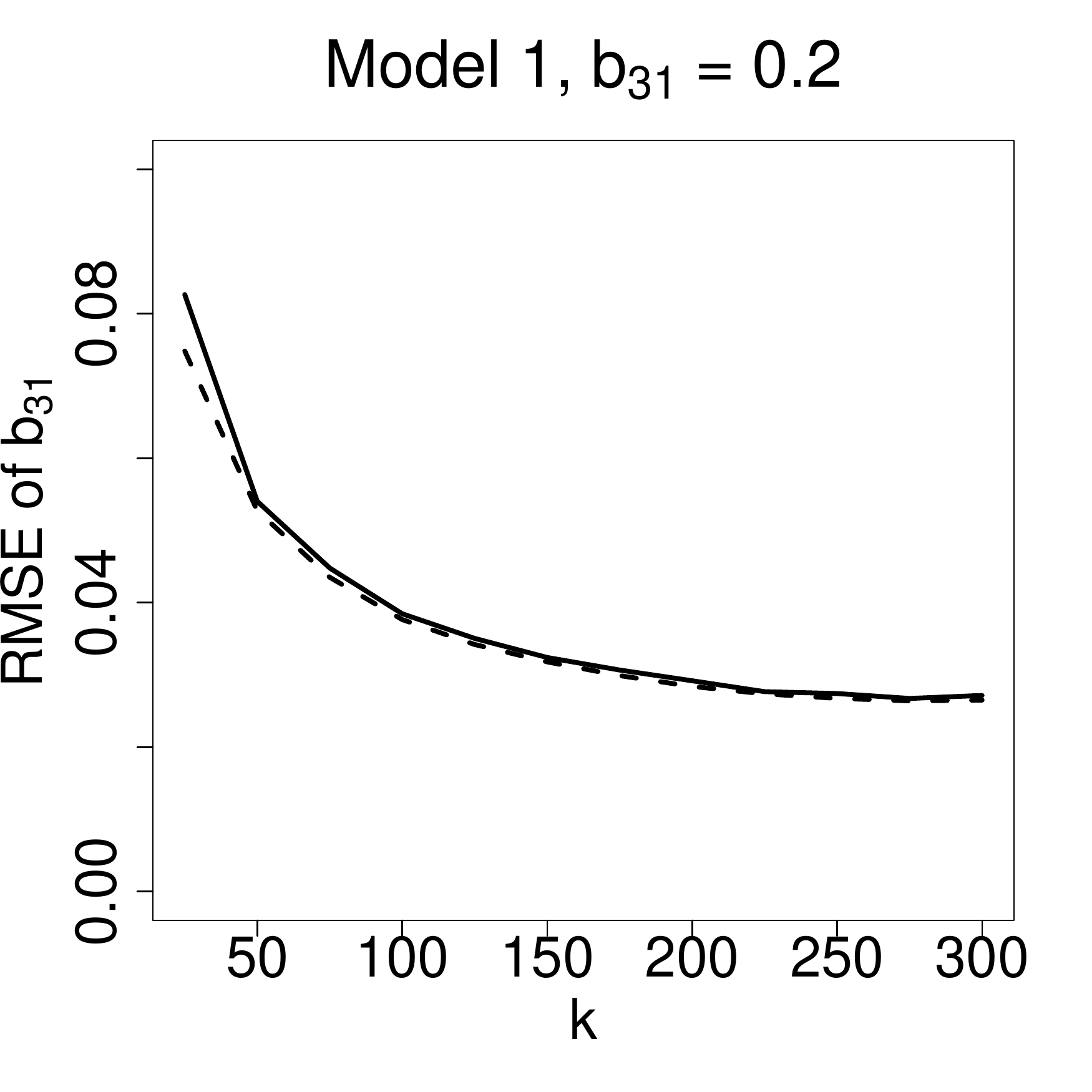}} \\
\subfloat{\includegraphics[width=0.245\textwidth]{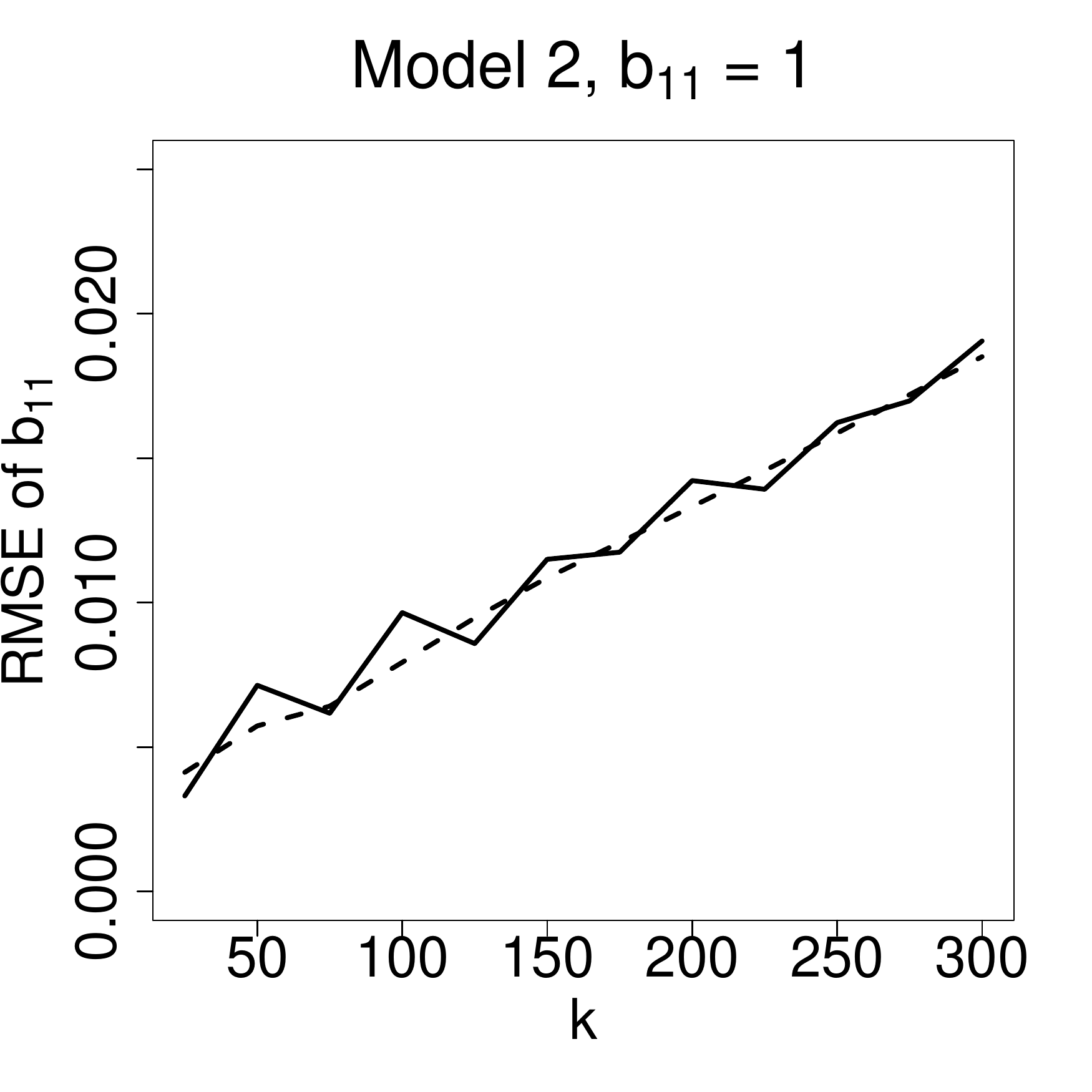}} 
\subfloat{\includegraphics[width=0.245\textwidth]{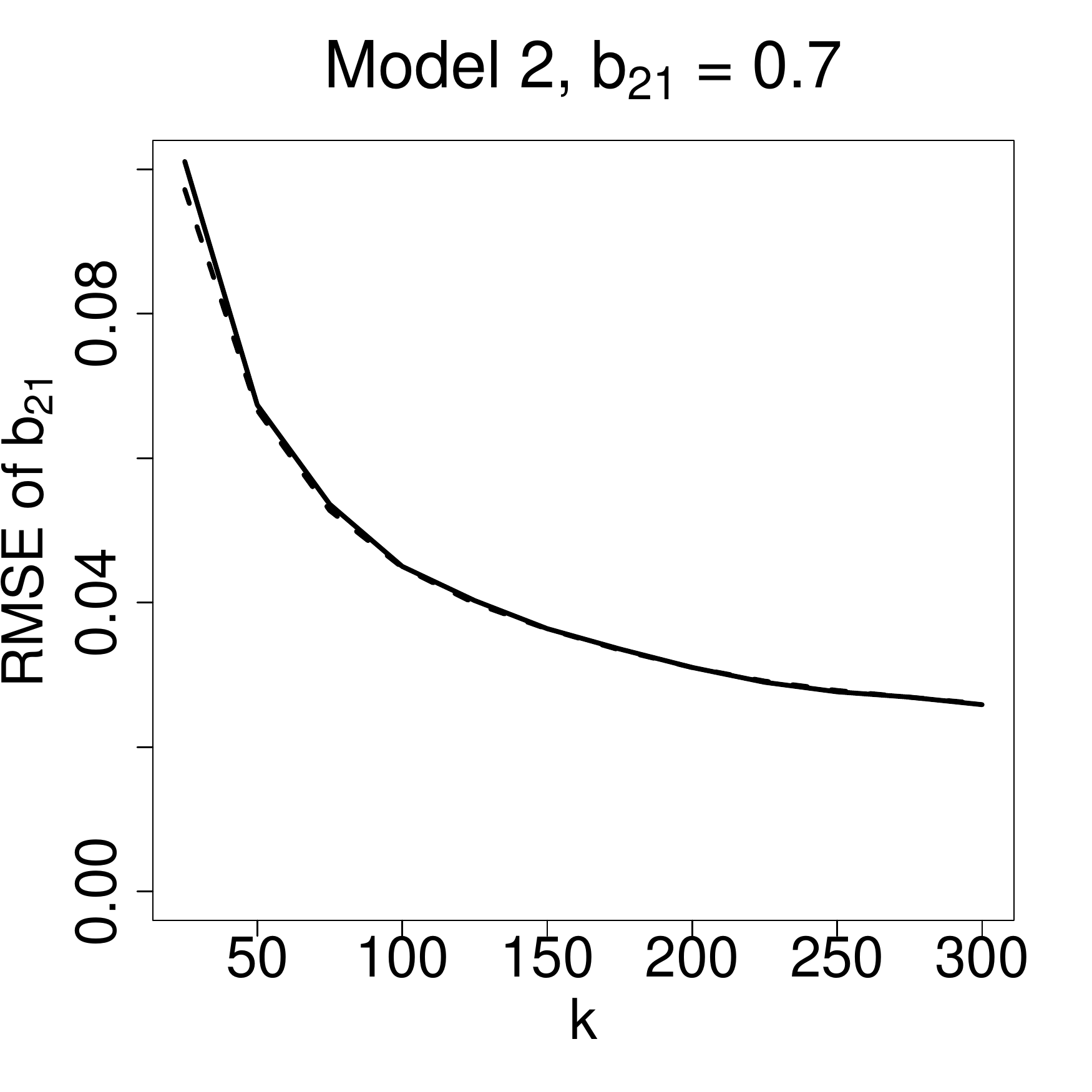}} 
\subfloat{\includegraphics[width=0.245\textwidth]{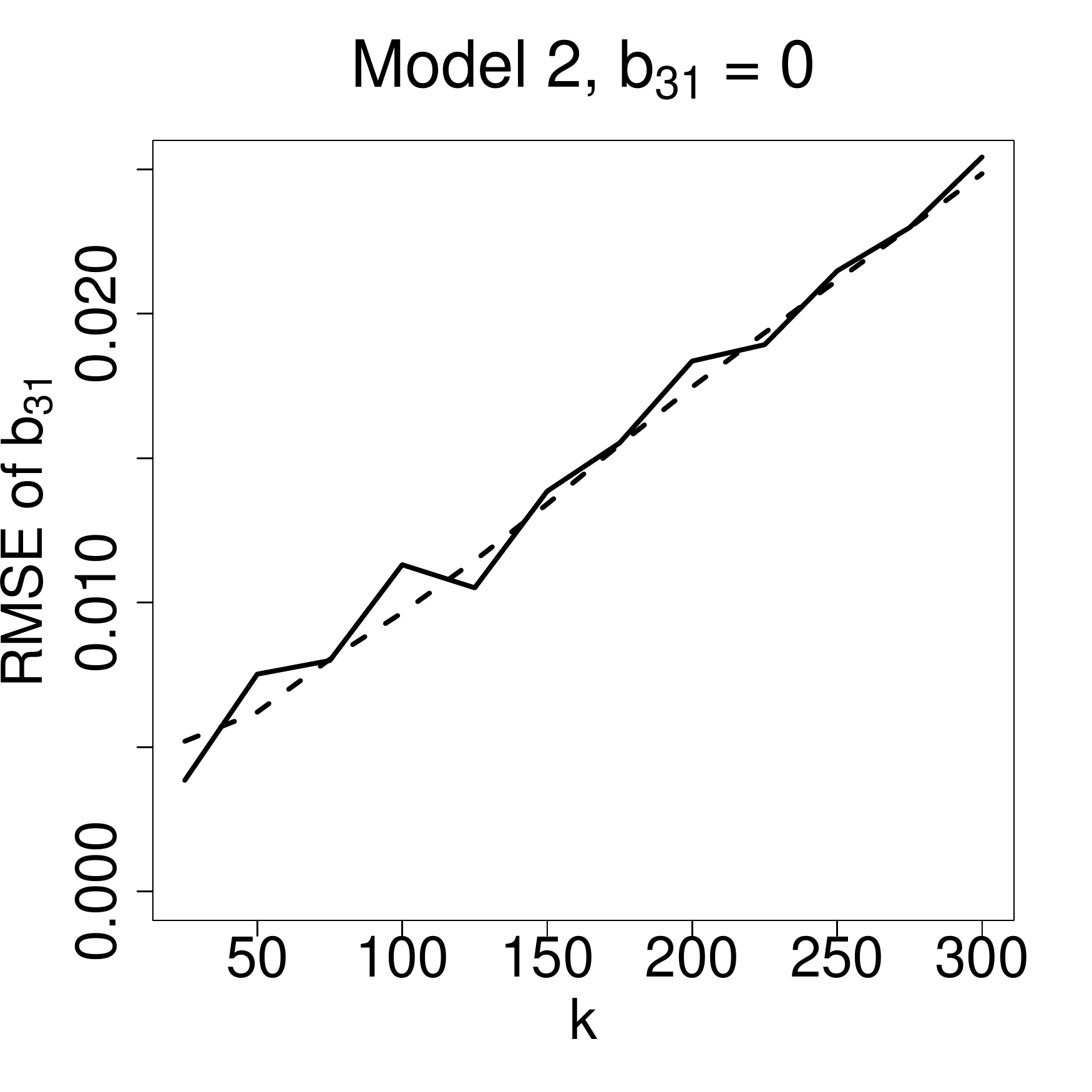}} \\
\subfloat{\includegraphics[width=0.245\textwidth]{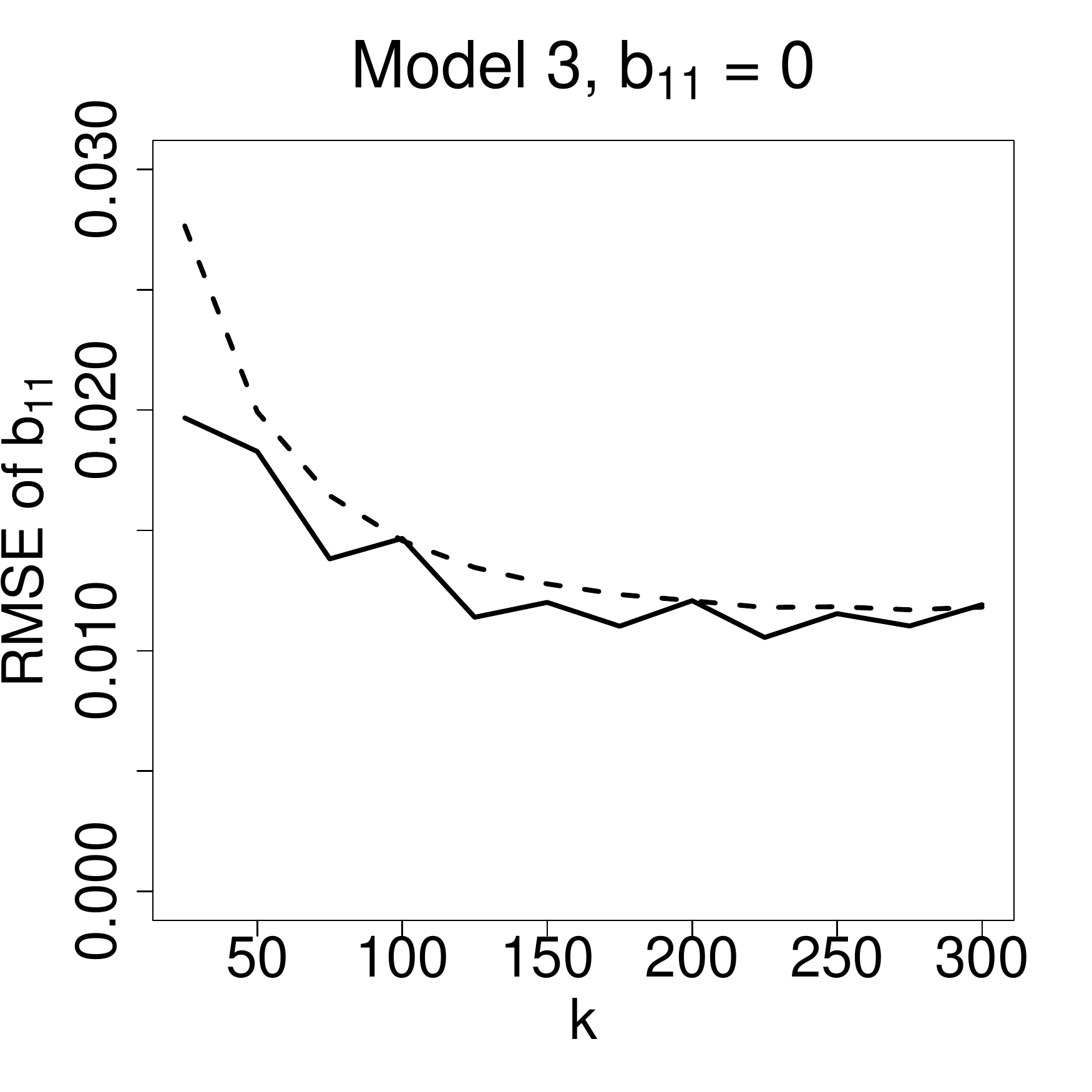}} 
\subfloat{\includegraphics[width=0.245\textwidth]{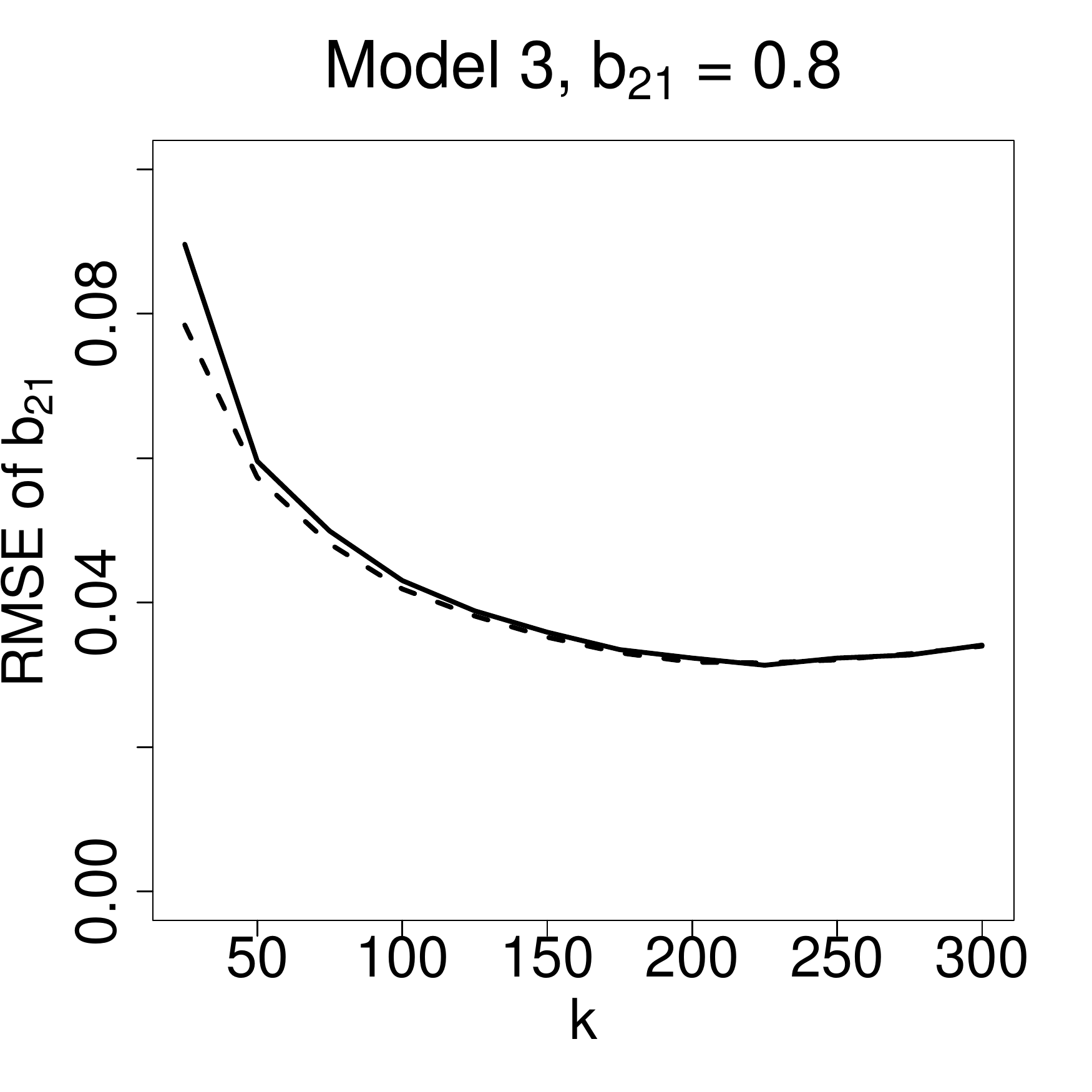}} 
\subfloat{\includegraphics[width=0.245\textwidth]{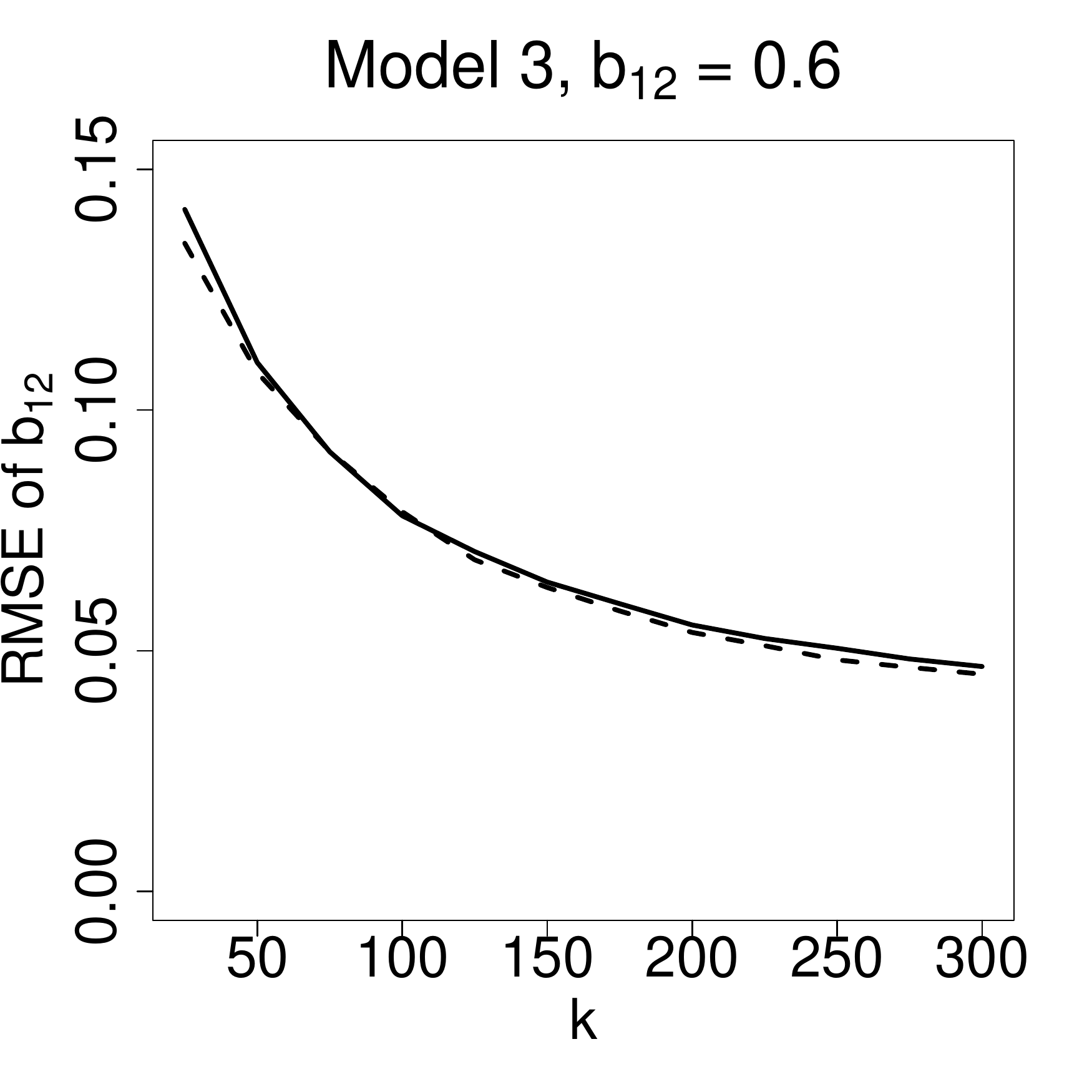}} 
\subfloat{\includegraphics[width=0.245\textwidth]{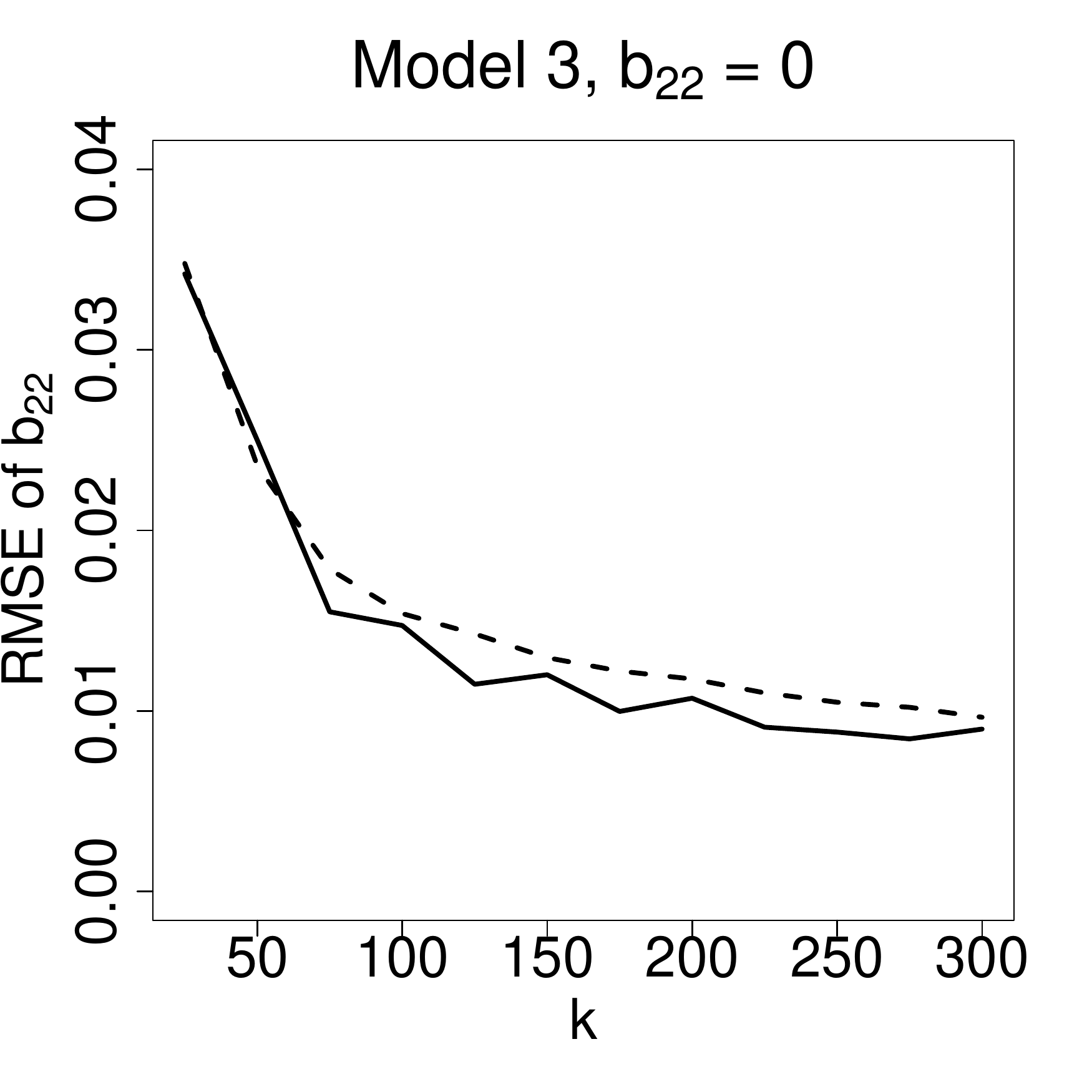}} 
\caption{RMSE for models 1--3 based on the empirical tail dependence function (solid lines), and the beta tail dependence function (dashed lines).} 
\label{fig:ellTot}
\end{figure}

\begin{table}[p]
\centering
 \begin{tabular}{lrrrrrrrr}
 $T_{n,k}^{(1)}$ & \multicolumn{2}{c}{$k = 25$} & \multicolumn{2}{c}{$k = 50$}  & \multicolumn{2}{c}{$k = 75$}  & \multicolumn{2}{c}{$k = 100$}  \\
\cmidrule(r){2-3}
\cmidrule(r){4-5}
\cmidrule(r){6-7}
\cmidrule(r){8-9}
& emp & beta & emp & beta & emp & beta & emp  & beta  \\
\cmidrule(r){2-9}
M1 & $2.9$ & $2.1$ & $6.1$ & $4.2$ & $5.9$ & $5.3$ & $8.0$ & $6.4$ \\
M2 & $0.3$ & $2.3$ & $ 8.7$ & $5.0$ & $12.4$  & $12.6$ & $40.3$  & $27.7$ \\
M3 & $3.3$ & $7.9$ & $7.9$ & $9.7$ & $6.5$  & $9.8$ & $9.8$  & $10.0$ \\
\midrule
$T_{n,k}^{(2)}$ & \multicolumn{2}{c}{$k = 25$} & \multicolumn{2}{c}{$k = 50$}  & \multicolumn{2}{c}{$k = 75$}  & \multicolumn{2}{c}{$k = 100$}  \\
\cmidrule(r){2-3}
\cmidrule(r){4-5}
\cmidrule(r){6-7}
\cmidrule(r){8-9}
& emp & beta & emp & beta & emp & beta & emp  & beta  \\
\cmidrule(r){2-9}
M1 & $3.0$ & $2.9$ & $6.2$ & $5.0$ & $6.1$ & $6.1$ & $8.5$ & $7.1$ \\
M2 & $0.3$ & $2.4$ & $9.2$ & $5.3$ & $12.5$  & $13.0$ & $41.3$  & $28.3$ \\
M3 & $5.4$ & $9.9$ & $10.3$ & $10.2$ & $8.4$  & $10.7$ & $10.7$  & $11.1$ \\
\bottomrule
\end{tabular}
\caption{Empirical level of $T_{n,k}^{(1)}$ and $T_{n,k}^{(2)}$ for models 1--3 (M1--M3) based on the empirical and the beta tail dependence function for a significance level of $0.05$.}
\label{tab:level1}
\end{table}

Figure~\ref{fig:ellTot} shows the RMSE of the parameter estimates of the three max-linear models explained above, based on the empirical tail dependence function (solid lines) and the beta tail dependence function (dashed lines). We see again that these two initial estimators lead to similar results. Parameters whose true values are on the boundary are better estimated for low $k$, while higher values of $k$ are preferred for parameters whose true values are in the interior of the parameter space.

Table \ref{tab:level1} shows the empirical level of the test statistics $T_{n,k}^{(1)}$ and $T_{n,k}^{(2)}$ for $k \in \{25,50,75,100\}$, using a significance level of $0.05$. The tests perform well for model 1 in general, while low values of $k$ are necessary for models 2 and 3. The beta tail dependence function outperforms the empirical tail dependence function for model 2, while the opposite is the case for model 3. For model 3, test statistic $T_{n,k}^{(2)}$ performs better than $T_{n,k}^{(1)}$.

To assess the power of the two test statistics, we consider the models described above, but we let
\begin{align*}
& b_{11}   \in \{0.8,0.85,\ldots,1\} \quad \text{for model 1}, \\
& b_{11}   \in \{0.8,0.85,\ldots 1\}, \, b_{31} \in \{0,0.05,\ldots 0.2 \} \quad  \text{for model 2}, \text{ and}\\
& b_{11}, b_{22}  \in \{0,0.05,\ldots 0.2 \} \quad  \text{for model 3}.
\end{align*}
Figures~\ref{fig:power2} and~\ref{fig:power3} show the empirical power of $T_{n,k}^{(1)}$ for models 1--2 and of $T_{n,k}^{(2)}$ for model 3. The most striking result of Figure 
\ref{fig:power3} is that the graphs are not symmetric, i.e., the power is much higher when $b_{22}$ is near zero (and $b_{11}$ is not) than when $b_{11}$ is near zero (and $b_{22}$ is not).

\begin{figure}[ht]
\centering
\subfloat{\includegraphics[width=0.4\textwidth]{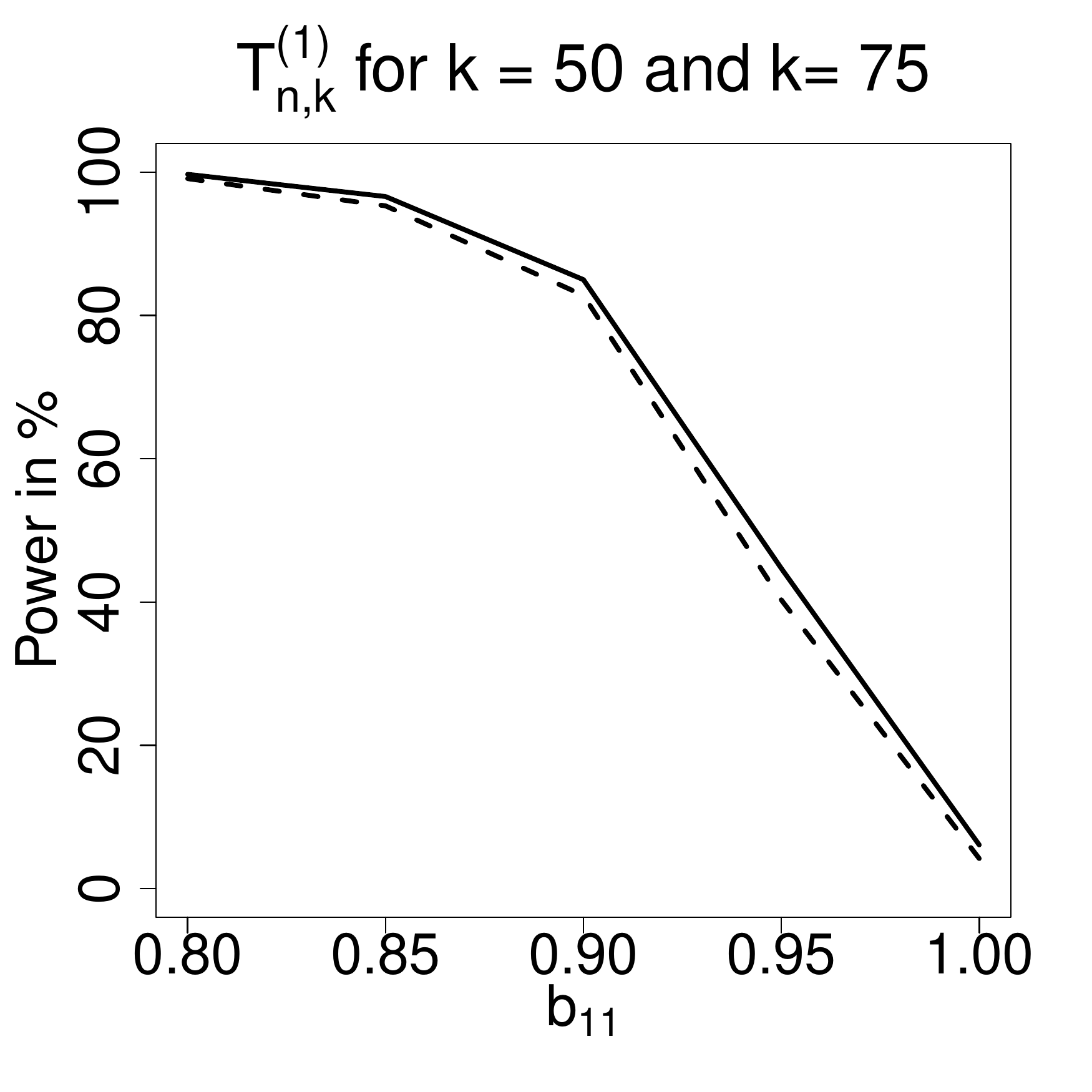}}   \\
\subfloat{\includegraphics[width=0.4\textwidth]{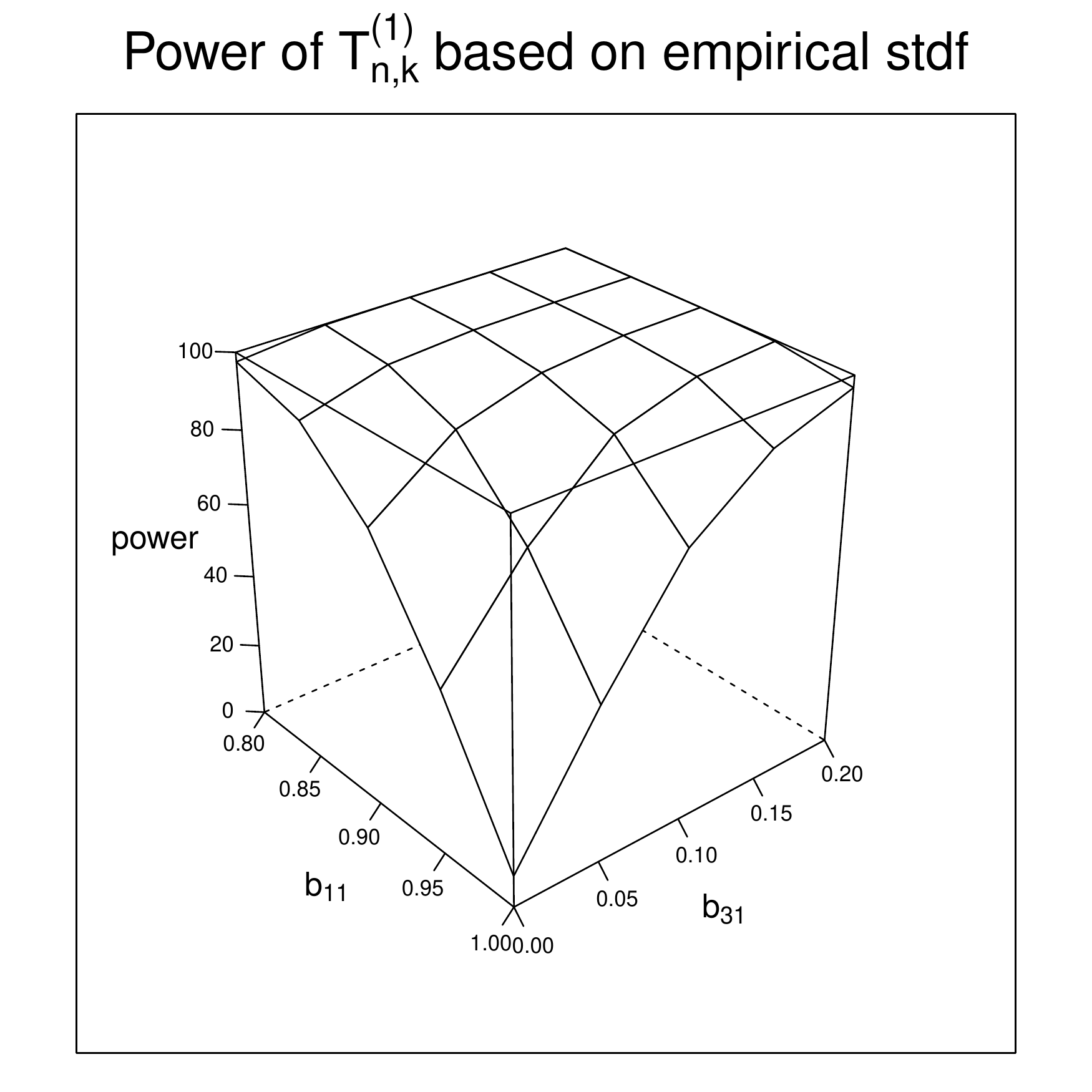}} 
\subfloat{\includegraphics[width=0.4\textwidth]{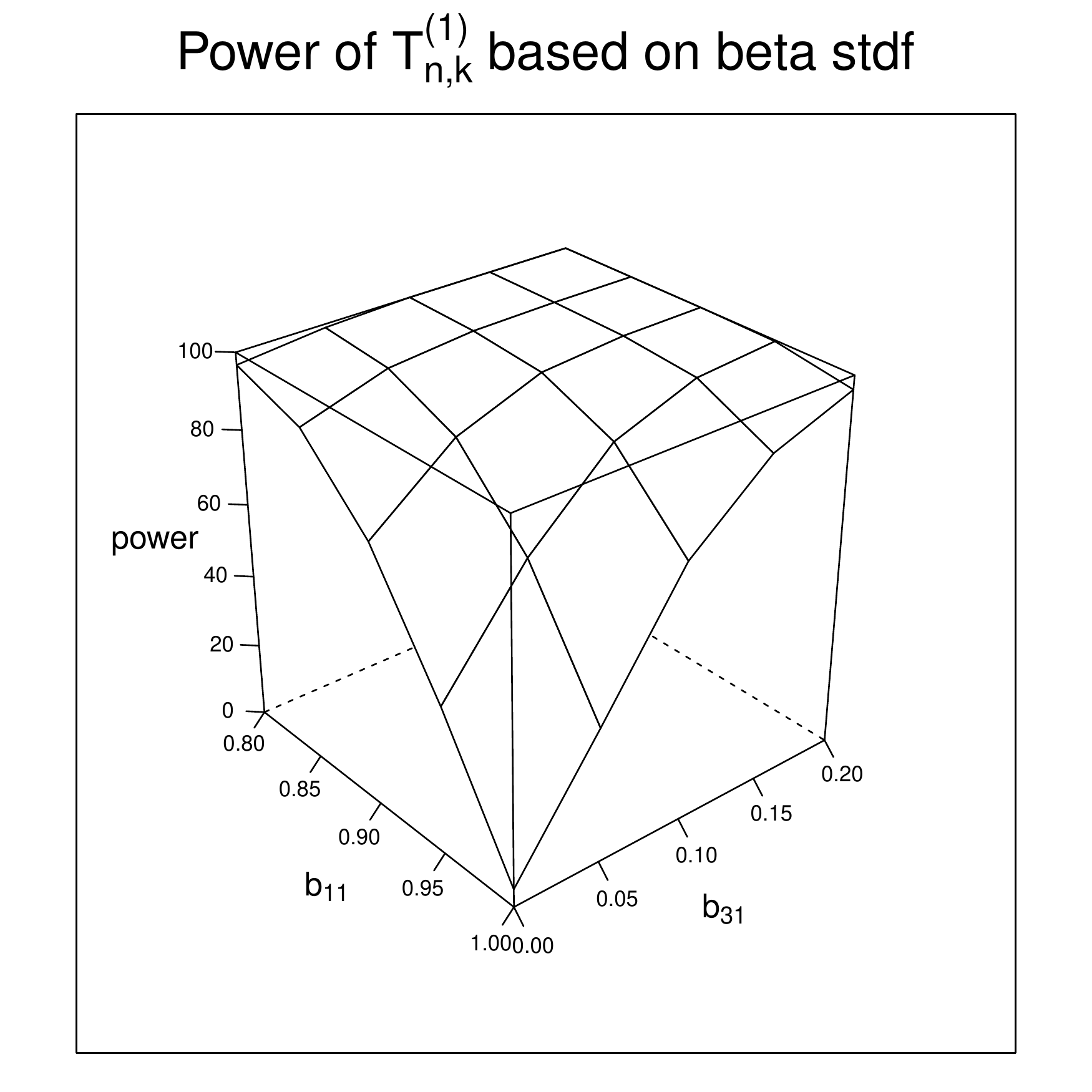}}  
\caption{Top: empirical power in $\%$ of $T_{n,k}^{(1)}$ for model 1 and $k = 50$ (dashed line), $k = 75$ (solid line), based on the empirical tail dependence function. Bottom: empirical power in $\%$ of $T_{n,k}^{(1)}$ for model 2 and $k = 50$ based on the empirical (left) and the beta (right) tail dependence function.} 
\label{fig:power2}
\end{figure}

\clearpage

\begin{figure}[ht]
\centering
\subfloat{\includegraphics[width=0.4\textwidth]{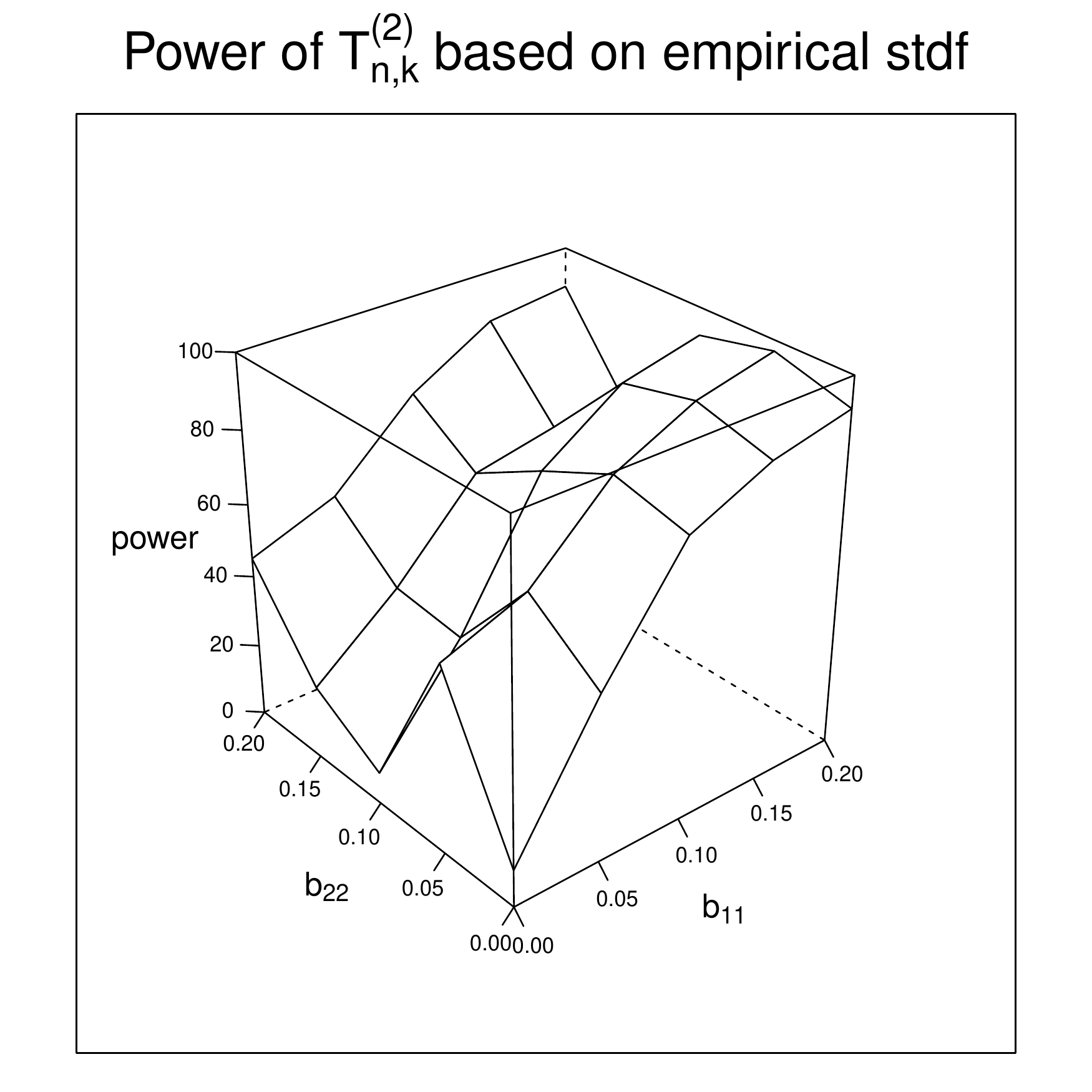}} 
\subfloat{\includegraphics[width=0.4\textwidth]{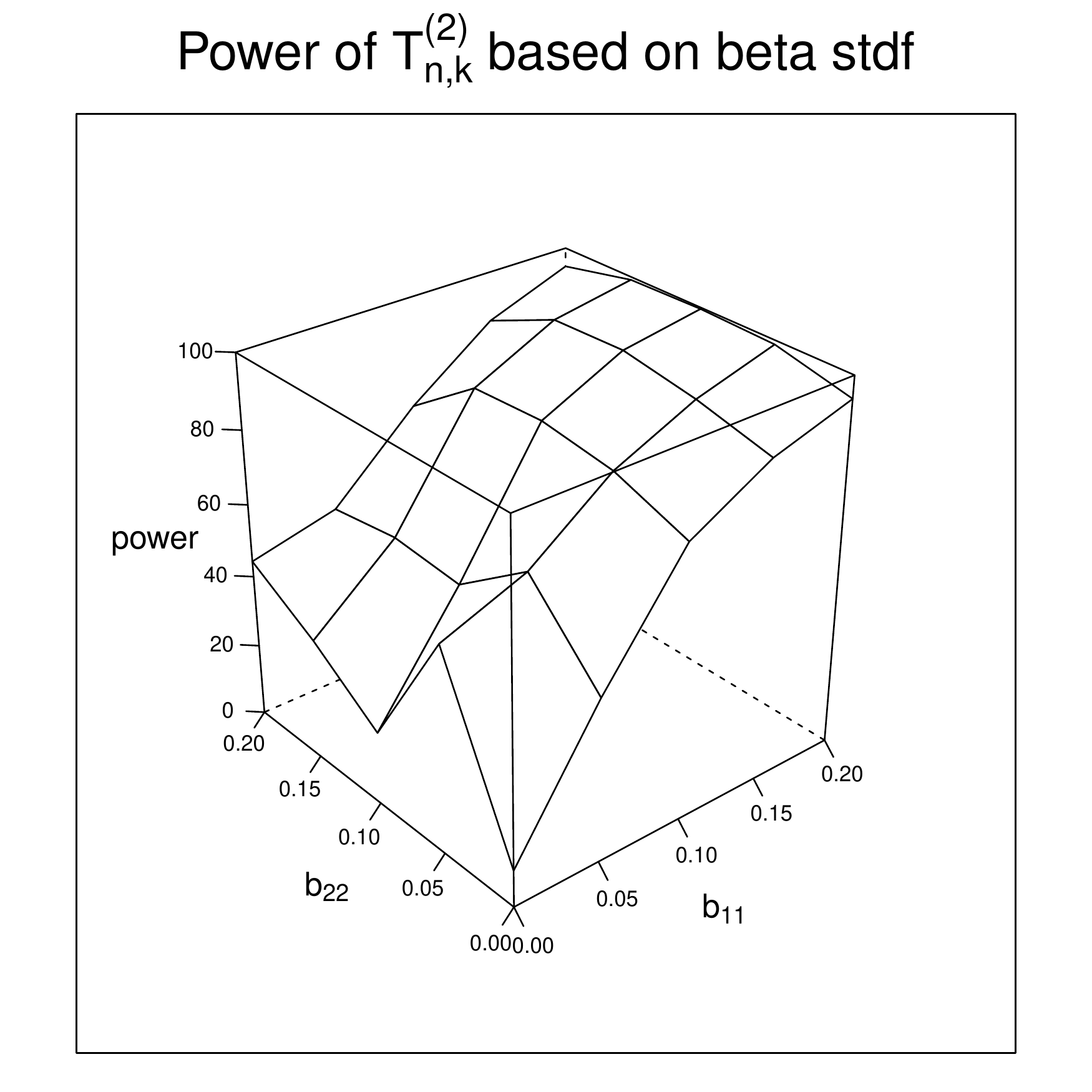}} 
\caption{Empirical power in $\%$ of $T_{n,k}^{(2)}$ for model 3 and $k = 50$, based on the empirical (left) and the beta (right) tail dependence function.} 
\label{fig:power3}
\end{figure}

\subsubsection{Testing the number of factors in a max-linear model}\label{sec:multi}
We studied the performance of the test statistics in case we wanted to identify a specific submodel of the max-linear model. 
We would also like to investigate the capability of the test statistics to correctly retrieve the true number of factors, i.e., for $s \in \{1,\ldots,r\}$, we wish to test the hypothesis $H_0 : (b_{1s},\ldots,b_{ds}) = \vc{0}$. However, as mentioned in Section \ref{sec:EVT}, the max-linear model is only defined for parameter values with $\sum_{j=1}^d b_{js} > 0$ for all $s \in \{1,\ldots,d\}$; when this condition is not met, there are zeroes on the diagonal of $\Sigma (\vtheta_0)$, making computation of $\widehat{\vc{\lambda}}_{\beta}$ impossible.
We solve this problem by computing the asymptotic distribution of the test statistics using $\widehat{\vtheta}_{n,k}$ rather than $\vtheta_0$. We consider a model with $d = 2$ and $r = 3$, with
parameter vector $\vtheta_0 = (b_{11},b_{21},b_{12},b_{22}) = (0.8,0.6,0.2,0.4)$, i.e., the third column of $B$ contains only zeroes and the model has effectively two factors. We estimate a three-factor model and we test $H_0 : (b_{13},b_{23}) = (0,0)$. 

Figure~\ref{fig:ellTot2} shows the RMSE of the parameter estimates based on the empirical tail dependence function (solid lines) and the beta tail dependence function (dashed lines). We see that, contrary to previous experiments, the beta tail dependence has much higher RMSE than the empirical tail dependence function: oversmoothing leads to a large bias for $(b_{13},b_{23})$ (and hence for the other parameters as well).

\begin{figure}[p]
\centering
\subfloat{\includegraphics[width=0.245\textwidth]{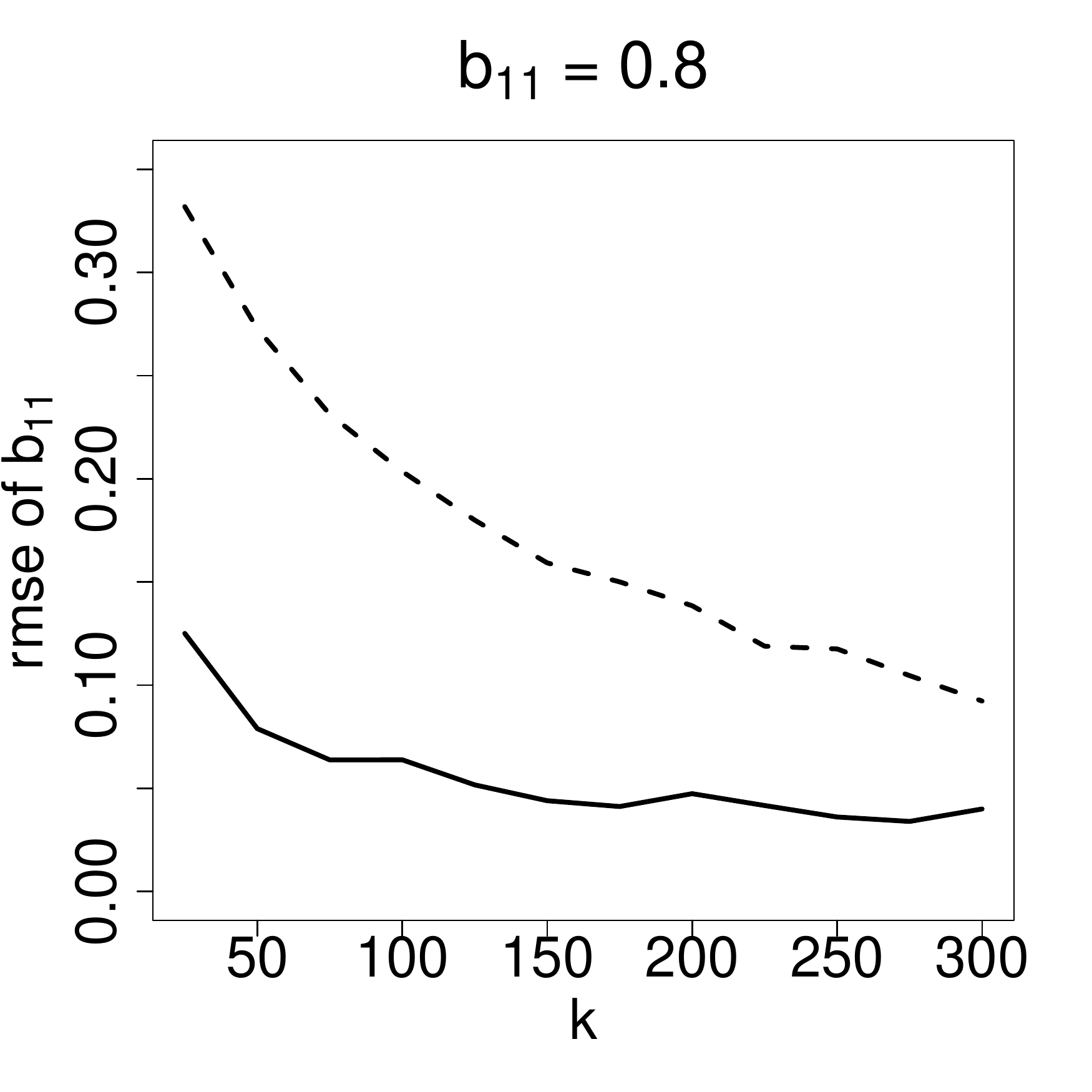}} 
\subfloat{\includegraphics[width=0.245\textwidth]{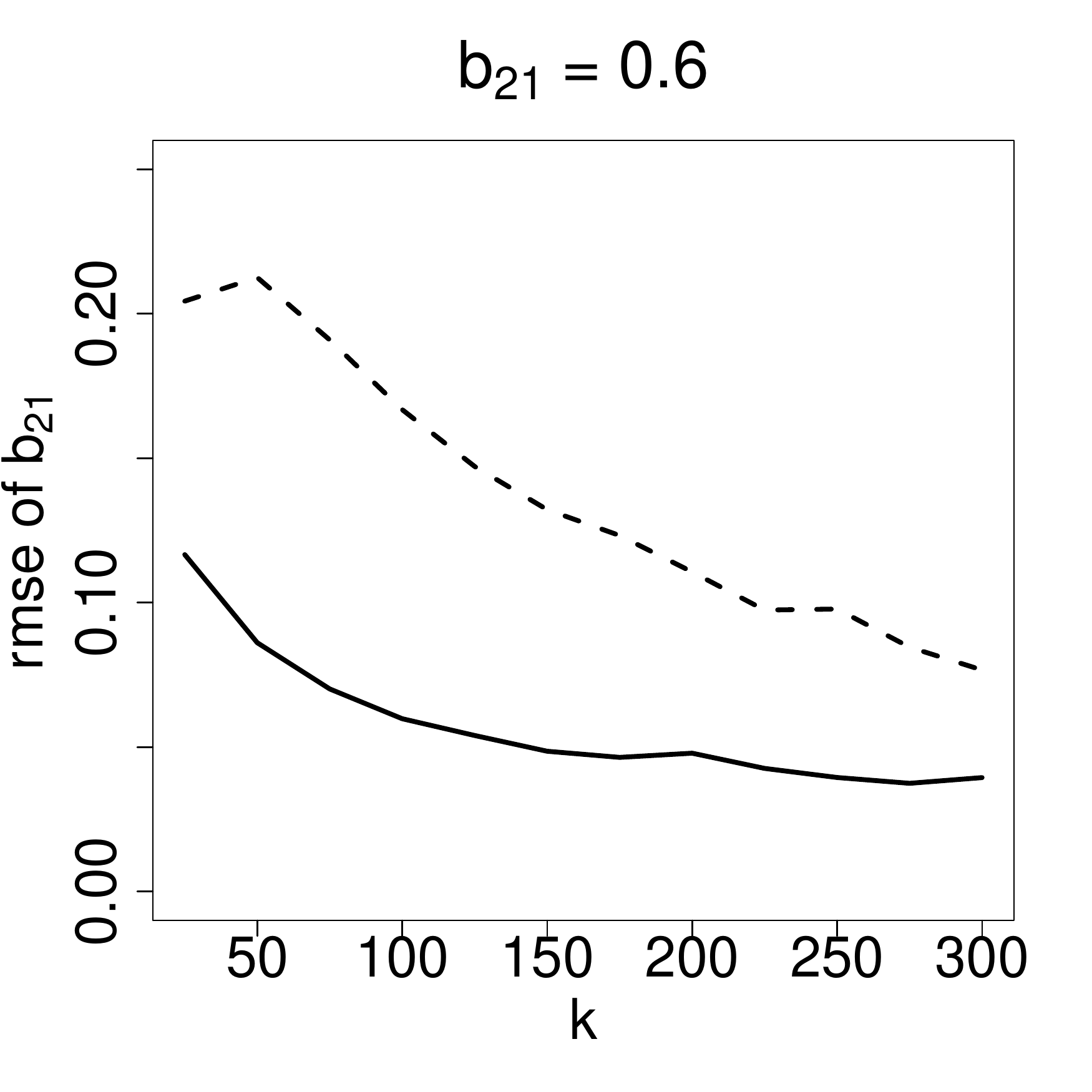}} 
\subfloat{\includegraphics[width=0.245\textwidth]{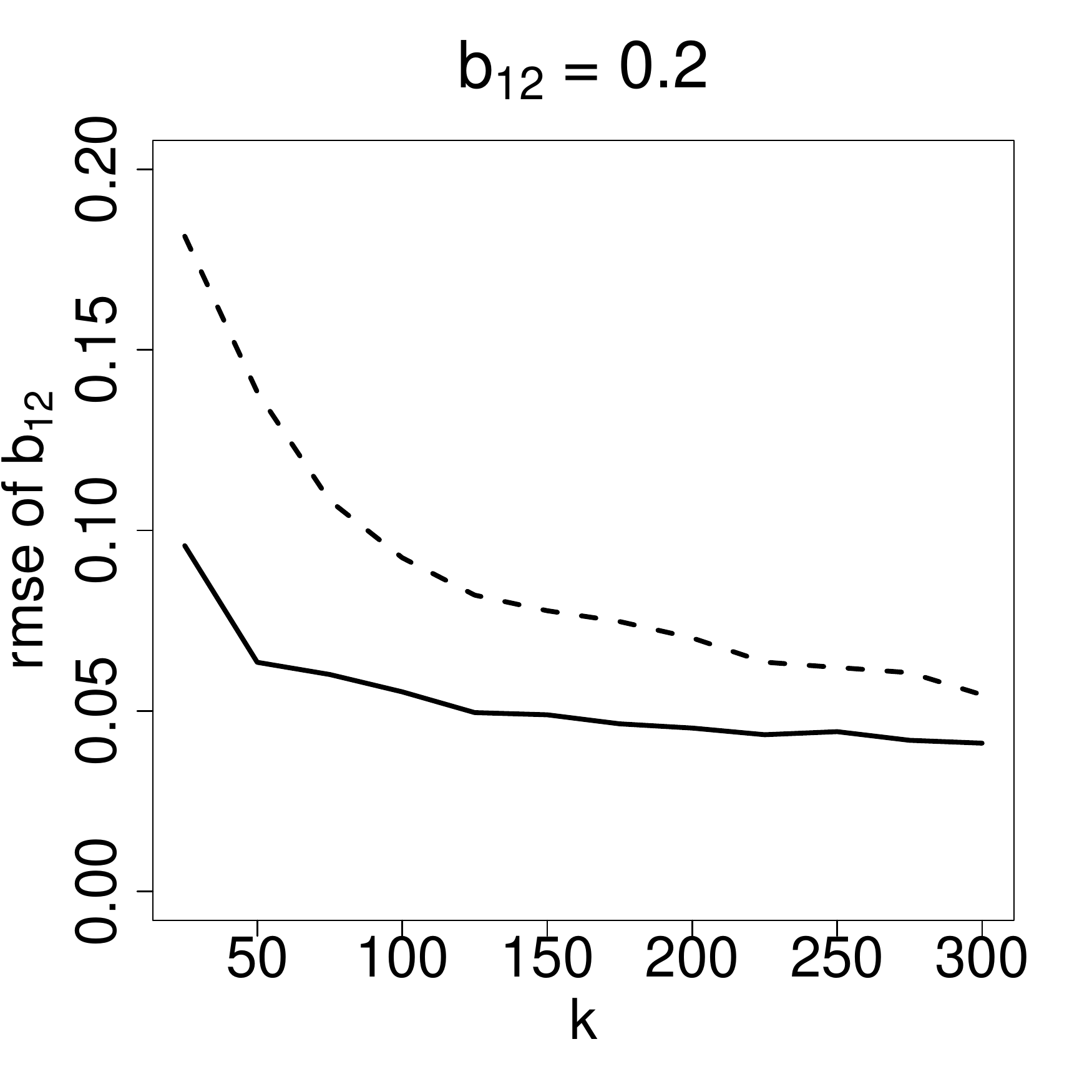}} 
\subfloat{\includegraphics[width=0.245\textwidth]{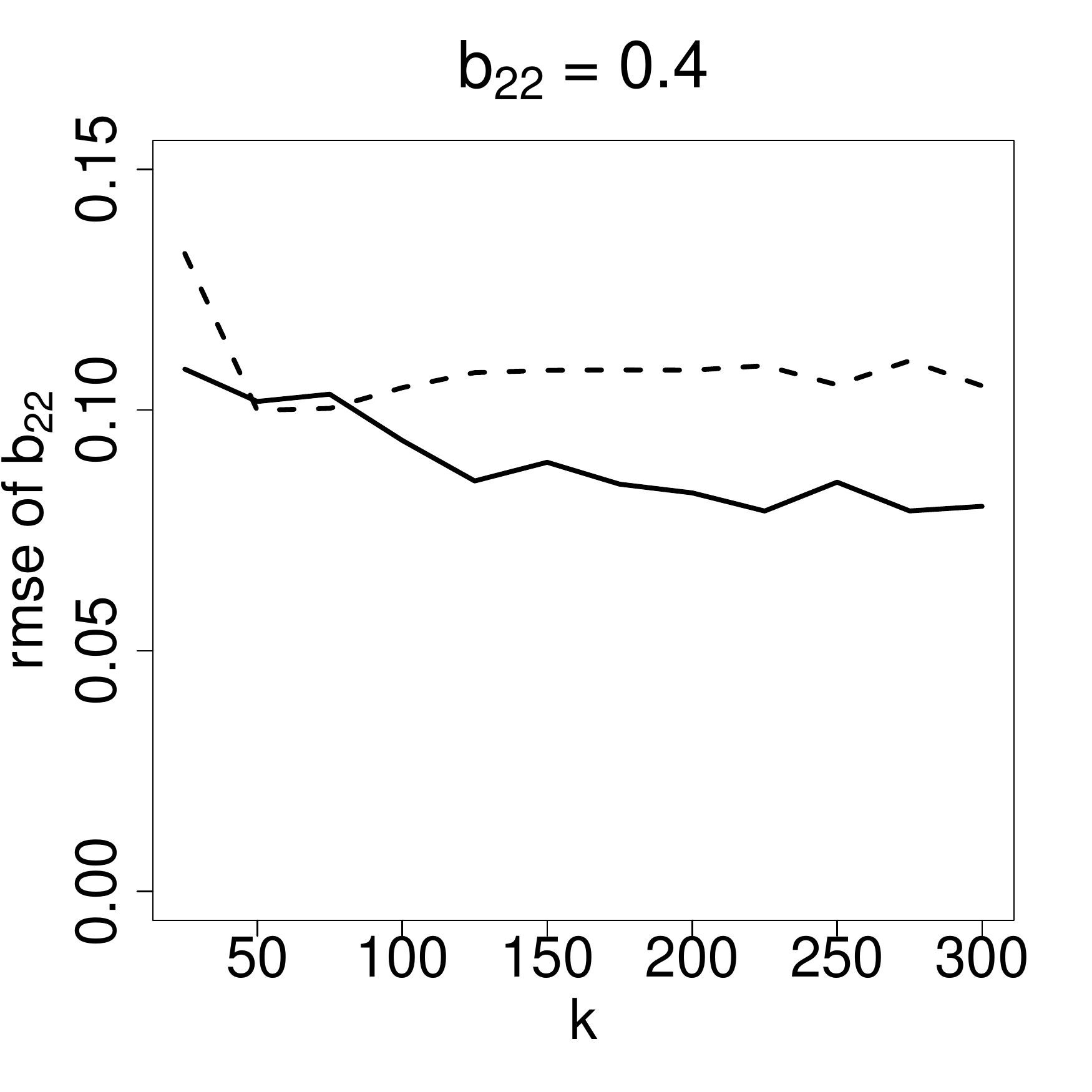}} 
\caption{RMSE based on the empirical tail dependence function (solid lines) and the beta tail dependence function (dashed lines).} 
\label{fig:ellTot2}
\end{figure}

\begin{table}[p]
\centering
 \begin{tabular}{lrrrrrrrr}
 & \multicolumn{2}{c}{$k = 25$} & \multicolumn{2}{c}{$k = 50$}  & \multicolumn{2}{c}{$k = 75$}  & \multicolumn{2}{c}{$k = 100$}  \\
\cmidrule(r){2-3}
\cmidrule(r){4-5}
\cmidrule(r){6-7}
\cmidrule(r){8-9}
& emp & beta & emp & beta & emp & beta & emp  & beta  \\
\cmidrule(r){2-9}
$T_{n,k}^{(2)}$ & $4.6$ & $23.9$ & $2.7$ & $24.9$ & $2.4$ & $23.8$ & $2.5$ & $20.0$ \\
\bottomrule
\end{tabular}
\caption{Empirical level of $T_{n,k}^{(2)}$ based on the empirical and the beta tail dependence function for a significance level of $0.05$.}
\label{tab:level2}
\end{table}

\begin{figure}[p]
\centering
\subfloat{\includegraphics[width=0.4\textwidth]{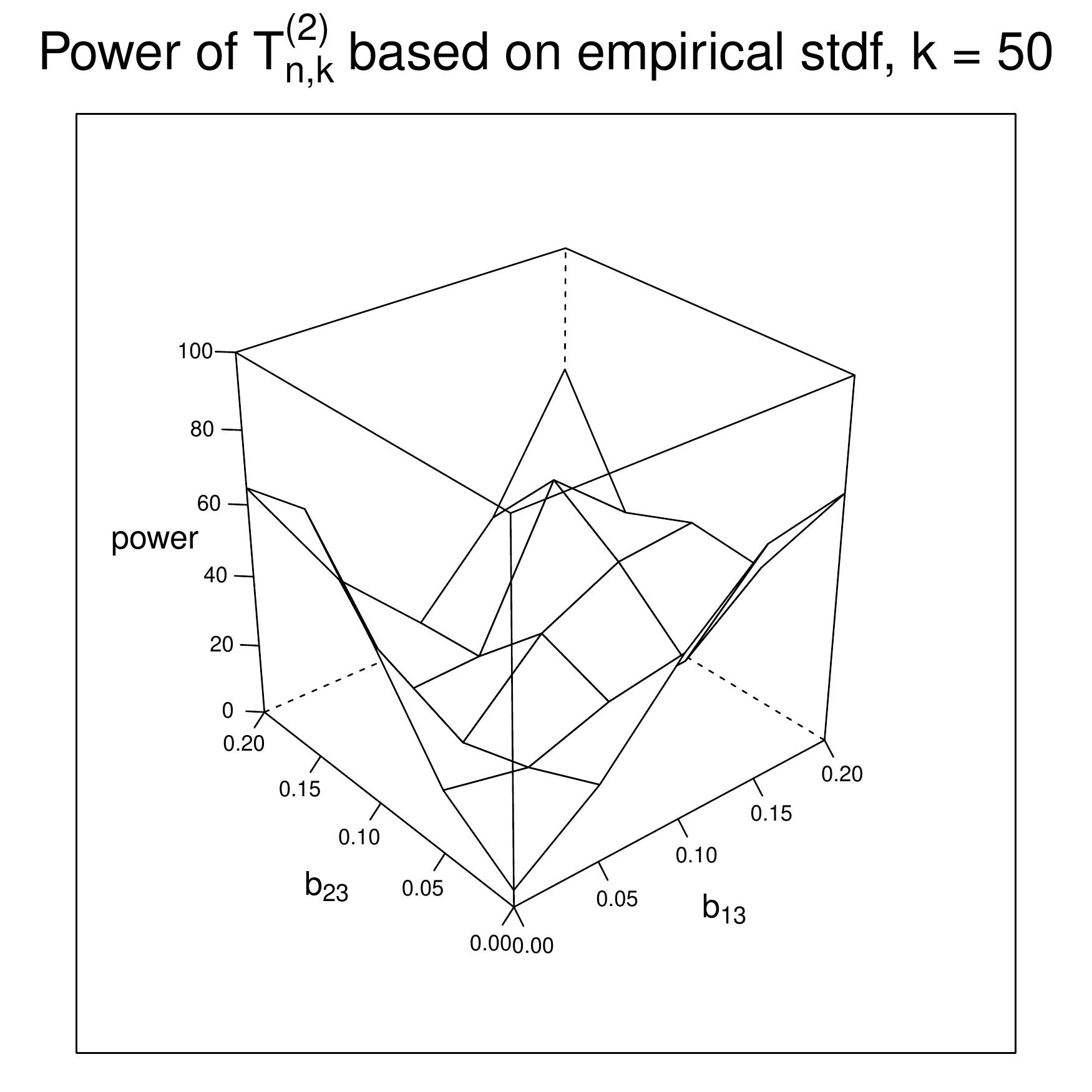}} 
\subfloat{\includegraphics[width=0.4\textwidth]{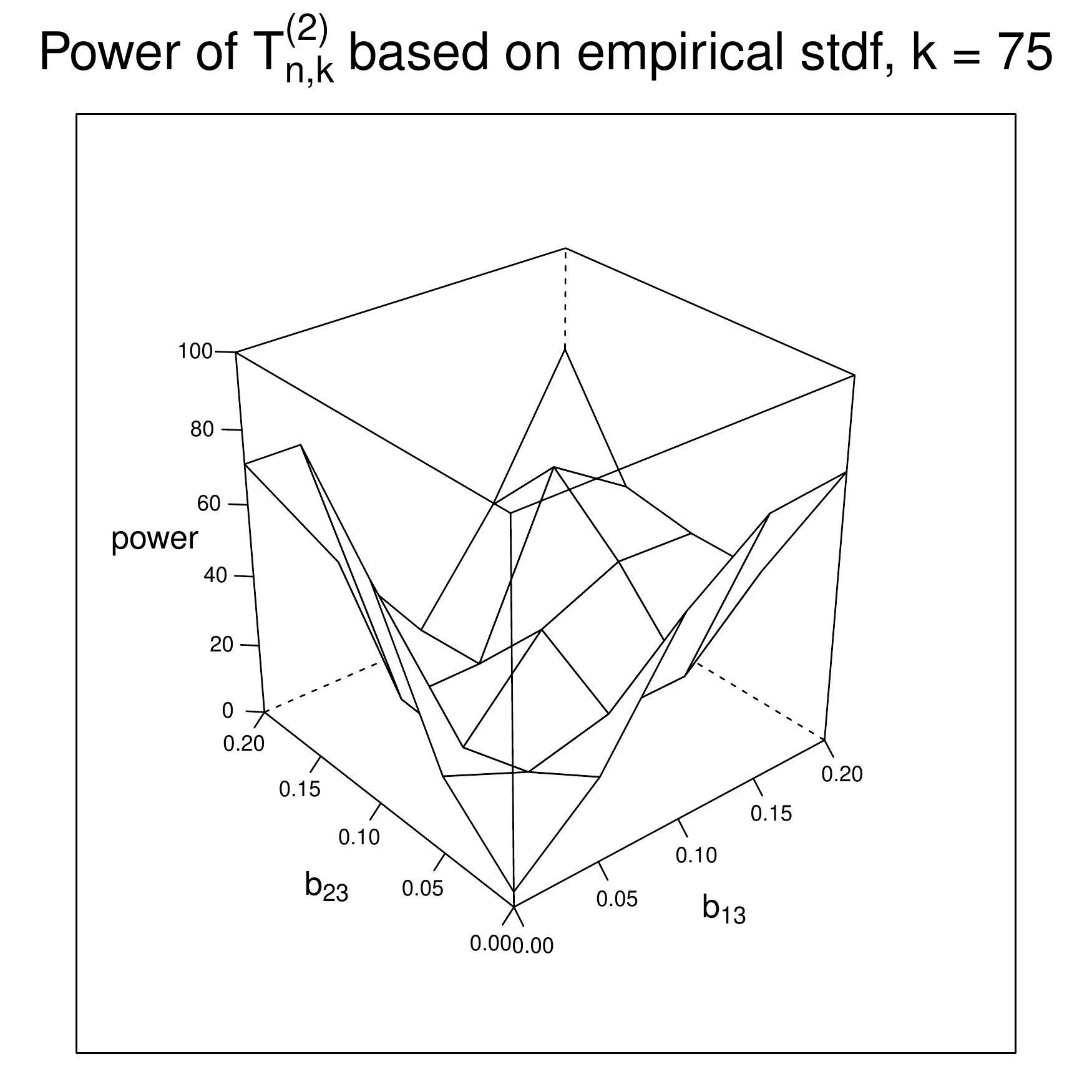}}  
\caption{Empirical power in $\%$ of $T_{n,k}^{(2)}$, based on the empirical tail dependence function for $k = 50$ (left) and $k = 75$ (right).} 
\label{fig:power4}
\end{figure}

Table~\ref{tab:level2} shows the empirical level of the test statistic $T_{n,k}^{(2)}$ using a significance level of $0.05$. We included the results based on the beta tail dependence function for completeness: because it tends to overestimate $(b_{13},b_{23})$, it rejects the null hypothesis far too often. Test statistic $T_{n,k}^{(1)}$ is not presented because for any $k$ and any initial estimator, it has an empirical level of $0$. Even when $| \widehat{\vtheta}_{n,k} - \widehat{\vtheta}_{n,k}^{(0)}|$ is large, $| f_{n, k} (\widehat{\vtheta}_{n,k} )- f_{n,k} (\widehat{\vtheta}_{n,k}^{(0)})|$ is small and hence the deviance-type test is not adapted to this type of model. Finally, Figure~\ref{fig:power4} shows the empirical power of $T_{n,k}^{(2)}$ for $b_{11}, b_{22}  \in \{0,0.05,\ldots 0.2 \}$ for two values of $k$. We remark that the power is similar when one parameter is far from its boundary and when both parameters are far from their boundary.

\section{Application to stock market indices}\label{sec:application}
Consider two major European stock market indices, the German DAX and the French CAC40. We take the daily negative log-returns of the prices of these two indices from \url{https://finance.yahoo.com} for the period of January 1st, 1997 to December 31st, 2017. We remove all dates for which at least one of the two series has missing values, ending up with a sample of size $n  = 5313$. Figure~\ref{fig:stocks} shows the time series plots of log-returns and the dependence structure for the returns standardized to unit Pareto margins, plotted on the exponential scale. The time series plots look stationary: we perform a Augmented Dickey--Fuller test, which gives $p$-values below $0.01$ for both series, and a KPSS test, giving $p$-values above $0.1$ for both series. Hence, the assumption of stationarity is reasonable.
Let $u_j$ denote the $0.95 \%$ quantile of $X_{1j},\ldots,X_{nj}$ for $j \in \{1,\ldots,d \}$. We fit a generalized Pareto distribution (GPD) to $X_{ij} - u_j \mid X_{ij} > u_j$ and obtain the parameter estimates $\widehat{\sigma}_1 = 1.12$ $(0.11)$, $\widehat{\gamma}_1 = 0.03$ $(0.08)$,  $\widehat{\sigma}_2 = 1.10$ $(0.10)$ and $\widehat{\gamma}_2 = 0.02$ $(0.07)$. 

\begin{figure}[ht]
\centering
\subfloat{\includegraphics[width=0.33\textwidth]{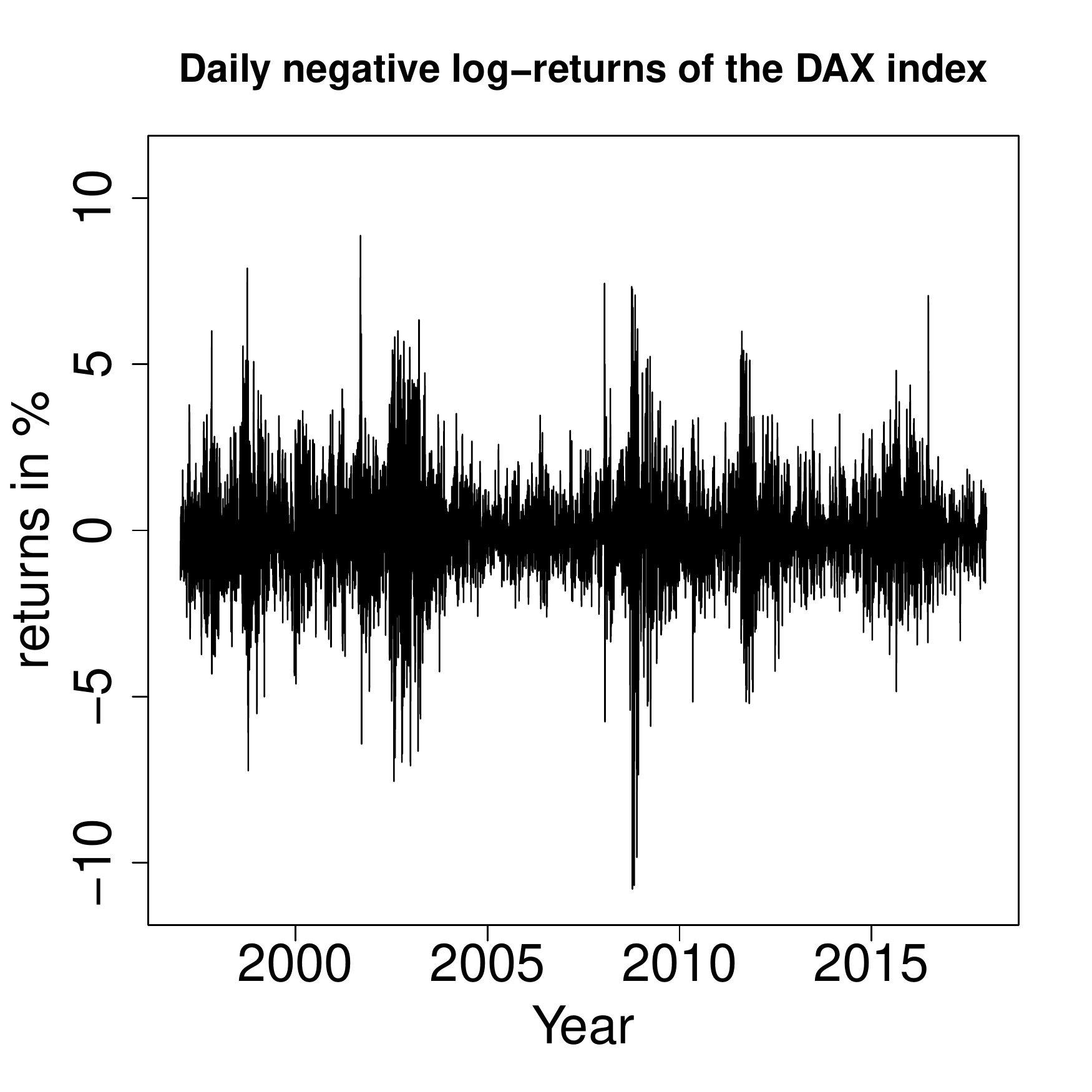}} 
\subfloat{\includegraphics[width=0.33\textwidth]{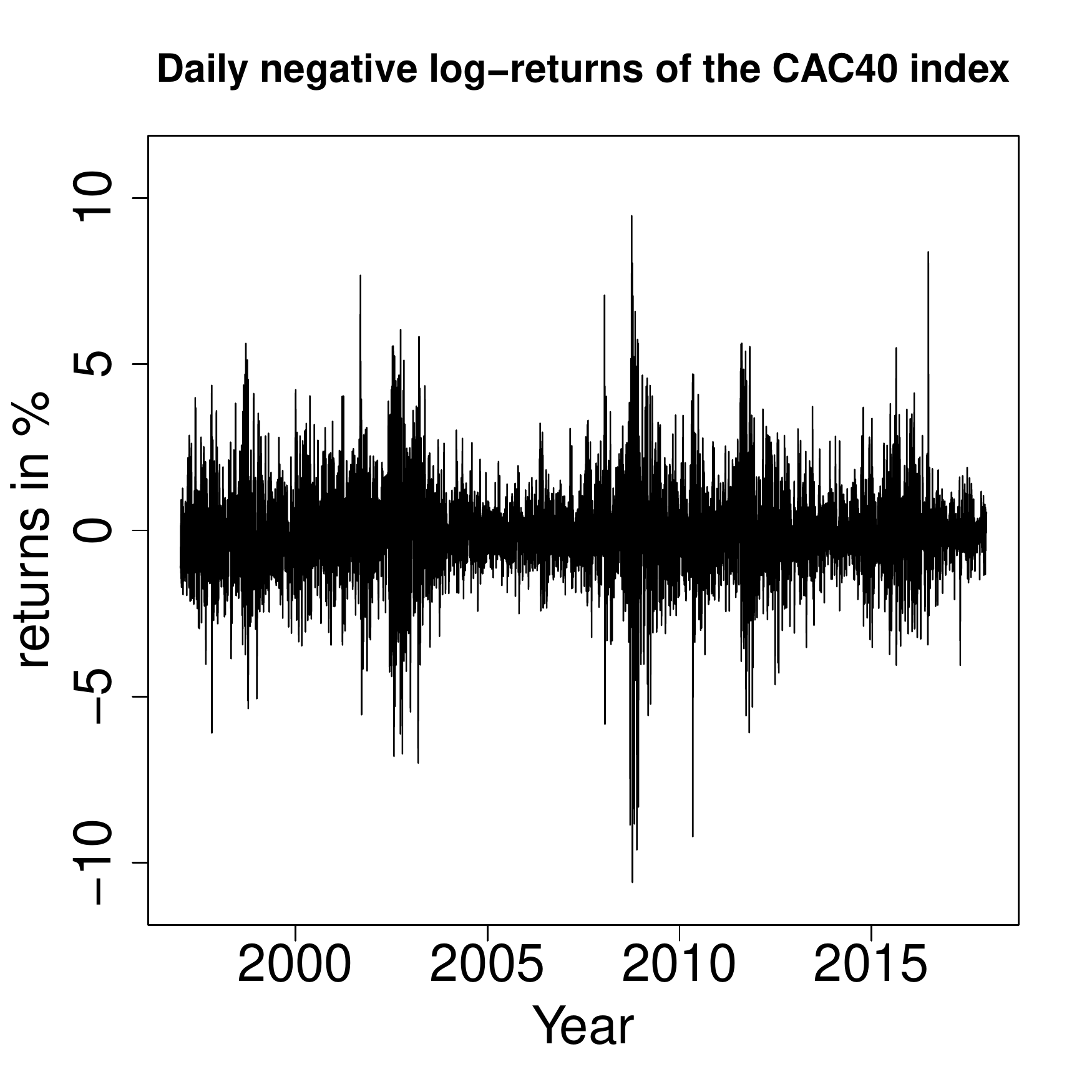}} 
\subfloat{\includegraphics[width=0.33\textwidth]{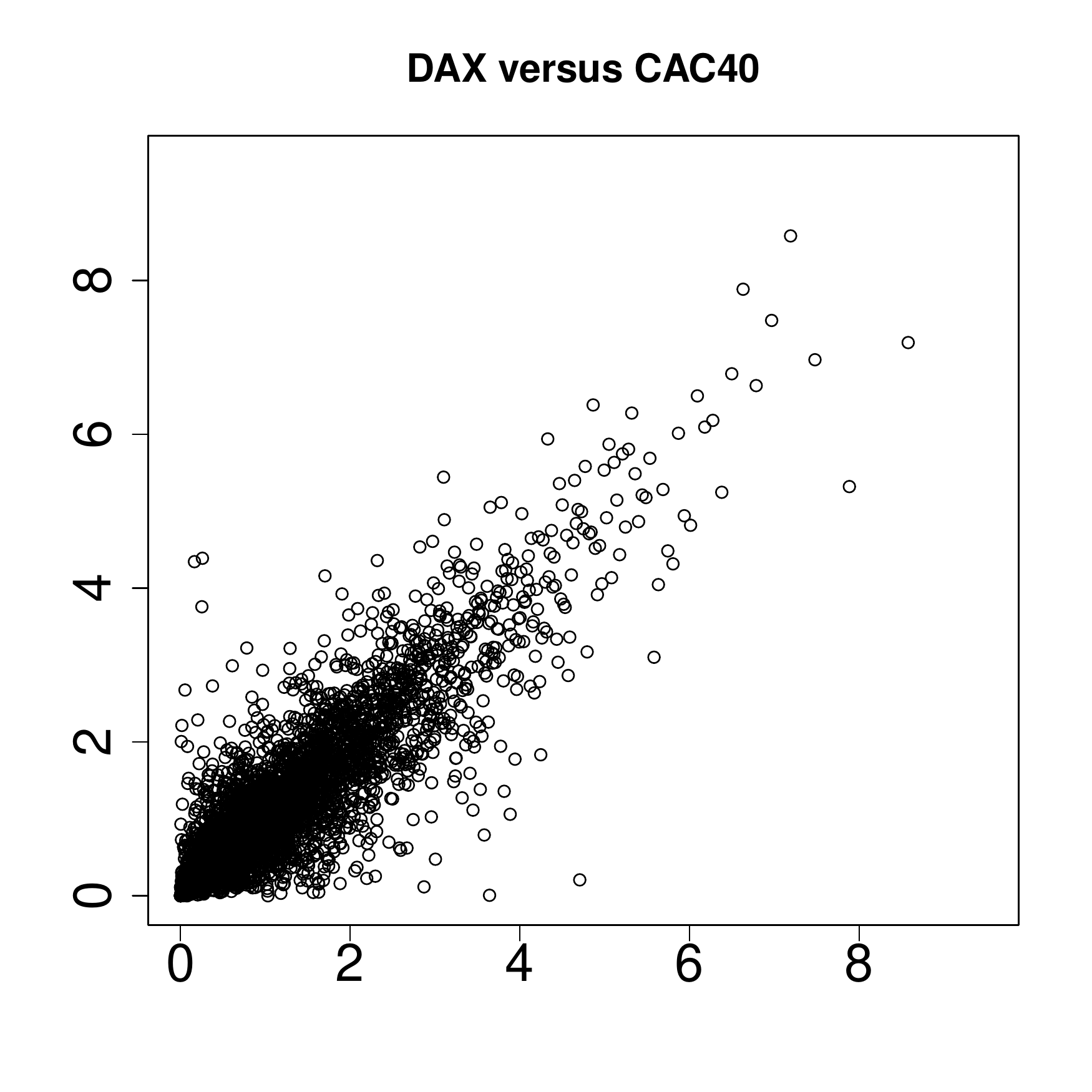}} 
\caption{Left and middle: time series plots of daily negative log-returns of the DAX and the CAC40; right: scatterplot of daily negative log-returns of the DAX versus the CAC40 on the unit Pareto scale, plotted on the exponential scale.} 
\label{fig:stocks}
\end{figure}

Table~\ref{tab:r4} shows the estimated parameter matrix $B = (b_{jt})_{j,t}$ for $j=1,2$ and $t = 1,2,3,4$ obtained by fitting a four-factor max-linear model to our data based on the empirical tail dependence function. Values of $k \in \{25,30,\ldots,70,75\}$ were considered: we chose $k = 40$ because parameter estimates are stable around this value. Standard errors were calculated using the asymptotic variance matrix of the estimator. 
The results in Table~\ref{tab:r4} suggest to test $H_0: (b_{14},b_{24}) = (0,0)$; the value of the test statistic $T_{n,k}^{(2)}$ is $0.04$. Comparing with a critical value of $17.4$ based on a significance level of $0.05$, we find that we cannot reject the null hypothesis.  Critical values are calculated by simulation from the asymptotic distribution of the test statistics, which is given in Corollaries \ref{cor2} and \ref{cor3}. We do not consider $T_{n,k}^{(1)}$ or the beta tail dependence function because of their bad performance (see Section \ref{sec:multi}).

\begin{table}[ht]
\centering
 \begin{tabular}{lcccc}
\toprule
DAX &  0.41 (0.11) & 0.14 (0.05) & 0.43 (0.11) & 0.03 (0.05)  \\
CAC40 & 0.43 (0.12) & 0.44 (0.10) & 0.13 (0.05)  & 0.00 (0.01)  \\
\bottomrule
\end{tabular}
\caption{Parameter matrix for a bivariate max-linear model with $r = 4$ factors for $k = 40$; standard errors are in parentheses.}
\label{tab:r4}
\end{table}

Table~\ref{tab:r3} shows the estimated parameter matrix $B = (b_{jt})_{j,t}$ for $j=1,2$ and $t = 1,2,3$ for a three-factor model. We test if a Marshall--Olkin model suffices, i.e., $H_0 : (b_{12},b_{23}) = (0,0)$. The value of the test statistic $T_{n,k}^{(2)}$ is $75.9$. Comparing to a critical value of $8.03$, we reject the null hypothesis of a Marshall--Olkin model.

\begin{table}[ht]
\centering
 \begin{tabular}{lccc}
\toprule
DAX & 0.41 (0.11) & 0.14 (0.06) & 0.46 (0.09)   \\
CAC40 & 0.43 (0.12) & 0.44 (0.10) & 0.13 (0.05)  \\
\bottomrule
\end{tabular}
\caption{Parameter matrix for a bivariate max-linear model with $r = 3$ factors for $k = 40$; standard errors are in parentheses.}
\label{tab:r3}
\end{table}

Figure \ref{fig:stocks2} shows two goodness-of-fit measures and the estimated probability of a joint exceedances of the two stocks. On the left, we plotted the level sets $\{(x_1,x_2): \ell (x_1,x_2) = c \}$ for $c \in \{0.2,0.4,0.8,1\}$ based on the empirical tail dependence function (solid lines) and on the fitted max-linear model (dashed lines). 

A common summary measure of dependence is the tail dependence coefficient,
\begin{equation*}
\chi (q) := \lim_{q \uparrow 1} \mathbb{P} \left[ F_1 (X_{11}) < q \mid  F_2 (X_{12}) < q \right] = 2 - \ell(1,1),
\end{equation*} 
see \eqref{eq:ell2}.  Figure \ref{fig:stocks2} (middle) shows nonparametric estimates of the tail dependence coefficient (black dots) $\widetilde{\chi}_{n,k} = 2 - \widetilde{\ell}_{n,k} (1,1)$
and its model-based counterpart (red dots) for decreasing $k$. The horizontal grey line corresponds to the value $k=40$ that was used for the parameter estimators and the tests, corresponding to $\widehat{\chi}_{n,k} \approx 0.68$, which is equal to the model-based $\chi$. The dotted lines correspond to $95 \%$ pointwise bootstrap confidence intervals.

\begin{figure}[ht]
\centering 
\subfloat{\includegraphics[width=0.33\textwidth]{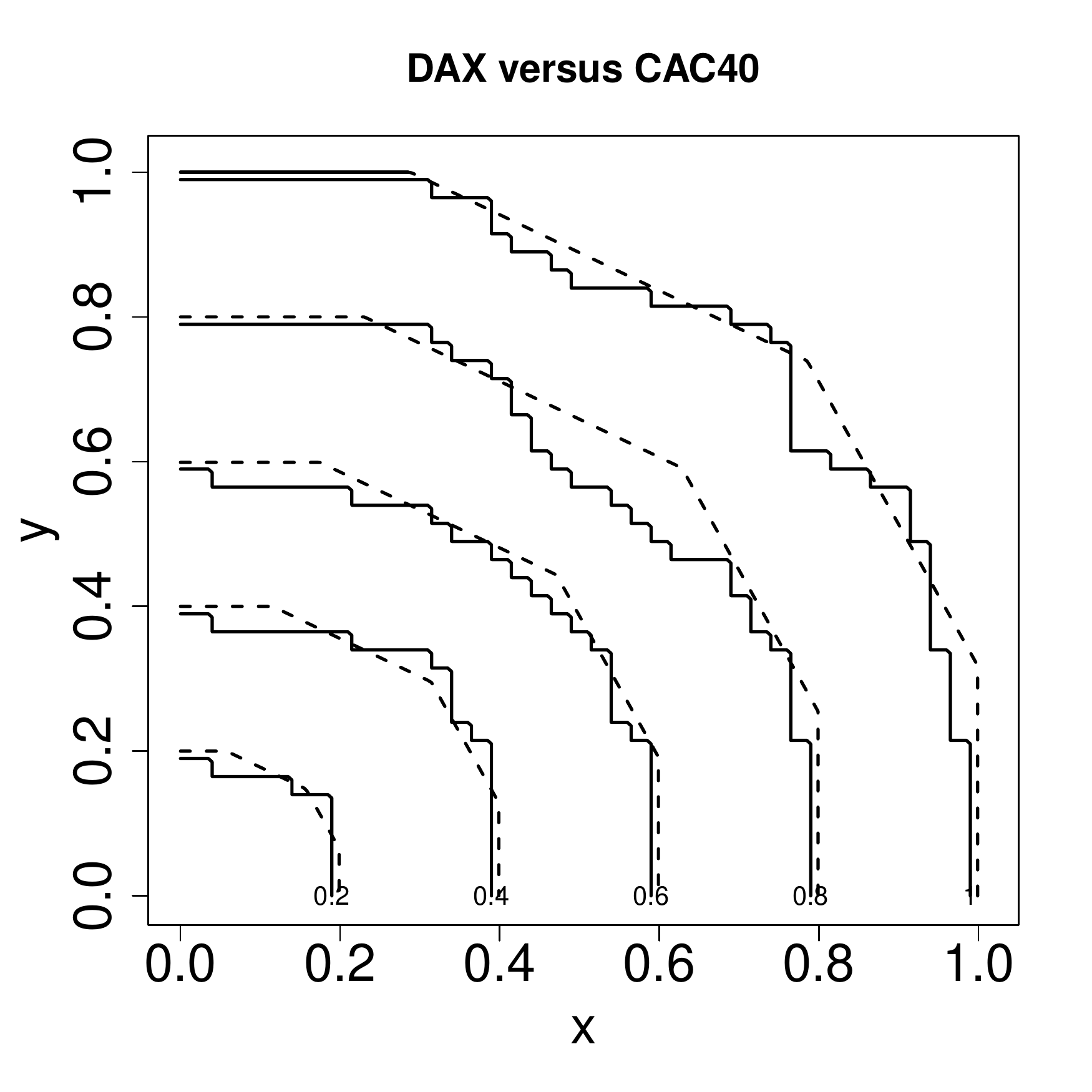}} 
\subfloat{\includegraphics[width=0.33\textwidth]{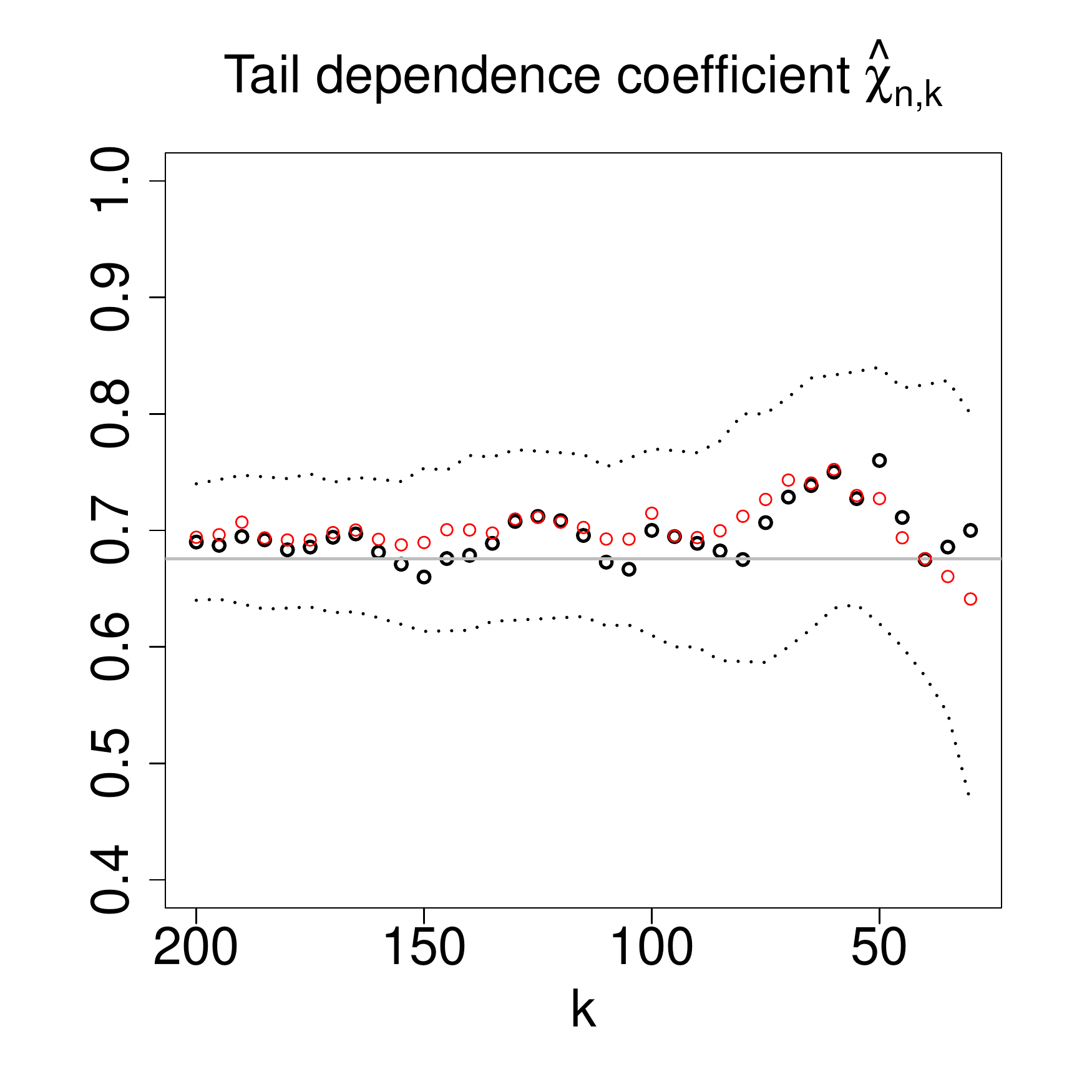}} 
\subfloat{\includegraphics[width=0.33\textwidth]{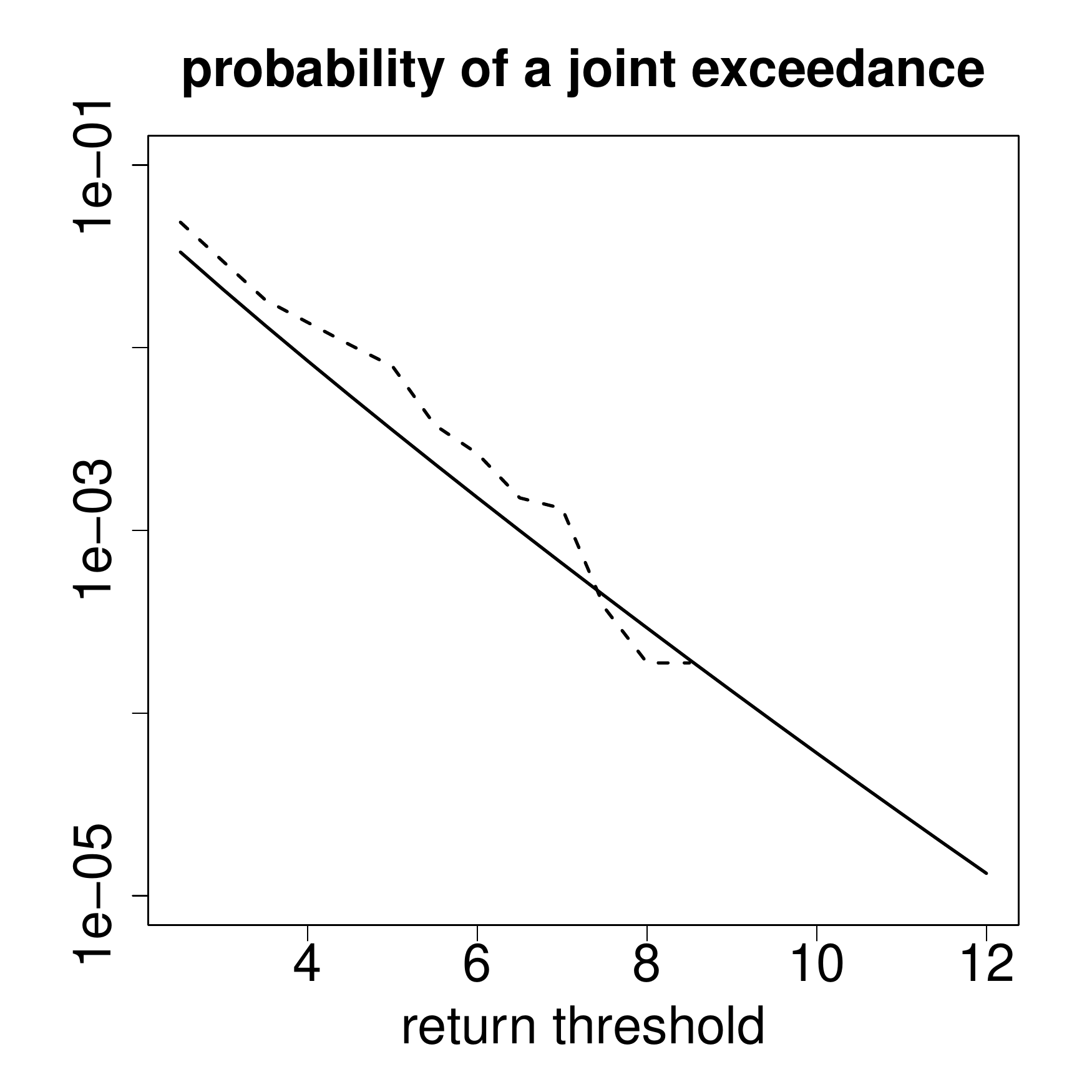}} 
\caption{Left: level sets $\{(x_1,x_2): \ell (x_1,x_2) = c \}$ for $c \in \{0.2,0.4,0.8,1\}$ based on the empirical tail dependence function (solid lines) and on the fitted max-linear model (dashed lines); middle: nonparametric (black dots) and model-based (red dots) estimates of the tail dependence coefficient $\widehat{\chi}_{n,k}$ for $k \in \{200, 195, \ldots, 30\}$; right: probability of a joint exceedance based on the fitted max-linear model (solid line) and on the empirical excesses (dashed line).} 
\label{fig:stocks2}
\end{figure}

The probability of a joint excess $\overline{F}(x_1,x_2) = \mathbb{P}[X_{11} > x_1, X_{12} > x_2]$ can be approximated for large $(x_1,x_2)$ by
\begin{equation*}
\overline{F}(x_1,x_2) = x_1^* + x_2^* - 1 + \exp \left\{ - \ell(x_1^*,x_2^*) \right\} \quad \text{for } x_j^* = 1 - F_j(x_j), \,\,\, j = 1,2.
\end{equation*}
Let $\overline{H} (\, \cdot \,; \gamma_j, \sigma_j)$ denote the survival function of a GPD with shape $\gamma_j$ and scale $\sigma_j$. 
For $x_1 > u_1$ and $x_2 > u_2$, we can estimate $\overline{F}(x_1,x_2)$ by 
\begin{equation*}\label{proba}
\widehat{\overline{F}}(x_1,x_2) = \widehat{x}_1^* + \widehat{x}_2^* - 1 + \exp \left\{ - \ell(\widehat{x}_1^*,\widehat{x}_2^* ; \widehat{\vc{\theta}}_{n,k}) \right\},
\end{equation*}
where $\ell ( \, \cdot \, ;  \widehat{\vc{\theta}}_{n,k})$ denotes the max-linear stable tail dependence function based on the fitted parameter estimates and
$\widehat{x}_j^* = 0.05 \times \overline{H} (x_j - u_j ; \widehat{\gamma}_j, \widehat{\sigma}_j)$.
The right-hand plot of Figure \ref{fig:stocks} shows the probabilities $\widehat{\overline{F}}(x,x)$ on the log-scale for $x \in \{2.5,3,\ldots,12\}$ (solid line). 
These can be compared to the empirical excess probabilities $n^{-1} \sum_{i = 1}^n \mathbbm{1} \left\{ X_{1i} > x, X_{2i} > x \right\}$ (dashed line), which are zero for $x > 8.5$.

\section{Discussion}\label{sec:discussion}
We proposed two test statistics for parameters that, under the null hypothesis, are on the boundary of the alternative hypothesis. Simulation studies showed that especially the Wald-type test statistic performs well. Moreover, these test statistics are convenient because their asymptotic distribution has an explicit expression. In practice, this distribution becomes cumbersome when testing for a higher dimensional ($c > 4$) parameter vector. A possible solution might be to consider a multiple testing procedure with a Bonferroni-type of correction.

We applied the test statistics to the max-linear model and the Brown--Resnick model, but the results are generic and could be used on any multivariate extreme-value model where submodels appear at boundary values or where the number of ``factors'' is of interest. Examples include \cite{fougeres2009}, where a mixture model is obtained based on a certain number of stable random variables, and \cite{gissibl2015}, where a max-linear model is defined on a directed acyclic graph. For the latter, our methods could be useful in the procedure to reconstruct the graph structure of a dataset \citep{gissibl2017}, since it is based on identifying pairwise tail dependence coefficients that are zero.

\appendix
\section{Proofs}
%Let ``for all $\gamma_k \rightarrow 0$" be an abbreviation of for all sequences of positive scalar constants $\{\gamma_k : k \geq 1\}$ for which $\gamma_k \rightarrow 0$ as $k \rightarrow \infty$. 
\begin{defn*}[\emph{Left/right partial derivatives}]
Let $f$ be a function whose support includes $\mathcal{X} \subset \RR^p$ and let $\vc{a} \in \mathcal{X}$. Suppose that $\mathcal{X} - \vc{a} $ equals the intersection of a union of orthants and an open cube $C_{\varepsilon} (\vc{0})$ for some $\varepsilon > 0$, i.e., $\mathcal{X} - \vc{a}$ is locally equal to a union of orthants. The function $f$ is said to have \emph{left/right (l/r) partial derivatives} of order 1 on $\mathcal{X}$ if
\begin{enumerate}
\item it has partial derivatives at each interior point of $\mathcal{X}$;
\item it has partial derivatives at each boundary point of $\mathcal{X}$ with respect to coordinates that can be perturbed to the left and right;
\item it has left (right) partial derivatives at each boundary point of $\mathcal{X}$ with respect to coordinates that can be perturbed only to the left (right). 
\end{enumerate}
\end{defn*}
The shape of $\mathcal{X}$ is such that for all $\vc{x} \in \mathcal{X}$ and for all $i \in \{1,\ldots,p\}$, it is possible to perturb $x_i$ to the left, the right, or both and stay within $\mathcal{X}$. This means that it is always possible to define the left, right, or two-sided partial derivative of $f$ with respect to $x_i$. 
We say that $f$ has l/r partial derivatives of order $k$ on $\mathcal{X}$ for $k \geq 2$ if $f$ has l/r partial derivatives of order $k-1$ on $\mathcal{X}$ and each of the latter has l/r partial derivatives on $\mathcal{X}$. When we say the $f$ has continuous l/r partial derivatives, continuity is defined in terms of local perturbations within $\mathcal{X}$ only. 
%Theorem 6 in \citet{andrews1999} gives a Taylor expansion for a function in terms of its l/r partial derivatives.

\begin{proof}[Proof of Theorem~\ref{theorem2}]
This theorem is a special case of \citet[Theorem 3b]{andrews1999}. A closely related work is \citet{andrews2002}, where the focus is on generalized method of moment estimators. This setting is closer to ours but uses a convergence rate of $\sqrt{n}$, whereas \citet{andrews1999} allows for more generality. The quantities $\ell_T$, $B_T$ and $R_T$ in \citet{andrews1999} correspond to $-(k/2) f_{n,k}$, $\sqrt{k} I_p$ and $-(k/2) R_{n,k}$ respectively in this paper.
Theorem 3b in \citet{andrews1999} holds under the assumptions 1, $2^*$, $3^*$, $5^*$, and 6 of that paper. We show that these are implied by ours:
\begin{description}
\item[Assumption 1] The assumptions made in Theorem~\ref{theorem1} imply assumption 1.
\item[Assumption $\mathbf{2^*}$] Assumption GMM2 in \citet{andrews2002} implies Assumption $2^*$; this is proven in \citet[Lemma 3]{andrews2002}. We show that our assumption imply GMM2:
assumption GMM2(a) holds because because $D_{n,k} (\vtheta)$ converges in probability to $L(\vtheta_0) - L(\vtheta)$. Assumption GMM2(b) holds if assumption GMM$2^{2*}(b)$ holds, which in turn is implied by the assumptions of Theorem~\ref{theorem1} and (A1). Assumption GMM2(c) holds since $\vc{D}(\vtheta_0) = \vc{0}$. Assumption GMM2(d) holds since $D_{n,k} (\vtheta) - \vc{D}(\vtheta) - D_{n,k}(\vtheta_0) = 0$. Finally, Assumption GMM2(e) holds automatically since our weight matrix isn't random.
\item[Assumption $\mathbf{3^*}$] Assumption $3^*$ holds because of (A3) and because $J$ is non-random, symmetric and non-singular.
\item[Assumption $\mathbf{5^*}$] Assumption $5^*$ holds because of (A2) and because $B_T = \sqrt{k} I_p$ and $\sqrt{k} \rightarrow \infty$ as $n \rightarrow \infty$. 
\item[Assumption 6] Assumption (A2) implies assumption 6.
%\begin{equation*}
%\sup_{\vtheta \in \Theta: \norm{\vtheta - \vtheta_0} \leq \gamma_n} \sqrt{n} \frac{\norm{D_{n,k} (\vtheta) - D(\vtheta) - D_{n,k}(\vtheta_0)}}{1 + \sqrt{n} \norm{\vtheta - \vtheta_0}} = o_p (1).
%\end{equation*}
\end{description}
\end{proof}

\begin{proof}[Proof of Corollary~\ref{cor1}]
This corollary is a special case of \citet[Corollary 1b]{andrews1999}, where no parameter $\psi$ appears. Assumptions 1, $2^*$, $3^*$,$5^*$ and 6--8 in that paper are needed: in the proof of Theorem~\ref{theorem2}, we've already shown that our assumptions imply 1, $2^*$, $3^*$,$5^*$ and 6; assumptions 7 and 8 hold by our assumption (A4).
\end{proof}

\section*{Acknowledgements}

The author would like to thank two reviewers and an associate editor for their careful reading of the paper and their constructive comments that greatly improved the generality of the paper. She would also like to thank Johan Segers for helpful comments on an earlier version of this paper.

\small
\renewcommand\refname{REFERENCES} 
\bibliographystyle{apalike} 
\bibliography{library}

\end{document}